\documentclass[12pt,reqno]{amsart}
\usepackage{amsmath}
\usepackage{amssymb}
\usepackage{verbatim}
\usepackage{mathrsfs}
\usepackage{graphicx} 
\usepackage[usenames,dvipsnames]{pstricks}
\usepackage{pst-plot,pstricks-add}
\usepackage{wavelet}
\usepackage{subcaption}

\usepackage[top=0.6in,bottom=0.5in,left=0.7in,right=0.7in]{geometry}

\newtheorem{lemma}{Lemma}
\newtheorem{theorem}[lemma]{Theorem}
\newtheorem{cor}[lemma]{Corollary}
\newtheorem{example}{Example}

\numberwithin{equation}{section}

\begin{document}

\title[Compactly Supported Multivariate Quasi-tight Framelets]{Quasi-tight Framelets
with Directionality or High Vanishing Moments Derived from Arbitrary Refinable Functions }

\author{Chenzhe Diao and Bin Han}

\address{Department of Mathematical and Statistical Sciences,
University of Alberta, Edmonton,\quad Alberta, Canada T6G 2G1.
\quad {\tt diao@ualberta.ca},\quad {\tt bhan@ualberta.ca}
}

\thanks{
Research was supported in part by
Natural Sciences and Engineering Research Council of Canada (NSERC Canada) under Grant RGP 228051.
}

\makeatletter \@addtoreset{equation}{section} \makeatother
\begin{abstract}
Construction of multivariate tight framelets is known to be a challenging problem because it is linked to the difficult problem on sum of squares of multivariate polynomials in real algebraic geometry.
Multivariate dual framelets with vanishing moments generalize tight framelets and are not easy to be constructed either, since their construction is related to
syzygy modules and factorization of multivariate polynomials. On the other hand, compactly supported multivariate framelets with directionality or high vanishing moments are of interest and importance in both theory and applications. In this paper we introduce the notion of a quasi-tight framelet, which is a dual framelet, but behaves almost like a tight framelet. Let $\phi\in \dLp{2}$ be an arbitrary compactly supported $\dm$-refinable function with a general dilation matrix $\dm$ and $\wh{\phi}(0)=1$ such that its underlying low-pass filter satisfies the basic sum rule.
We first constructively prove by a step-by-step algorithm that we can always easily derive from the arbitrary $\dm$-refinable function $\phi$ a directional compactly supported quasi-tight $\dm$-framelet in $\dLp{2}$ associated with a directional quasi-tight $\dm$-framelet filter bank, each of whose high-pass filters has only two nonzero coefficients with opposite signs. If in addition all the coefficients of its low-pass filter are nonnegative, such a quasi-tight $\dm$-framelet becomes a directional tight $\dm$-framelet in $\dLp{2}$.
Furthermore, we show by a constructive algorithm that we can always derive from the arbitrary $\dm$-refinable function $\phi$ a compactly supported quasi-tight $\dm$-framelet in $\dLp{2}$ with the highest possible order of vanishing moments. We shall also present a result on quasi-tight framelets whose associated high-pass filters are purely differencing filters with the highest order of vanishing moments.
Several examples will be provided to illustrate our main theoretical results and algorithms in this paper.
\end{abstract}


 \maketitle

\pagenumbering{arabic}


\section{Introduction and Motivations}

Separable multivariate wavelets and framelets with a diagonal dilation matrix can be trivially constructed through tensor product from one-dimensional wavelets and framelets. However, such separable wavelets and framelets are known to give preference to the axis coordinate directions (e.g., see \cite{kv92,sbk05}) and they are only a very special family of multivariate wavelets and framelets. It is important but often much challenging to study nonseparable/general multivariate wavelets and framelets in both theory and applications. Currently, there is a growing interest in wavelet analysis on studying and constructing (nonseparable) multivariate wavelets and framelets. There exist a huge amount of literature on wavelets, framelets and their many impressive applications, to only mention a tiny portion of them (in particular, multivariate framelets that are closely related to this paper), e.g., see
\cite{cpss13,cpss15,ch01,ehl07,eh08,fjs16,han03,han14,hjsz18,hlz17,hz14,hz15,hr08,js15,kps16book,ls06,rs98,sz16,sko09} and many references therein.
However, construction of multivariate wavelets and framelets are widely known as a challenging problem in the literature. In this paper, we mainly concentrate on multivariate framelets derived from an arbitrarily given refinable function. To explain our motivations of this paper, let us first recall some basic definitions and concepts.

For a function $f: \dR\rightarrow \C$ and a $d\times d$ real-valued matrix $U$, throughout this paper we shall adopt the following notation:
\be \label{f:ds}
f_{U;k}(x):=|\det(U)|^{1/2} f(Ux-k),\qquad x,k\in \dR.
\ee
A $d\times d$ matrix $\dm$ is called \emph{a dilation matrix} if it is an integer matrix such that all its eigenvalues are greater than one in modulus.
For $\phi,\psi^1,\ldots,\psi^s\in \dLp{2}$, we say that $\{\phi;\psi^1,\ldots,\psi^s\}$ is \emph{an $\dm$-framelet} in $\dLp{2}$ if there exist positive constants $C_1$ and $C_2$ such that
\be \label{framelet}
C_1 \|f\|_{\dLp{2}}^2\le
\sum_{k\in \dZ} |\la f, \phi(\cdot-k)\ra|^2+
\sum_{j=0}^\infty \sum_{\ell=1}^s \sum_{k\in \dZ} |\la f, \psi^\ell_{\dm^j;k}\ra|^2\le C_2 \|f\|_{\dLp{2}}^2,\qquad
\forall\, f\in \dLp{2}.
\ee
If \eqref{framelet} holds with $C_1=C_2=1$, then
$\{\phi;\psi^1,\ldots,\psi^s\}$ is called \emph{a (normalized) tight $\dm$-framelet} in $\dLp{2}$.
Let $\tilde{\phi},\tilde{\psi}^1,\ldots,\tilde{\psi}^s\in \dLp{2}$. We say that $(\{\tilde{\phi};\tilde{\psi}^1,\ldots,\tilde{\psi}^s\},\{\phi;\psi^1,\ldots,\psi^s\})$
is \emph{a dual $\dm$-framelet} in $\dLp{2}$ if both $\{\tilde{\phi};\tilde{\psi}^1,\ldots,\tilde{\psi}^s\}$ and $\{\phi;\psi^1,\ldots,\psi^s\}$ are $\dm$-framelets in $\dLp{2}$ such that
\be \label{df}
\la f,g\ra=\sum_{k\in \dZ} \la f, \tilde{\phi}(\cdot-k)\ra\la \phi(\cdot-k),g\ra+
\sum_{j=0}^\infty \sum_{\ell=1}^s \sum_{k\in \dZ}
\la f, \tilde{\psi}^\ell_{\dm^j;k}\ra \la \psi^\ell_{\dm^j;k},g\ra,\qquad \forall\, f,g\in \dLp{2}
\ee
with the above series converging absolutely. It follows directly from \eqref{df} that every function $f\in \dLp{2}$ has the following multiscale framelet representation:
\be \label{df:repr}
f=\sum_{k\in \dZ} \la f, \tilde{\phi}(\cdot-k)\ra \phi(\cdot-k)+
\sum_{j=0}^\infty \sum_{\ell=1}^s \sum_{k\in \dZ}
\la f, \tilde{\psi}^\ell_{\dm^j;k}\ra \psi^\ell_{\dm^j;k}
\ee
with the series converging unconditionally in $\dLp{2}$.

For an integrable function $f\in \dLp{1}$, its Fourier transform $\wh{f}$ in this paper is defined to be $\wh{f}(\xi):=\int_{\dR} f(x) e^{-ix\cdot \xi} dx$ for $\xi\in \dR$.
The Fourier transform can be naturally extended to square integrable functions in $\dLp{2}$ and tempered distributions.
A dual $\dm$-framelet $(\{\tilde{\phi};\tilde{\psi}^1,\ldots,\tilde{\psi}^s\},\{\phi;\psi^1,\ldots,\psi^s\})$ in $\dLp{2}$ is often constructed from some special $\dm$-refinable functions $\phi$ and $\tilde{\phi}$ through the refinable structure
\be \label{refstr}
\wh{\phi}(\dm^\tp \xi)=\wh{a}(\xi)\wh{\phi}(\xi),\quad
\wh{\psi^\ell}(\dm^\tp\xi)=\wh{b_\ell}(\xi)\wh{\phi}(\xi),\qquad a.e.\, \xi\in \dR, \ell=1,\ldots,s
\ee
and
\be \label{refstr:tilde}
\wh{\tilde{\phi}}(\dm^\tp \xi)=\wh{\tilde{a}}(\xi)\wh{\tilde{\phi}}(\xi),\quad
\wh{\tilde{\psi}^\ell}(\dm^\tp\xi)=\wh{\tilde{b}_\ell}(\xi)
\wh{\tilde{\phi}}(\xi),\qquad a.e.\, \xi\in \dR, \ell=1,\ldots,s
\ee
for some $2\pi\dZ$-periodic measurable functions $\wh{a},\wh{b_1},\ldots,\wh{b_s},\wh{\tilde{a}},
\wh{\tilde{b}_1},\ldots, \wh{\tilde{b}_s}$ on $\dR$.
A function $\phi$ satisfying the first identity in \eqref{refstr}, i.e., $\wh{\phi}(\dm^\tp\xi)=\wh{a}(\xi)\wh{\phi}(\xi)$, is called \emph{an $\dm$-refinable function} with the (refinement) filter/mask $a$.
In fact, as shown in \cite[Theorem~4.5.4]{hanbook}, all tight framelets $\{\phi;\psi^1,\ldots,\psi^s\}$ must be derived from a refinable function $\phi$ through the refinable structure in \eqref{refstr}.

By $\dlp{0}$ we denote the set of all finitely supported sequences/filters $a=\{a(k)\}_{k\in \dZ}: \dZ\rightarrow \C$ on $\dZ$. For a filter $a\in \dlp{0}$,
its support is $\supp(a):=\{k\in \dZ \setsp a(k)\ne 0\}$ and
its Fourier series (or symbol) is defined to be $\wh{a}(\xi):=\sum_{k\in \dZ} a(k)e^{-ik\cdot\xi}$ for $\xi\in \dR$, which is a $2\pi\dZ$-periodic trigonometric polynomial in $d$ variables. In particular, by $\td$ we denote \emph{the Dirac sequence} such that $\td(0)=1$ and $\td(k)=0$ for all $\dZ\bs\{0\}$. For $\gamma\in \dZ$, we also use the notation $\td_\gamma$ to stand for the sequence $\td(\cdot-\gamma)$, i.e., $\td_\gamma(\gamma)=1$ and $\td_\gamma(k)=0$ for all $k\in \dZ\bs \{\gamma\}$.
Note that $\wh{\td_\gamma}(\xi)=e^{-i\gamma\cdot \xi}$.

Let $a,b_1,\ldots,b_s,\tilde{a},\tilde{b}_1,\ldots,\tilde{b}_s\in \dlp{0}$ such that $\wh{a}(0)=\wh{\tilde{a}}(0)=1$. Define
\be \label{reffunc}
\wh{\phi}(\xi):=\prod_{j=1}^\infty \wh{a}((\dm^\tp)^{-j}\xi),\qquad
\xi\in \dR
\ee
and
similarly define $\wh{\tilde\phi}(\xi):=\prod_{j=1}^\infty \wh{\tilde{a}}((\dm^\tp)^{-j}\xi)$ for $\xi\in \dR$.
Then it is trivial to observe that the first identity in \eqref{refstr} is satisfied and both $\phi$ and $\tilde{\phi}$ are $\dm$-refinable functions/distributions. Define $\psi^1,\ldots,\psi^s$ as in \eqref{refstr} and $\tilde{\psi}^1,\ldots,\tilde{\psi}^s$ as in \eqref{refstr:tilde}. It is known (e.g., see
\cite[Theorems~7.1.6 and 7.1.7]{hanbook}, \cite[Theorem~2.3]{han03}, and \cite{han97,rs97}) that
$(\{\tilde{\phi};\tilde{\psi}^1,\ldots,\tilde{\psi}^s\},\{\phi;\psi^1,\ldots,\psi^s\})$
is a dual $\dm$-framelet in $\dLp{2}$ if and only if $\phi,\tilde{\phi}\in \dLp{2}$,
\be \label{vm:b}
\wh{b_1}(0)=\cdots=\wh{b_s}(0)=\wh{\tilde{b}_1}(0)=\cdots=\wh{\tilde{b}_s}(0)=0,
\ee
and $(\{\tilde{a};\tilde{b}_1,\ldots,\tilde{b}_s\},\{a;b_1,\ldots,b_s\})$ is \emph{a dual $\dm$-framelet filter bank}, i.e.,
\be \label{dffb:0}
\wh{\tilde{a}}(\xi)\ol{\wh{a}(\xi+2\pi \omega)}+
\wh{\tilde{b}_1}(\xi)\ol{\wh{b}_1(\xi+2\pi \omega)}
+\cdots+\wh{\tilde{b}_s}(\xi)\ol{\wh{b}_s(\xi+2\pi \omega)}=\td(\omega),\qquad \forall\, \omega\in \Omega_\dm, \xi\in \dR,
\ee
where $\Omega_\dm:=[(\dm^\tp)^{-1}\dZ]\cap [0,1)^d$.
A filter bank $\{a;b_1,\ldots,b_s\}$ is called \emph{a tight $\dm$-framelet filter bank} if $(\{a;b_1,\ldots,b_s\},\{a;b_1,\ldots,b_s\})$ is a dual $\dm$-framelet filter bank.
It is known in \cite[Theorem~7.1.8]{hanbook} that $\{\phi;\psi^1,\ldots,\psi^s\}$ is a tight $\dm$-framelet in $\dLp{2}$ if and only if $\{a;b_1,\ldots,b_s\}$ is a tight $\dm$-framelet filter bank.  Hence one does not need to check the necessary conditions $\phi\in \dLp{2}$ and $\wh{b_1}(0)=\cdots=\wh{b_s}(0)=0$ in advance for tight framelets.
Consequently, the construction of dual/tight framelets boils down to the construction of dual/tight framelet filter banks.

Define $\NN:=\N\cup\{0\}$.
For $\mu=(\mu_1,\ldots,\mu_d)^\tp\in \dNN$ and $x=(x_1,\ldots,x_d)^\tp \in \dR$, we define $|\mu|:=\mu_1+\cdots+\mu_d$,
$x^\mu:=x_1^{\mu_1}\cdots x_d^{\mu_d}$ and partial derivative operators $\partial^\mu:=\partial_1^{\mu_1}\cdots\partial_d^{\mu_d}$.
For $m\in \NN$ and smooth functions $f$ and $g$ on $\dR$, we shall simply use the big $\bo$ notation $f(\xi)=g(\xi)+\bo(\|\xi\|^m)$ as $\xi \to 0$ to stand for $\partial^\mu f(0)=\partial^\mu g(0)$ for all $\mu\in \dNN$ with $|\mu|<m$.
For a compactly supported function $\psi\in \dLp{2}$, we say that $\psi$ has \emph{order $m$ vanishing moments} if
\be \label{vm}
\int_{\dR} \psi(x) x^\mu dx=0,\quad \forall\, \mu\in \dNN, |\mu|<m
\quad \mbox{or equivalently},\quad
\wh{\psi}(\xi)=\bo(\|\xi\|^m),\quad \xi \to 0.
\ee
In particular, we define $\vmo(\psi):=m$ for the largest possible integer $m$ in \eqref{vm}. For a filter $b\in \dlp{0}$, we similarly define $\vmo(b)$ to be the largest possible integer $n\in \NN$ such that $\wh{b}(\xi)=\bo(\|\xi\|^n)$ as $\xi \to 0$. That is, the filter $b$ has the order $\vmo(b)$ vanishing moments.
If $\psi$ is derived from a function $\phi$ with $\wh{\phi}(0)\ne 0$ through $\wh{\psi}(\dm^\tp \xi)=\wh{b}(\xi)\wh{\phi}(\xi)$ for some $b\in \dlp{0}$, then it is trivial to see that $\vmo(\psi)=\vmo(b)$.
Note that the condition in \eqref{vm:b} is equivalent to saying that all the filters $b_1,\ldots,b_s,\tilde{b}_1,\ldots,\tilde{b}_s$ have order one vanishing moment (i.e., the basic vanishing moment).
The notion of vanishing moments plays the key role for the sparsity of a framelet representation in \eqref{df:repr}. Therefore, vanishing moments are one of the most desirable properties of
wavelets and framelets (see \cite{daubook}).
To effectively capture singularities such as edges in images in high-dimensional data and functions, another key desirable feature of multivariate framelets is directionality. Framelets with directionality are crucial in many applications such as image/video processing (see \cite{cd04,gkl06,hz14,hz15,shearlet,lcsht16,sbk05} and many references therein).

In the following, let us discuss the difficulties and challenges in constructing multivariate tight framelets and multivariate dual framelets with vanishing moments. To do so, let us first rewrite the equations in \eqref{dffb:0} for a dual $\dm$-framelet filter bank $(\{\tilde{a};\tilde{b}_1,\ldots,\tilde{b}_s\},\{a;b_1,\ldots,b_s\})$. For a matrix or a matrix function $A(\xi)$, we define $A^\star(\xi):=\ol{A(\xi)}^\tp$, the transpose of the complex conjugate of $A(\xi)$. Define
\be \label{ddm}
\ddm:=|\det(\dm)| \quad \mbox{and}\quad \{\omega_1,\ldots,\omega_{\ddm}\}:=\Omega_\dm:=
[(\dm^\tp)^{-1}\dZ]\cap [0,1)^d.
\ee
Define $b_0:=a$ and $\tilde{b}_0:=\tilde{a}$.
Then it is not difficult to observe that  \eqref{dffb:0} is equivalent to
\be \label{dffb1}
\sum_{\ell=0}^s \Big[\wh{\tilde{b}_\ell}(\xi+2\pi \omega_1),\ldots, \wh{\tilde{b}_\ell}(\xi+2\pi \omega_{\ddm})\Big]^\star
\Big[\wh{b_\ell}(\xi+2\pi \omega_1),\ldots, \wh{b_\ell}(\xi+2\pi \omega_{\ddm})\Big]
=I_{\ddm},
\ee
where $I_{\ddm}$ is the $\ddm\times \ddm$ identity matrix.
For $u\in \dlp{0}$ and $\gamma\in \dZ$, \emph{its $\gamma$-coset sequence $u^{[\gamma]}$} is defined to be $u^{[\gamma]}:=\{u(\gamma+\dm k)\}_{k\in \dZ}$. Define
\be \label{Gamma}
\{\gamma_1,\ldots,\gamma_{\ddm}\}:=\Gamma_{\dm}:=
[\dm[0,1)^d]\cap \dZ.
\ee
Then
$\wh{u}(\xi)=\sum_{\gamma\in \Gamma_\dm} e^{-i\gamma \cdot \xi}\wh{u^{[\gamma]}}(\dm^\tp \xi)$ and we have
\be \label{coset:U}
[\wh{u}(\xi+2\pi \omega_1),\ldots,\wh{u}(\xi+2\pi \omega_{\ddm})]=\Big[\wh{u^{[\gamma_1]}}(\dm^\tp\xi),
\ldots,\wh{u^{[\gamma_{\ddm}]}}(\dm^\tp \xi)\Big] E(\xi) U,
\ee
where
\be \label{EU}
E(\xi):=\diag(e^{-i\gamma_1 \cdot \xi},\ldots,e^{-i\gamma_{\ddm} \cdot \xi})\quad \mbox{and}\quad
U:=( e^{-i\gamma_j\cdot2\pi \omega_k})_{1\le j,k\le \ddm}.
\ee
Because $E(\xi) E^\star(\xi)=I_{\ddm}$ and $U U^\star=\ddm I_{\ddm}$ for all $\xi\in \dR$, it is now straightforward to
deduce that \eqref{dffb1} is equivalent to
\be \label{dffb}
\sum_{\ell=1}^s \Big[\wh{\tilde{b}^{[\gamma_1]}_\ell}(\xi),\ldots, \wh{\tilde{b}^{[\gamma_{\ddm}]}_\ell}(\xi)\Big]^\star
\Big[\wh{b^{[\gamma_1]}_\ell}(\xi),\ldots, \wh{b^{[\gamma_{\ddm}]}_\ell}(\xi)\Big]
=\ddm^{-1}\cN_{\tilde{a},a}(\xi)
\ee
with
\be \label{Nata}
\cN_{\tilde{a},a}(\xi):= I_{\ddm}-\ddm\Big[\wh{\tilde{a}^{[\gamma_1]}}(\xi),\ldots, \wh{\tilde{a}^{[\gamma_{\ddm}]}}(\xi)\Big]^\star
\Big[\wh{a^{[\gamma_1]}}(\xi),\ldots, \wh{a^{[\gamma_{\ddm}]}}(\xi)\Big].
\ee
In particular, a filter bank $\{a;b_1,\ldots,b_s\}$ is a tight $\dm$-framelet filter bank if and only if
\be \label{tffb}
\sum_{\ell=1}^s \Big[\wh{b^{[\gamma_1]}_\ell}(\xi),\ldots, \wh{b^{[\gamma_{\ddm}]}_\ell}(\xi)\Big]^\star
\Big[\wh{b^{[\gamma_1]}_\ell}(\xi),\ldots, \wh{b^{[\gamma_{\ddm}]}_\ell}(\xi)\Big]
=\ddm^{-1}\cN_{a}(\xi)\quad \mbox{with}\quad \cN_a(\xi):=\cN_{a,a}(\xi).
\ee
Using the simple fact that all the nonzero eigenvalues of $AB$ and $BA$ are the same, one concludes that $\cN_a$ has the eigenvalue
$1-\ddm \sum_{j=1}^{\ddm} |\wh{a^{[\gamma_j]}}(\xi)|^2$ and all its other eigenvalues are $1$. Hence, we have
$\det(\cN_a(\xi))=1-\ddm \sum_{j=1}^{\ddm} |\wh{a^{[\gamma_j]}}(\xi)|^2$. Taking the determinants on both sides of \eqref{tffb} and applying the Cauchy-Binet formula to the left-hand side of \eqref{tffb}, Charina et al. in \cite{cpss13} observed that
\be \label{sos}
\cA(\xi):=\det(\cN_a(\xi))=
1-\ddm \sum_{j=1}^{\ddm} |\wh{a^{[\gamma_j]}}(\xi)|^2=
\sum_{\ell=1}^{s_a} |\wh{u_\ell}(\xi)|^2,\qquad \forall\, \xi\in \dR,
\ee
for some finitely supported sequences $u_1,\ldots, u_{s_a}\in \dlp{0}$ with $s_a=\binom{s}{\ddm}$.
That is, the nonnegative $2\pi\dZ$-periodic trigonometric polynomial $\cA$ in \eqref{sos} can be written as a sum of Hermitian squares of $2\pi\dZ$-periodic trigonometric polynomials.
By \eqref{coset:U}, one can easily observe that $\cA(\dm^\tp \xi)=1-\sum_{\omega\in \Omega_\dm}|\wh{a}(\xi+2\pi \omega)|^2$.
Conversely, suppose that there exist $u_1, \ldots, u_{s_a}\in \dlp{0}$ for some integer $s_a$ such that $\cA$ can be written as a sum of Hermitian squares as in \eqref{sos}. Then it is known in Lai and St\"ockler \cite{ls06} that one can always construct a tight $\dm$-framelet filter bank $\{a;b_1,\ldots,b_s\}$ with $s:=s_a+\ddm$.
In dimension one, due to the Fej\'er-Riesz lemma, for a nonnegative $2\pi$-periodic trigonometric polynomial $\cA$ in \eqref{sos}, there always exists $u_1\in \lp{0}$ such that \eqref{sos} is satisfied with $s_a=1$. However,
as discussed in \cite{cpss13}, the problem on sums of Hermitian squares in \eqref{sos} is much more complicated in dimension higher than one and is known to be a challenging problem in real algebraic geometry. See \cite{cpss13,cpss15,ls06} for a detailed discussion on \eqref{sos} and its applications to the construction of multivariate tight framelet filter banks.
To the best of our knowledge, there is currently no known algorithm to efficiently compute $u_1,\ldots, u_{s_a}\in \dlp{0}$ in \eqref{sos}.
So far, all current known constructions of multivariate nonseparable tight framelets are developed for special low-pass filters,
for example, see \cite{cpss13,cpss15,ch01,fjs16,han03,han14,hjsz18,kps16book,ls06,rs98,sz16,sko09} and many references therein.
Despite recent progresses and enormous effort, construction of multivariate tight framelets still remains as a challenging problem and most constructed multivariate tight framelets in the literature lack some desirable properties such as directionality, vanishing moments and/or symmetry.

Though dual framelet filter banks offer flexibility over tight framelet filter banks, it is not easy to construct them either. To construct a dual $\dm$-framelet filter bank, one has to factorize the matrix $\cN_{\tilde{a},a}$ in \eqref{dffb} so that all the high-pass filters satisfy the basic vanishing moment in \eqref{vm:b}.
For a one-dimensional $2\pi$-periodic trigonometric polynomial $\cA(\xi)$, due to the fundamental theorem of algebra, if $\cA(\xi)$ has a zero with multiplicity $m\in \N$ at $\xi=0$, then one can always write
$\cA(\xi)=(1-e^{-i\xi})^m\cB(\xi)$ for some $2\pi$-periodic trigonometric polynomial $\cB$ with $\cB(0)\ne 0$. This factorization technique for separating out the special factor $(1-e^{-i\xi})^m$ is the key for constructing one-dimensional dual framelet filter banks with high vanishing moments (e.g., see \cite{dh04,han15} for details).
However, such a factorization technique is not available for dimensions higher than one; there are also no special multivariate trigonometric polynomials playing the role of $(1-e^{-i\xi})^m$ for us to generalize the construction of one-dimensional dual framelet filter banks to multiple dimensions.
In fact, most (or generic) multivariate trigonometric polynomials cannot be factorized into products of two nontrivial trigonometric polynomials for dimensions higher than one.
Consequently, from any given pair of low-pass filters, it is often difficult to construct a dual $\dm$-framelet filter bank satisfying the basic vanishing moment condition in \eqref{vm:b}.
For example, the constructions in \cite{ehl07,eh08} are linked to syzygy modules of multivariate Laurent polynomials in algebra and the constructed high-pass filters also have much larger supports than that of their associated low-pass filters.

To reduce the above mentioned difficulties for multivariate tight or dual framelets,
now let us introduce the notion of quasi-tight framelets and explain our motivations. For a given dual $\dm$-framelet $(\{\tilde{\phi};\tilde{\psi}^1,\ldots,\tilde{\psi}^s\},
\{\phi;\psi^1,\ldots,\psi^s\})$ in $\dLp{2}$, every function $f\in \dLp{2}$ has the representation in \eqref{df:repr}.
Each $\tilde{\psi}^\ell$ models some desired feature capturing certain key singularities such as edges or textures in images. For example, $\tilde{\psi}^\ell$ may behave like an edge or texture. If $f$ contains such (scaled and shifted) feature $\tilde{\psi}^\ell_{\dm^j;k}$ for some $1\le \ell\le s, j\in \NN$ and $k\in \dZ$, then the coefficient $\la f, \tilde{\psi}^\ell_{\dm^j;k}\ra$ has a large significant magnitude. Therefore, we can capture such desired feature in $f$ by observing a significant coefficient $\la f, \tilde{\psi}^\ell_{\dm^j;k}\ra$. However, we reconstruct $f$ through \eqref{df:repr} by using $\la f, \tilde{\psi}^\ell_{\dm^j;k}\ra\psi^\ell_{\dm^j;k}$. If $\psi^\ell$ is very similar/close to $\tilde{\psi}^\ell$, then we indeed are able to capture the desired feature $\tilde{\psi}^\ell_{\dm^j;k}$ in $f$. However, if this is not the case, then the representation in \eqref{df:repr}, which can exactly reconstruct $f$, do not make much sense to extract and represent features in $f$. This is probably the main reason that only dual framelets $(\{\tilde{\phi};\tilde{\psi}^1,\ldots,\tilde{\psi}^s\},
\{\phi;\psi^1,\ldots,\psi^s\})$ in $\dLp{2}$, with $\tilde{\psi}^\ell$ being similar to $\psi^\ell$ for all $\ell=1,\ldots,s$, are useful in applications. The ideal case that $\tilde{\psi}^\ell=\psi^\ell$ for all $\ell=1,\ldots,s$ leads to tight framelets. However, as we discussed before, construction of multivariate tight framelets is a challenging problem, not even to mention that most known constructed multivariate tight framelets lack directionality, which is a very much desired property of framelets for many applications such as image processing (e.g., see \cite{cd04,gkl06,han12,hz14,hz15,shearlet,sbk05} and many references therein).

Instead of requiring $\tilde{\psi}^\ell=\psi^\ell$ as in a tight framelet, as observed in \cite[Example~3.2.2]{hanbook},
we can easily achieve our objective by naturally considering $\tilde{\psi}^\ell=\eps_\ell \psi^\ell$ with $\eps_\ell\in \{-1,1\}$. This motivates us to introduce the notion of quasi-tight framelets. For $\phi,\psi^1,\ldots,\psi^s\in \dLp{2}$ and $\eps_1,\ldots,\eps_s\in \{-1,1\}$, we say that
$\{\phi;\psi^1,\ldots,\psi^s\}_{(\eps_1,\ldots,\eps_s)}$ is \emph{a quasi-tight $\dm$-framelet} in $\dLp{2}$ if
$(\{\phi;\eps_1 \psi^1,\ldots \eps_s\psi^s\},\{\phi;\psi^1,\ldots,\psi^s\})$ is a dual $\dm$-framelet in $\dLp{2}$.
Equivalently speaking, $\{\phi;\psi^1,\ldots,\psi^s\}_{(\eps_1,\ldots,\eps_s)}$ is a quasi-tight $\dm$-framelet in $\dLp{2}$
if and only if $\{\phi;\psi^1,\ldots,\psi^s\}$ is an $\dm$-framelet in $\dLp{2}$ satisfying \eqref{framelet} and
every function $f\in \dLp{2}$ has the following representation:
\be \label{qtf:repr}
f=\sum_{k\in \dZ} \la f, \phi(\cdot-k)\ra \phi(\cdot-k)+
\sum_{j=0}^\infty \sum_{\ell=1}^s \sum_{k\in \dZ}
\eps_\ell \la f, \psi^\ell_{\dm^j;k}\ra \psi^\ell_{\dm^j;k}
\ee
with the series converging unconditionally in $\dLp{2}$.
When $\eps_1=\cdots=\eps_s=1$, a quasi-tight $\dm$-framelet becomes a tight $\dm$-framelet. A quasi-tight framelet is often obtained from a quasi-tight framelet filter bank. For $a,b_1,\ldots,b_s\in \dlp{0}$ and $\eps_1,\ldots,\eps_s\in \{-1,1\}$, we say that $\{a;b_1,\ldots,b_s\}_{(\eps_1,\ldots,\eps_s)}$ is \emph{a quasi-tight $\dm$-framelet filter bank} if
%
\be \label{qtffb}
\ol{\wh{a}(\xi)}\wh{a}(\xi+2\pi \omega)+
\eps_1 \ol{\wh{b_1}(\xi)}\wh{b}_1(\xi+2\pi \omega)
+\cdots+\eps_s \ol{\wh{b_s}(\xi)}\wh{b}_s(\xi+2\pi \omega)=\td(\omega),\qquad \forall\, \omega\in \Omega_\dm, \xi\in \dR.
\ee
%
\cite[Example~3.2.2]{hanbook}
probably is the first  to observe an example of a quasi-tight $2$-framelet $\{\phi;\psi^1,\psi^2\}_{(-1,1)}$ and a quasi-tight $2$-framelet filter bank
$\{a;b_1,b_2\}_{(-1,1)}$, where
\[
a=\{-\tfrac{1}{16},\tfrac{1}{4},\tfrac{5}{8},\tfrac{1}{4},-\tfrac{1}{16}\}_{[-2,2]},
\quad
b_1=\{-\tfrac{1}{16},\tfrac{1}{4},-\tfrac{3}{8},\tfrac{1}{4},-\tfrac{1}{16}\}_{[-2,2]},
\quad
b_2=\{-\tfrac{\sqrt{2}}{4},\tfrac{\sqrt{2}}{2},-\tfrac{\sqrt{2}}{4}\}_{[0,2]},
\]
and $\wh{\phi}(\xi):=\prod_{j=1}^\infty \wh{a}(2^{-j}\xi)$, $\wh{\psi^1}(\xi):=\wh{b_1}(\xi/2)\wh{\phi}(\xi/2)$, and
$\wh{\psi^2}(\xi):=\wh{b_2}(\xi/2)\wh{\phi}(\xi/2)$
with $\phi,\psi^1,\psi^2\in \Lp{2}$.
The above example in \cite[Example~3.2.2]{hanbook}
was accidentally obtained by applying the general algorithm developed in \cite{han15} for constructing dual framelet filter banks to the above low-pass filter $a$.

Let $\{a;b_1,\ldots, b_s\}_{(\eps_1,\ldots,\eps_s)}$ be a quasi-tight $\dm$-framelet filter bank with $\wh{a}(0)=1$.
Let $m$ denote the smallest order of the vanishing moments among the high-pass filters, i.e., $m:=\min(\vmo(b_1),\ldots,\vmo(b_s))$.
For $\omega\in \Omega_{\dm}\bs\{0\}$, one can easily deduce from \eqref{qtffb} that
$\ol{\wh{a}(\xi)}\wh{a}(\xi+2\pi \omega)+\bo(\|\xi\|^{m})=0$ as $\xi \to 0$. Consequently, by $\wh{a}(0)=1$, the filter $a$ must satisfy \emph{order $m$ sum rules with respect to $\dm$} (e.g., see \cite{jia98} for discussion on sum rules), i.e.,
\be \label{sr}
\wh{a}(\xi+2\pi \omega)=\bo(\|\xi\|^m),\quad \xi \to 0, \forall\, \omega\in \Omega_\dm\bs\{0\}.
\ee
We define $\sr(a,\dm)$ to be the largest nonnegative integer $m$ satisfying \eqref{sr}. Since $\wh{a}(0)=1$, we now trivially see that \eqref{vm:b} implies $\sr(a,\dm)\ge 1$ and consequently, the filter $a$ must satisfy
\be \label{sr:one}
\sum_{k\in \dZ} a(\gamma+\dm k)=|\det(\dm)|^{-1},\qquad \forall\; \gamma\in \Gamma_\dm:=[\dm[0,1)^d]\cap \dZ.
\ee
One can also trivially deduce from \eqref{qtffb} with $\omega=0$ that $1-|\wh{a}(\xi)|^2=\bo(\|\xi\|^{2m})$ as $\xi \to 0$.
Consequently,
for a quasi-tight $\dm$-framelet filter bank $\{a;b_1,\ldots, b_s\}_{(\eps_1,\ldots,\eps_s)}$, we always have
\be \label{qtffb:vm}
\min(\vmo(b_1),\ldots,\vmo(b_s))\le
\min(\sr(a,\dm), \tfrac{1}{2}\vmo(u_a)) \quad \mbox{with}\quad \wh{u_a}(\xi):=1-|\wh{a}(\xi)|^2.
\ee
%

Though the idea of quasi-tight framelets is simple, we have the following two main results of this paper on multivariate quasi-tight framelets with directionality or high vanishing moments.

\begin{theorem}\label{thm:qtf}
Let $\dm$ be a $d\times d$ dilation matrix and let $a\in \dlp{0}$ be a finitely supported real-valued sequence on $\dZ$ satisfying the basic sum rule condition in \eqref{sr:one}.
Then there always exist finitely supported real-valued high-pass filters $b_1,\ldots,b_s\in \dlp{0}$ and $\eps_1,\ldots,\eps_s\in \{-1,1\}$ such that $\{a;b_1,\ldots, b_s\}_{(\eps_1,\ldots,\eps_s)}$ is a quasi-tight $\dm$-framelet filter bank and every high-pass filter $b_\ell$ takes the form $c_\ell(\td_{\alpha_\ell}-\td_{\beta_\ell})$
for some $c_\ell\in \R$ and $\alpha_\ell,\beta_\ell\in \supp(a)$
(hence $b_\ell$ naturally has directionality, basic vanishing moments and symmetry property) for all $\ell=1,\ldots,s$.
Moreover,
\begin{enumerate}
\item[(1)] if in addition $\phi\in \dLp{2}$,
then $\{\phi; \psi^1,\ldots,\psi^s\}_{(\eps_1,\ldots,\eps_s)}$ is a (directional) compactly supported quasi-tight $\dm$-framelet in $\dLp{2}$, where 
\begin{equation}\label{phi:psi}
\wh{\phi}(\xi):=\prod_{j=1}^\infty \wh{a}((\dm^\tp)^{-j}\xi)
\quad \mbox{and}\quad
\wh{\psi^\ell}(\dm^\tp \xi):=\wh{b_\ell}(\xi)\wh{\phi}(\xi),\qquad \xi\in \dR, \ell=1,\ldots,s;
\end{equation}

\item[(2)] if in addition the following condition holds:
\be \label{qtf:tf:cond}
\mbox{all the coefficients in}\;\; \ol{\wh{a^{[\gamma_j]}}(\xi)}\wh{a^{[\gamma_k]}}(\xi)\;
\;
\mbox{are nonnegative for all}\; j,k=1,\ldots, \ddm,
\ee
(for example, the above condition in \eqref{qtf:tf:cond} is satisfied if the filter $a$ has nonnegative coefficients.)
then we can take $\eps_1=\cdots=\eps_s=1$ and therefore, $\{a;b_1,\ldots,b_s\}$ is a (directional) tight $\dm$-framelet filter bank and $\{\phi;\psi^1,\ldots,\psi^s\}$ is a (directional) tight $\dm$-framelet in $\dLp{2}$.
\end{enumerate}
\end{theorem}

Due to the special structure/construction, all the directional high-pass filters in Theorem~\ref{thm:qtf} have only order one vanishing moment. Due to the importance of vanishing moments, it is natural and important to ask whether one can construct a quasi-tight $\dm$-framelet filter bank achieving the highest possible order $\min(\sr(a,\dm),\frac{1}{2}\vmo(u_a))$ vanishing moments in \eqref{qtffb:vm} with $\wh{u_a}(\xi):=1-|\wh{a}(\xi)|^2$. This question is satisfactorily answered by the following result, for which we shall prove through a constructive algorithm.

\begin{theorem}\label{thm:qtf:vm}
Let $\dm$ be a $d\times d$ dilation matrix and let $a\in \dlp{0}$ be a finitely supported real-valued sequence on $\dZ$ satisfying the basic sum rule condition in \eqref{sr:one}.
Then there always exist finitely supported real-valued high-pass filters $b_1,\ldots,b_s\in \dlp{0}$ and $\eps_1,\ldots,\eps_s\in \{-1,1\}$ such that $\{a;b_1,\ldots, b_s\}_{(\eps_1,\ldots,\eps_s)}$ is a quasi-tight $\dm$-framelet filter bank and all the high-pass filters $b_1,\ldots,b_s$ have
order $m$ vanishing moments with $m:=\min(\sr(a,\dm),\frac{1}{2}\vmo(u_a))\ge 1$ and $\wh{u_a}(\xi):=1-|\wh{a}(\xi)|^2$. Define $\phi,\psi^1,\ldots,\psi^s$ as in \eqref{phi:psi}.
If $\phi\in \dLp{2}$, then $\{\phi; \psi^1,\ldots,\psi^s\}_{(\eps_1,\ldots,\eps_s)}$ is a compactly supported quasi-tight $\dm$-framelet in $\dLp{2}$ such that all $\psi^1,\ldots,\psi^s$ have at least order
$m$ vanishing moments.
\end{theorem}

The structure of the paper is as follows. In Section~2 we shall prove Theorem~\ref{thm:qtf} by a step-by-step algorithm and obtain
a general result on factorizing a Hermite matrix of $2\pi\dZ$-periodic trigonometric polynomials.
Then we shall provide a few examples of directional quasi-tight or tight framelets to illustrate the algorithm developed in the proof of Theorem~\ref{thm:qtf}.
In Section~3, we prove several auxiliary results first.
Then we shall prove Theorem~\ref{thm:qtf:vm} by a constructive algorithm.
We shall also present a result on quasi-tight framelets whose associated high-pass filters are purely differencing filters with the highest order of vanishing moments.
Finally, employing our proposed algorithms in this paper, we shall present in Section~4 several examples of quasi-tight framelets and quasi-tight framelet filter banks with high vanishing moments.
At the end of this paper, we shall further discuss several topics on quasi-tight framelets to be studied in the future.

\section{Proof of Theorem~\ref{thm:qtf} and Examples for Directional Quasi-tight Framelets}

In this section we will prove Theorem~\ref{thm:qtf} by a step-by-step algorithm. Then we establish a result on spectral factorization of a Hermite matrix of $2\pi\dZ$-periodic trigonometric polynomials. Finally we shall provide a few examples of directional quasi-tight or tight framelets to illustrate Theorem~\ref{thm:qtf}.

Recall that for a real number $c$,
the sign function is defined to be $\mbox{sgn}(c):=1$ for $c>0$, $\mbox{sgn}(0):=0$, and $\mbox{sgn}(c):=-1$ for $c<0$.
%
%
We first prove Theorem~\ref{thm:qtf} as follows.

\begin{proof}[Proof of Theorem~\ref{thm:qtf}]
By definition, we notice that $\{a;b_1,\ldots,b_s\}_{(\eps_1,\ldots,\eps_s)}$ is a quasi-tight $\dm$-framelet filter bank if and only if
\be \label{qtffb:2}
\sum_{\ell=1}^s \eps_\ell \ddm \Big[\wh{b^{[\gamma_1]}_\ell}(\xi),\ldots, \wh{b^{[\gamma_{\ddm}]}_\ell}(\xi)\Big]^\star
\Big[\wh{b^{[\gamma_1]}_\ell}(\xi),\ldots, \wh{b^{[\gamma_{\ddm}]}_\ell}(\xi)\Big]
=\cN_{a}(\xi)
\ee
with
\be \label{Na}
\cN_a(\xi):=
I_{\ddm}-
\ddm \Big[\wh{a^{[\gamma_1]}}(\xi),\ldots, \wh{a^{[\gamma_{\ddm}]}}(\xi)\Big]^\star
\Big[\wh{a^{[\gamma_1]}}(\xi),\ldots, \wh{a^{[\gamma_{\ddm}]}}(\xi)\Big].
\ee
We now construct the desired high-pass filters $b_1,\ldots,b_s$ by a recursive algorithm. Let $\cN=\cN_a$.
The main idea of the following proof has three steps:
(1) Eliminate the nonzero terms in the off-diagonal entries of $\cN$ one by one so that all the off-diagonal entries in the updated $\cN$ are identically zero. (2) Eliminate the nonzero nonconstant terms in the diagonal entries of the updated $\cN$ one by one so that the final updated $\cN$ is a constant diagonal matrix. (3) Prove that the constant diagonal matrix $\cN$ is the zero matrix.

Let $\ell:=1$. Suppose that some of the off-diagonal entries $[\cN]_{j,k}$ are not identically zero for some $1\le j<k\le \ddm$. Then $[\cN]_{j,k}$ has a nonzero term $c e^{-i\gamma\cdot\xi}$ with $c\ne 0$ and $\gamma\in \dZ$. By the definition of $\cN$, we observe that the term $c e^{-i\gamma\cdot\xi}$ must appear as one of the terms in $-\ol{\wh{a^{[\gamma_j]}}(\xi)}
\wh{a^{[\gamma_k]}}(\xi)$.
Therefore, there must exist $\alpha,\beta\in \dZ$ such that $a^{[\gamma_j]}(\alpha) a^{[\gamma_k]}(\beta)\ne 0$ (i.e., $\{\gamma_j+\dm \alpha, \gamma_k+\dm \beta\}\subset \supp(a)$) and $\beta-\alpha=\gamma$. Define
\be \label{b:ell}
b_\ell:=\sqrt{|c|/\ddm} (\td_{\gamma_j+\dm \alpha}-\td_{\gamma_k+\dm \beta}),
\qquad \eps_\ell:=-\mbox{sgn}(c).
\ee
If the additional condition in \eqref{qtf:tf:cond} holds, then $c$ must be a negative number and hence, $\eps_\ell=1$.
It follows directly from the definition of $b_\ell$ in \eqref{b:ell} that
\[
\wh{b_\ell^{[\gamma_j]}}(\xi)=\sqrt{|c|/\ddm} e^{-i\alpha\cdot \xi},
\quad
\wh{b_\ell^{[\gamma_k]}}(\xi)=-\sqrt{|c|/\ddm} e^{-i\beta\cdot \xi}
\quad \mbox{and}\quad
\wh{b_\ell^{[\gamma_p]}}(\xi)=0\qquad \forall\, p\in \{1,\ldots,\ddm\}\bs\{j,k\}.
\]
Consequently, the $\ddm\times \ddm$ matrix
\be \label{B:ell}
B_\ell(\xi):=\eps_\ell \ddm \Big[\wh{b^{[\gamma_1]}_\ell}(\xi),\ldots, \wh{b^{[\gamma_{\ddm}]}_\ell}(\xi)\Big]^\star
\Big[\wh{b^{[\gamma_1]}_\ell}(\xi),\ldots, \wh{b^{[\gamma_{\ddm}]}_\ell}(\xi)\Big]
\ee
has only four nonzero entries with
\be \label{B:ell:jk}
\begin{split}
[B_\ell(\xi)]_{j,j}&=\eps_\ell \ddm |\wh{b_\ell^{[\gamma_j]}}(\xi)|^2
=\eps_\ell |c|=-c,\\
[B_\ell(\xi)]_{k,k}&=\eps_\ell \ddm |\wh{b_\ell^{[\gamma_k]}}(\xi)|^2=
\eps_\ell |c|=-c,\\
[B_\ell(\xi)]_{j,k}&=\eps_\ell \ddm \ol{\wh{b_\ell^{[\gamma_j]}}(\xi)}
\wh{b_\ell^{[\gamma_k]}}(\xi)
=-\eps_\ell |c|e^{-i(\beta-\alpha)\cdot\xi}=
c e^{-i\gamma\cdot\xi},\\
[B_\ell(\xi)]_{k,j}&=\eps_\ell \ddm \ol{\wh{b_\ell^{[\gamma_k]}}(\xi)}
\wh{b_\ell^{[\gamma_j]}}(\xi)
=-\eps_\ell |c| e^{-i(\alpha-\beta)\cdot\xi}=
c e^{i\gamma\cdot\xi}.
\end{split}
\ee
Now replace/update $\cN$ by $\cN-B_\ell$ and replace $\ell$ by $\ell+1$ (i.e., increase $\ell$ by one). Because $\cN^\star(\xi)=\cN(\xi)$, by the above four identities, we conclude that the term $c e^{-i\gamma\cdot\xi}$ does not appear in the $(j,k)$-entry of the updated $\cN$ and only the constant terms in the diagonal entries of the previous $\cN$ are modified.
Hence, we can repeat this procedure until all the off-diagonal entries in $\cN$ are identically zero.

Now we deal with the diagonal matrix $\cN$.
Suppose that some of the diagonal entries $[\cN]_{j,j}$ are not constant for some $1\le j\le \ddm$. Then $[\cN]_{j,j}$ has a nonzero nonconstant term $c e^{-i\gamma\cdot\xi}$ with $c\ne 0$ and $\gamma\in \dZ\bs \{0\}$.
By the definition and construction of $\cN$, we observe that the term $c e^{-i\gamma\cdot\xi}$ must appear as one of the terms in $-\ol{\wh{a^{[\gamma_j]}}(\xi)}\wh{a^{[\gamma_j]}}(\xi)$.
Therefore, there must exist $\alpha,\beta\in \dZ$ such that $a^{[\gamma_j]}(\alpha) a^{[\gamma_j]}(\beta)\ne 0$ (i.e., $\{\gamma_j+\dm \alpha, \gamma_j+\dm \beta\}\subset \supp(a)$) and $\beta-\alpha=\gamma\ne 0$. Define
\be \label{b:ell:diagonal}
b_\ell:=\sqrt{|c|/\ddm} (\td_{\gamma_j+\dm \alpha}-\td_{\gamma_j+\dm \beta}),
\qquad \eps_\ell:=-\mbox{sgn}(c).
\ee
If the additional condition in \eqref{qtf:tf:cond} holds, then $c$ must be a negative number and hence, $\eps_\ell=1$.
It follows directly from the definition of $b_\ell$ in \eqref{b:ell:diagonal} that
\[
\wh{b_\ell^{[\gamma_j]}}(\xi)=\sqrt{|c|/\ddm} (e^{-i\alpha\cdot \xi}-e^{-i\beta\cdot\xi})
\quad \mbox{and}\quad
\wh{b_\ell^{[\gamma_p]}}(\xi)=0\qquad \forall\, p\in \{1,\ldots,\ddm\}\bs\{j\}.
\]
Consequently, the $\ddm\times \ddm$ matrix
$B_\ell(\xi)$ defined in \eqref{B:ell}
has only one nonzero entry at
\be \label{B:ell:jj}
[B_\ell(\xi)]_{j,j}=\eps_\ell \ddm |\wh{b_\ell^{[\gamma_j]}}(\xi)|^2=
\eps_\ell |c|(2-e^{-i(\alpha-\beta)\cdot\xi}-e^{i(\alpha-\beta)\cdot\xi})
=c(e^{-i\gamma\cdot\xi}+e^{i\gamma\cdot\xi}-2).
\ee
Now replace/update $\cN$ by $\cN-B_\ell$ and replace $\ell$ by $\ell+1$ (i.e., increase $\ell$ by one). Because $\cN^\star(\xi)=\cN(\xi)$, by the above identity, we conclude that neither the term $c e^{-i\gamma\cdot\xi}$ nor
$c e^{i\gamma\cdot\xi}$ appears in the $(j,j)$-entry of the updated $\cN$ and only the constant term in the $(j,j)$-entry of the previous $\cN$ is modified. We can repeat this procedure until all the nonzero nonconstant terms in the updated $\cN$ are identically zero. We set $s:=\ell-1$.

Therefore, we end up with
\be \label{qtf:2}
\cN_a(\xi)=\cN+\sum_{\ell=1}^s B_\ell(\xi)
\ee
such that $\cN$ is a diagonal matrix of constants. We now prove that $\cN=0$.
Note that both \eqref{B:ell:jk} and \eqref{B:ell:jj} trivially imply that the sum of every row of $B_\ell(0)$ must be zero. Since the filter $a$ satisfies the basic sum rule condition in \eqref{sr:one}, we have $\wh{a^{[\gamma_1]}}(0)=\cdots=
\wh{a^{[\gamma_{\ddm}]}}(0)=\ddm^{-1}$.
Now we trivially deduce from the definition of $\cN_a$ in \eqref{Na} that
all the diagonal entries of $\cN_a(0)$ are $1-\ddm^{-1}$ and all the off-diagonal entries of $\cN_a(0)$ are $-\ddm^{-1}$. Consequently the sum of every row of $\cN_a(0)$ is $(1-\ddm^{-1})+(\ddm-1)(-\ddm^{-1})=0$. Therefore, we conclude from \eqref{qtf:2} that the sum of every row of $\cN$ must be zero. However, $\cN$ is a diagonal matrix of constants and thus, we must have $\cN=0$. Since $\cN=0$, by our definition of $B_\ell$ in \eqref{B:ell} and using \eqref{qtf:2}, we conclude that
\eqref{qtffb:2} is satisfied and $\{a;b_1,\ldots,b_s\}_{(\eps_1,\ldots,\eps_s)}$ is a quasi-tight $\dm$-framelet filter bank.

If the additional condition in \eqref{qtf:tf:cond} is satisfied, by our above construction we have $\eps_1=\cdots=\eps_s=1$ and hence $\{a;b_1,\ldots,b_s\}$ is a tight $\dm$-framelet filter bank.
\end{proof}

For $u\in \dlp{0}$, by $N(\wh{u})$ we denote the total number of nonzero terms in the $2\pi\dZ$-periodic trigonometric polynomial $\wh{u}$. That is, $N(\wh{u})=\#\mbox{supp}(u)$, the cardinality of the support of the filter $u$.
From the above proof of Theorem~\ref{thm:qtf}, it is not difficult to conclude that the number $s$ of high-pass filters in Theorem~\ref{thm:qtf} is given by
\[
s=\sum_{1\le j<k\le \ddm} N\left(\ol{\wh{a^{[\gamma_j]}}(\xi)}
\wh{a^{[\gamma_k]}}(\xi)\right)+
\frac{1}{2} \sum_{j=1}^{\ddm}\left( N\left(|\wh{a^{[\gamma_j]}}(\xi)|^2\right)-1\right).
\]
Moreover, by the special structure of the high-pass filters in Theorem~\ref{thm:qtf}, we also have $s\le \binom{\#\supp(a)}{2}$.

Suppose that $\{a;b_1,\ldots,b_s\}_{(\eps_1,\ldots,\eps_s)}$ is a quasi-tight $\dm$-framelet filter bank. Then it is trivial to observe that
$\{a;b_1(\cdot+\dm k_1),\ldots,b_s(\cdot+\dm k_s)\}_{(\eps_1,\ldots,\eps_s)}$ are quasi-tight $\dm$-framelet filter banks for all $k_1,\ldots,k_s\in \dZ$.
If $b_1=c b_2$ for some $c\in \R$, then $\{a; \sqrt{|\eps_1 c^2+\eps_2|} b_2,\ldots,b_s\}_{(\mbox{sgn}(\eps_1c^2+\eps_2),\eps_3,\ldots,\eps_s)}$ is also a quasi-tight $\dm$-framelet filter bank with at most $s-1$ number of high-pass filters.
Up to the above trivial variants, from our proof of Theorem~\ref{thm:qtf}, one can conclude that the constructed quasi-tight $\dm$-framelet filter bank in Theorem~\ref{thm:qtf} is essentially unique.

For $\dm=2I_d$, let us consider the simplest low-pass filter: the Haar low-pass filter $a$ with $a(k)=2^{-d}$ for $k\in \{0,1\}^d$ and $a(k)=0$ for $k\in \dZ\bs \{0,1\}^d$.
Since $\wh{a^{[\gamma]}}=2^{-d}$ for all $\gamma\in \Omega_{2I_d}:=\{0,1\}^d$ and since all the coefficients of $a$ are nonnegative, it follows easily from the above proof/algorithm of Theorem~\ref{thm:qtf} that $\{a;b_1,\ldots,b_s\}$ with $s:=2^{d-1}(2^d-1)$ is a tight $2I_d$-framelet filter bank with directionality, where $b_1,\ldots,b_s$ are given by $2^{-d}(\td_{\gamma_1}-\td_{\gamma_2})$ for every undirected edge with different endpoints $\gamma_1,\gamma_2\in \{0,1\}^d$.
Note that the associated $2I_d$-refinable function is the Haar refinable function $\phi$ with $\phi=\chi_{[0,1]^d}$. Consequently, the corresponding $\{\phi;\psi^1,\ldots,\psi^s\}$ is a tight $2I_d$-framelet in $\dLp{2}$ with directionality.
This yields the same construction in \cite{lcsht16} for the two-dimensional Haar framelet (i.e. $d=2$) and \cite[Theorem~1]{hlz17} for any dimension $d\in \N$.
Theorem~\ref{thm:qtf}
can be also easily applied to box spline refinable functions to obtain directional box spline tight $2I_d$-framelets in \cite[Theorem~2]{hlz17} which are obtained by applying the projection method in \cite{han14} to the above $d$-dimensional directional Haar framelets. Theorem~\ref{thm:qtf} improves \cite{hlz17} on directional tight $2I_d$-framelets by being applicable to arbitrary $\dm$-refinable functions whose low-pass filters have nonnegative coefficients.

For $1\le j\le r$, by $e_j$ we denote the $j$th unit coordinate column vector in $\R^r$, i.e., $e_j$ is the $r\times 1$ vector with its only nonzero element being $1$ at the $j$th entry. Using the idea in the proof of Theorem~\ref{thm:qtf}, we now establish a result generalizing the well-known result on the matrix spectral factorization.

\begin{theorem}\label{thm:qsos}
Let $\cA$ be an $r\times r$ matrix of $2\pi\dZ$-periodic trigonometric polynomials in $d$ variables with real coefficients such that $\cA^\star(\xi)=\cA(\xi)$ for all $\xi\in \dR$. For $1\le j\le r$, let $\kappa_r$ be the sum of the $j$th column of $\cA(0)$, i.e., $(\kappa_1,\ldots,\kappa_r):=
(1,\ldots,1)\cA(0)$.
Then there exist $\eps_0,\ldots,\eps_s\in \{-1,1\}$ and $r\times 1$ vectors $\fu_1,\ldots,\fu_s$ of $2\pi\dZ$-periodic trigonometric polynomials with real coefficients such that
\be \label{qsos:hermite}
\cA(\xi)=\mbox{sgn}(\kappa_1) |\kappa_1|e_1 e_1^\star+\cdots+\mbox{sgn}(\kappa_r)
|\kappa_r|e_r e_r^\star+
\eps_1 \fu_1(\xi) \fu_1^\star(\xi)+\cdots
+\eps_s \fu_s(\xi) \fu_s^\star(\xi),
\ee
and each vector function of $\fu_\ell$ has only two nonzero entries with one being $c_\ell e^{-i\alpha_\ell \cdot \xi}$ and the other being $-c_\ell e^{-i\beta_\ell \cdot \xi}$ for some $c_\ell\in \R$ and $\alpha_\ell,\beta_\ell\in \dZ$ for $\ell=1,\ldots,s$. In terms of the matrix form, \eqref{qsos:hermite} can be rewritten as
\be \label{qsos}
\cA(\xi)=\cB(\xi) \mbox{diag}\bigl(\mbox{sgn}(\kappa_1),\ldots,
\mbox{sgn}(\kappa_r),\eps_1,\ldots,\eps_s\bigl)
\cB^\star(\xi),
\ee
where $\cB(\xi):=[\sqrt{|\kappa_1|}e_1,\ldots,\sqrt{|\kappa_r|}e_r, \fu_1,\ldots,\fu_s]$ is an $r\times (r+s)$ matrix obtained by putting all the column vectors $e_1,\ldots, e_r, \fu_1,\ldots,\fu_s$ together.
\end{theorem}

\begin{proof}
To prove the claim, we construct the desired vectors $\fu_1,\ldots,\fu_s$ of $2\pi\dZ$-periodic trigonometric polynomials by following a similar recursive algorithm as in the proof of Theorem~\ref{thm:qtf}.

Let $\ell:=1$ and $\cN:=\cA$. Suppose that some of the off-diagonal entries $[\cN]_{j,k}$ are not identically zero for some $1\le j<k\le r$. Then $[\cN]_{j,k}$ has a nonzero term $c e^{-i\gamma\cdot\xi}$ with $c\ne 0$ and $\gamma\in \dZ$.
Define
%
\[
\fu_\ell(\xi):=
\sqrt{|c|}e_j-\sqrt{|c|} e^{i\gamma\cdot \xi} e_k \quad \mbox{and}\quad \eps_\ell:=-\mbox{sgn}(c).
\]
Consequently, the $r\times r$ matrix
$\eps_\ell \fu_\ell(\xi) \fu_\ell^\star(\xi)$ has only four nonzero entries $-c, -c, ce^{-i\gamma\cdot \xi}$ and $ce^{i\gamma\cdot \xi}$ at the positions $(j,j), (k,k), (j,k)$ and $(k,j)$, respectively.
Now replace/update $\cN$ by $\cN-\eps_\ell \fu_\ell(\xi) \fu_\ell^\star(\xi)$ and replace $\ell$ by $\ell+1$ (i.e., increase $\ell$ by one). Because $\cN^\star(\xi)=\cN(\xi)$, we conclude that the term $c e^{-i\gamma\cdot\xi}$ does not appear in the $(j,k)$-entry of the updated $\cN$ and only the constant terms in the diagonal entries of the previous $\cN$ are modified.
Hence, we can repeat this procedure until all the off-diagonal entries in $\cN$ are identically zero.

Now we deal with the diagonal matrix $\cN$.
Suppose that some of the diagonal entries $[\cN]_{j,j}$ are not constant for some $1\le j\le r$. Then $[\cN]_{j,j}$ has a nonzero nonconstant term $c e^{-i\gamma\cdot\xi}$ with $c\ne 0$ and $\gamma\in \dZ\bs \{0\}$.
Define
\[
\fu_\ell(\xi):=\sqrt{|c|}(1-e^{-i\gamma\cdot\xi})e_j\quad
\mbox{and}\quad \eps_\ell:=-\mbox{sgn}(c).
\]
Then $\eps_\ell \fu_\ell(\xi)\fu_\ell^\star(\xi)$ has only one nonzero entry $c(e^{-i\gamma\cdot\xi}+e^{i\gamma\cdot\xi}-2)$ at the $(j,j)$-entry.
Now replace/update $\cN$ by $\cN-\eps_\ell \fu_\ell(\xi)\fu_\ell^\star(\xi)$ and replace $\ell$ by $\ell+1$. Because $\cN^\star(\xi)=\cN(\xi)$, we conclude that neither the term $c e^{-i\gamma\cdot\xi}$ nor
$c e^{i\gamma\cdot\xi}$ appears in the $(j,j)$-entry of the updated $\cN$ and only the constant term in the $(j,j)$-entry of the previous $\cN$ is modified. We can repeat this procedure until all the nonzero nonconstant terms in the updated $\cN$ are identically zero. We set $s:=\ell-1$.

Therefore, we end up with
$\cA(\xi)=\cN+\sum_{\ell=1}^s \eps_\ell \fu_\ell(\xi)\fu_\ell^\star(\xi)$
such that $\cN$ is a diagonal matrix of constants.
Note that the sum of every column in $\eps_\ell \fu_\ell(0)\fu_\ell^\star(0)$ is zero. Consequently, we must have $\cN=\mbox{diag}(\kappa_1,\ldots,\kappa_r)$.
This completes the proof.
\end{proof}

Though the sum of Hermitian squares of $2\pi\dZ$-periodic trigonometric polynomials is a challenging problem in real algebraic geometric,
as a direct consequence of Theorem~\ref{thm:qsos}, we have the following result on quasi-sum of Hermitian squares of $2\pi\dZ$-periodic trigonometric polynomials.

\begin{cor}\label{cor:qsos}
Let $\cA$ be a $2\pi\dZ$-periodic trigonometric polynomial in $d$ variables with real coefficients such that $\ol{\cA(\xi)}=\cA(\xi)$. Then there exist $\eps_1,\ldots,\eps_s\in \{-1,1\}$ and $2\pi\dZ$-periodic trigonometric polynomials $\fu_1,\ldots,\fu_s$ with real coefficients such that
\be \label{qsos:scalar}
\cA(\xi)=\mbox{sgn}(\cA(0)) (\sqrt{|\cA(0)|})^2+\eps_1|\fu_1(\xi)|^2+\cdots
+\eps_s |\fu_s(\xi)|^2,\qquad \forall\, \xi\in \dR
\ee
and each function $\fu_\ell$ has only two nonzero entries with one being $c_\ell e^{-i\alpha_\ell \cdot \xi}$ and the other being $-c_\ell e^{-i\beta_\ell \cdot \xi}$ for some $c_\ell\in \R$ and $\alpha_\ell,\beta_\ell\in \dZ$ for $\ell=1,\ldots,s$.
\end{cor}

We finish this section by providing several examples of directional quasi-tight or tight framelets.
To obtain a quasi-tight $\dm$-framelet $\{\phi;\psi^1,\ldots,\psi^s\}_{(\eps_1,\ldots,\eps_s)}$ in $\dLp{2}$, we have to check the technical condition $\phi\in \dLp{2}$.
Let $a\in \dlp{0}$ with $\wh{a}(0)=1$ and $m:=\sr(a,\dm)$.
For $1\le p\le \infty$, we now recall a technical quantity (e.g., see \cite[(7.2.2)]{hanbook} and \cite[(4.3)]{han03jat} for its definition and importance in wavelet analysis):
\begin{equation}\label{smaM}
\sm_p(a,\dm):=\tfrac{d}{p}-\log_{\rho(\dm)} \rho_m(a,M)_p \quad \mbox{and}\quad \sm(a,\dm):=\sm_2(a,\dm),
\end{equation}
where $\rho(\dm)$ is the spectral radius of $\dm$ and
\[
\rho_m(a,\dm)_p:=|\det(\dm)|\sup\{\lim_{n\to \infty}\|a_n*(\nabla^\mu \td)\|_{\dlp{p}}^{1/n} \setsp \mu\in \dNN, |\mu|=m\},
\]
where $\wh{a_n}(\xi):=\wh{a}(\xi)\wh{a}(\dm^\tp\xi)\cdots \wh{a}((\dm^\tp)^{n-1}\xi)$ and $\nabla^\mu \td$ is defined later in \eqref{nabla}. Let $\phi$ be defined in \eqref{phi:psi}. If $\sm(a,\dm)>0$, then $\phi\in \dLp{2}$ and moreover, $\int_{\dR} |\wh{\phi}(\xi)|^2 (1+\|\xi\|^2)^\tau d\xi<\infty$ for all $0\le \tau<\sm(a,\dm)$.

We now provide a few examples to illustrate Theorem~\ref{thm:qtf}.
We first demonstrate that the condition in \eqref{qtf:tf:cond} can be also satisfied by some low-pass filters having negative coefficients as well.

\begin{example}{\rm
For $d=1$ and $\dm=2$, we consider the following low-pass filter
\[
a=\{\tfrac{5}{29}, \tfrac{5}{29}, -\tfrac{1}{58},
{\bf\underline{-\tfrac{1}{58}}},
\tfrac{5}{29}, \tfrac{5}{29}, \tfrac{5}{29}, \tfrac{5}{29}
\}_{[-3,4]}.
\]
Then clearly $a$ satisfies the basic sum rules in \eqref{sr:one} and the additional condition in \eqref{qtf:tf:cond} is satisfied: By the definition of the filter $a$, we have
%
\begin{align*}
&\ol{\wh{a^{[0]}}(\xi)}\wh{a^{[1]}}(\xi)=
\tfrac{25}{841} e^{-2i\xi} + \tfrac{45}{1682} e^{-i\xi} + \tfrac{20}{841} + \tfrac{301}{3364} e^{i\xi}
+ \tfrac{20}{841} e^{2i\xi} + \tfrac{45}{1682}e^{3i\xi} + \tfrac{25}{841}e^{4i\xi}
\\
&|\wh{a^{[0]}}(\xi)|^2= |\wh{a^{[1]}}(\xi)|^2 =
\tfrac{25}{841} e^{-3i\xi}
+ \tfrac{45}{1682} e^{-2i\xi}
+ \tfrac{20}{841} e^{-i\xi}
+ \tfrac{301}{3364}
+ \tfrac{20}{841} e^{i\xi}
+ \tfrac{45}{1682} e^{2i\xi}
+ \tfrac{25}{841} e^{3i\xi}.
\end{align*}
Note that not all coefficients of the filter $a$ are nonnegative but \eqref{qtf:tf:cond} is satisfied. By Theorem~\ref{thm:qtf}, we obtain a tight $2$-framelet filter bank $\{a,b_1,\ldots, b_{13}\}$ given by
\begin{align*}
&b_1 = \tfrac{2\sqrt{5}}{29} \{ {\bf\underline{-1}}, 1\}_{[0, 1]}, \quad
b_{2} = \tfrac{2\sqrt{5}}{29}\{ {\bf\underline{-1}}, 0, 1 \}_{[0, 2]}, \quad
b_{3} = \tfrac{2\sqrt{5}}{29} \{ -1, {\bf\underline{0}}, 1 \}_{[-1, 1]},\quad
b_4 = \tfrac{\sqrt{301}}{58}\{-1, {\bf\underline{1}} \}_{[-1, 0]}, \quad\\
&b_5 = \tfrac{2\sqrt{5}}{29}\{-1, {\bf\underline{0}}, 0, 1\}_{[-1, 2]}, \quad
b_6 = \tfrac{3\sqrt{10}}{58}\{-1, 0, {\bf\underline{0}}, 1\}_{[-2, 1]},\quad
b_7 = \tfrac{3\sqrt{10}}{58}\{-1, 0, {\bf\underline{0}}, 0, 1\}_{[-2, 2]}, \quad\\
&b_{8} = \tfrac{3\sqrt{10}}{58}\{-1, {\bf\underline{0}}, 0, 0, 1\}_{[-1, 3]}, \quad
b_9 = \tfrac{3\sqrt{10}}{58}\{-1, {\bf\underline{0}}, 0, 0, 0, 1\}_{[-1, 4]},\quad
b_{10} = \tfrac{5}{29}\{-1, 0, {\bf\underline{0}}, 0, 0, 1\}_{[-2, 3]}, \quad\\
&b_{11} = \tfrac{5}{29}\{-1, 0, {\bf\underline{0}}, 0, 0, 0, 1\}_{[-2, 4]}, \quad
b_{12} = \tfrac{5}{29}\{-1, 0, 0, {\bf\underline{0}}, 0, 0, 1\}_{[-3, 3]}, \quad
b_{13} = \tfrac{5}{29}\{-1, 0, 0, {\bf\underline{0}}, 0, 0, 0, 1\}_{[-3, 4]}.
\end{align*}
By calculation, we have $\sm(a,2)\approx 0.992335$ and hence $\phi\in \Lp{2}$.
Then $\{\phi;\psi^1,\ldots,\psi^{13}\}$ is a tight $2$-framelet in $\Lp{2}$, where $\phi,\psi^1,\ldots,\psi^{13}$ are defined in \eqref{phi:psi} with $\dm=2$ and $s=13$.
}\end{example}

\begin{example}  \label{ex:1dVM1}{\rm
For $d = 1$ and $\dm = 2$, we consider the following interpolatory filter
\begin{equation}\label{aI4}
\ta = \{ -\tfrac{1}{32}, 0, \tfrac{9}{32},
{\bf\underline{\tfrac{1}{2}}},
\tfrac{9}{32}, 0, -\tfrac{1}{32} \}_{[-3, 3]}. \end{equation}
Using Theorem~\ref{thm:qtf}, we have a quasi-tight $2$-framelet filter bank
$ \{\ta; \tb_1, \ldots, \tb_7\}_{(\eps_1, \ldots,\eps_7)}$, where
\begin{align*}
&\tb_1 = \tfrac{3}{8}\{
{\bf\underline{-1}}, 1\}_{[0, 1]},
\quad
\tb_2 = \tfrac{3}{8}\{-1,
{\bf\underline{1}}\}_{[-1, 0]}, \quad
\tb_3 = \tfrac{3\sqrt{7}}{32}\{-1,
{\bf\underline{0}}, 1\}_{[-1, 1]}, \quad
\tb_4 = \tfrac{1}{32}\{-1, 0, 0,
{\bf\underline{0}},
0, 0, 1\}_{[-3, 3]},
\\
&\tb_5 = \tfrac{1}{8}\{-1,
{\bf\underline{0}},
0, 1\}_{[-1, 2]}, \quad
	\tb_6 = \tfrac{1}{8}\{-1, 0,
	{\bf\underline{0}},
	1\}_{[-2, 1]},
	\quad
	\tb_7 = \tfrac{3\sqrt{2}}{32}\{-1,
	{\bf\underline{0}},
	0, 0, 1\}_{[-1, 3]}
\end{align*}
with $\eps_1 =\cdots=\eps_4=1$ and $\eps_5 =\cdots=\eps_7=-1$. Since $\sm(a,2)\approx 2.440765$, $\phi\in \Lp{2}$ and $\{\phi;\psi^1,\ldots,\psi^7\}$ is a quasi-tight $2$-framelet in $\Lp{2}$, where $\phi,\psi^1,\ldots,\psi^7$ are defined in \eqref{phi:psi} with $\dm=2$ and $s=7$.
See Figure~\ref{fig:1dVM1} for the graphs of $\phi,\psi^1,\ldots,\psi^7$.
}\end{example}

\begin{figure}[ht!]
	\centering
	\begin{subfigure}[]{0.24\textwidth}
		\includegraphics[width=\textwidth, height=0.8\textwidth]{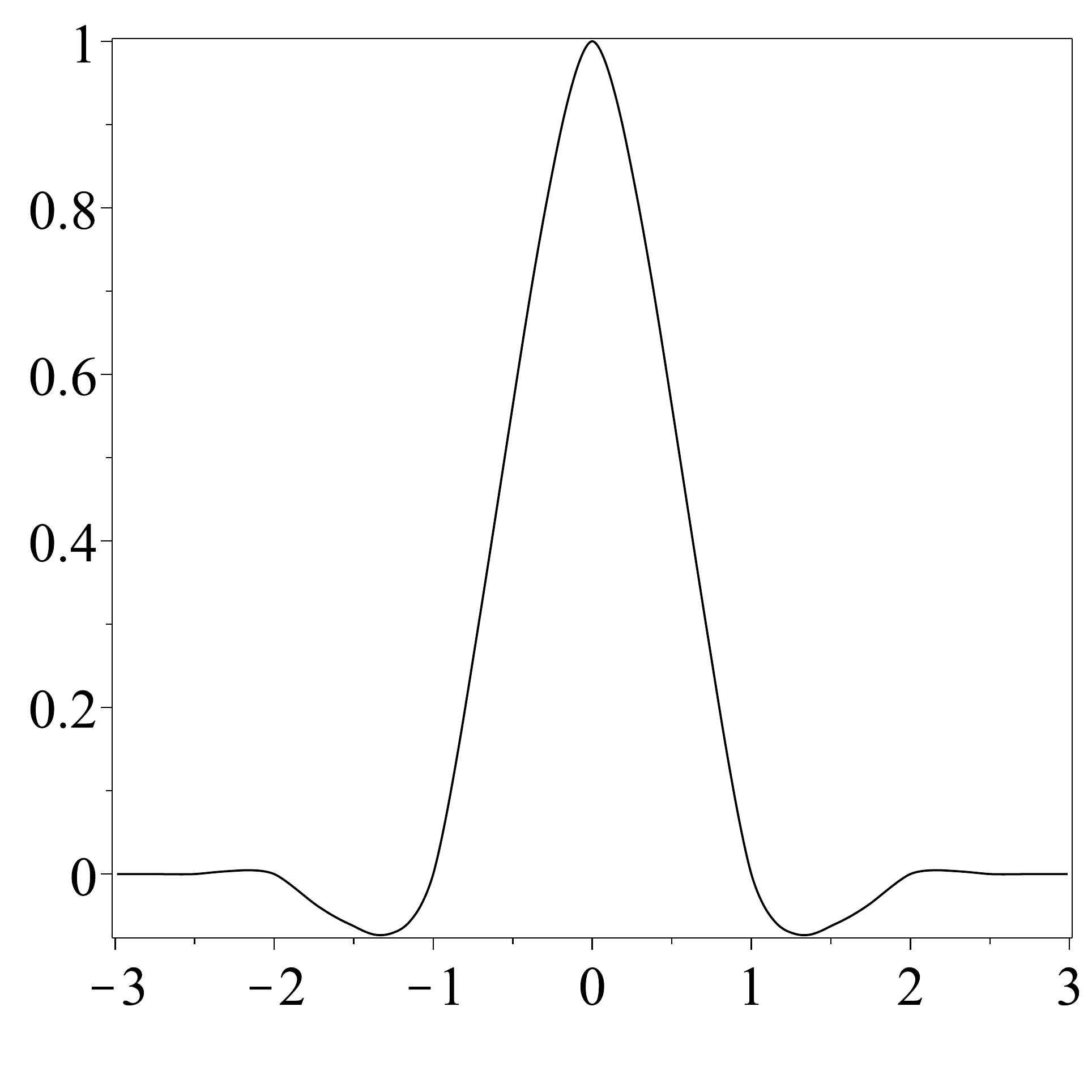}
		\caption{$\phi$}
	\end{subfigure}
	\begin{subfigure}[]{0.24\textwidth}
		\includegraphics[width=\textwidth, height=0.8\textwidth]{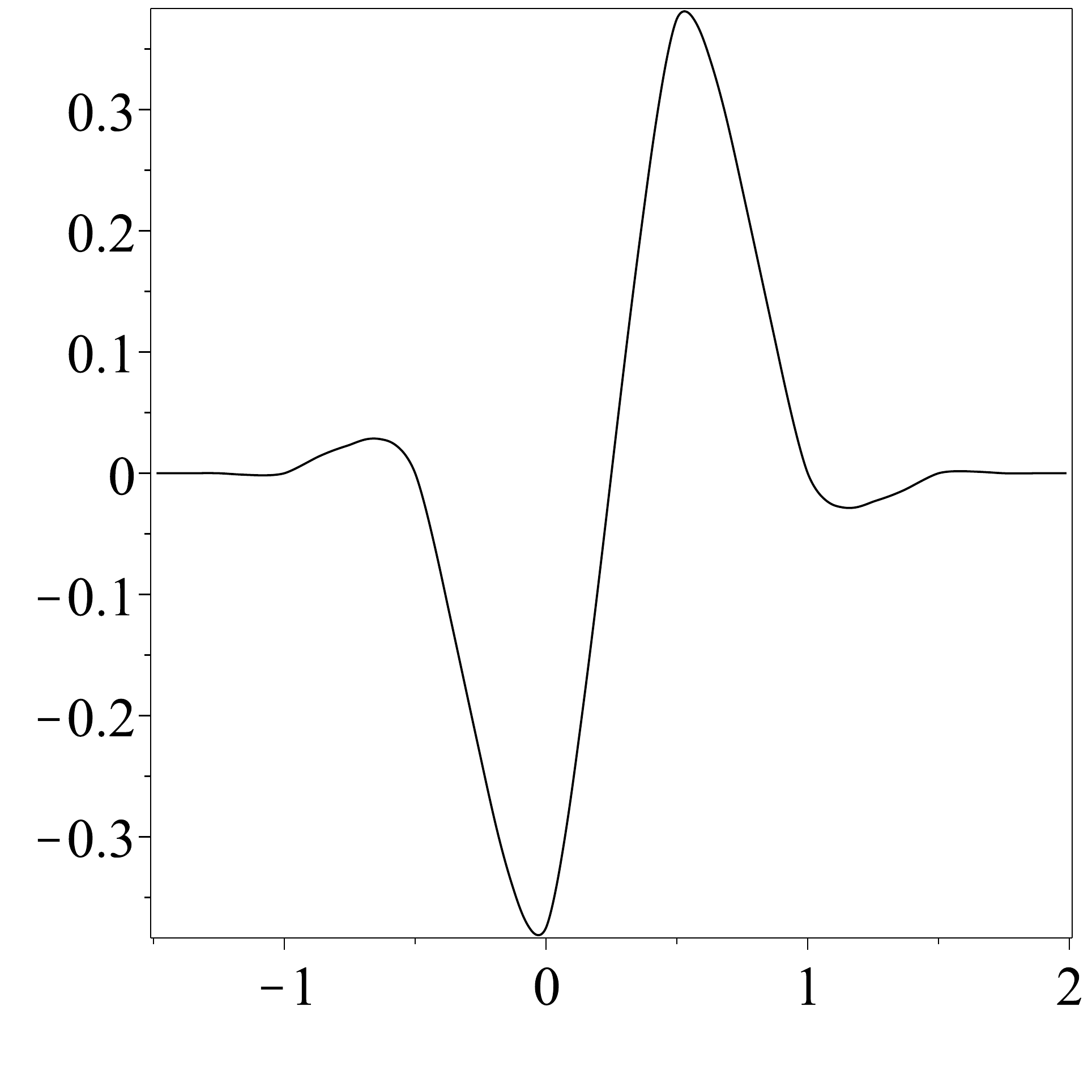}
		\caption{$\psi^1$}
	\end{subfigure}
	\begin{subfigure}[]{0.24\textwidth}
		\includegraphics[width=\textwidth, height=0.8\textwidth]{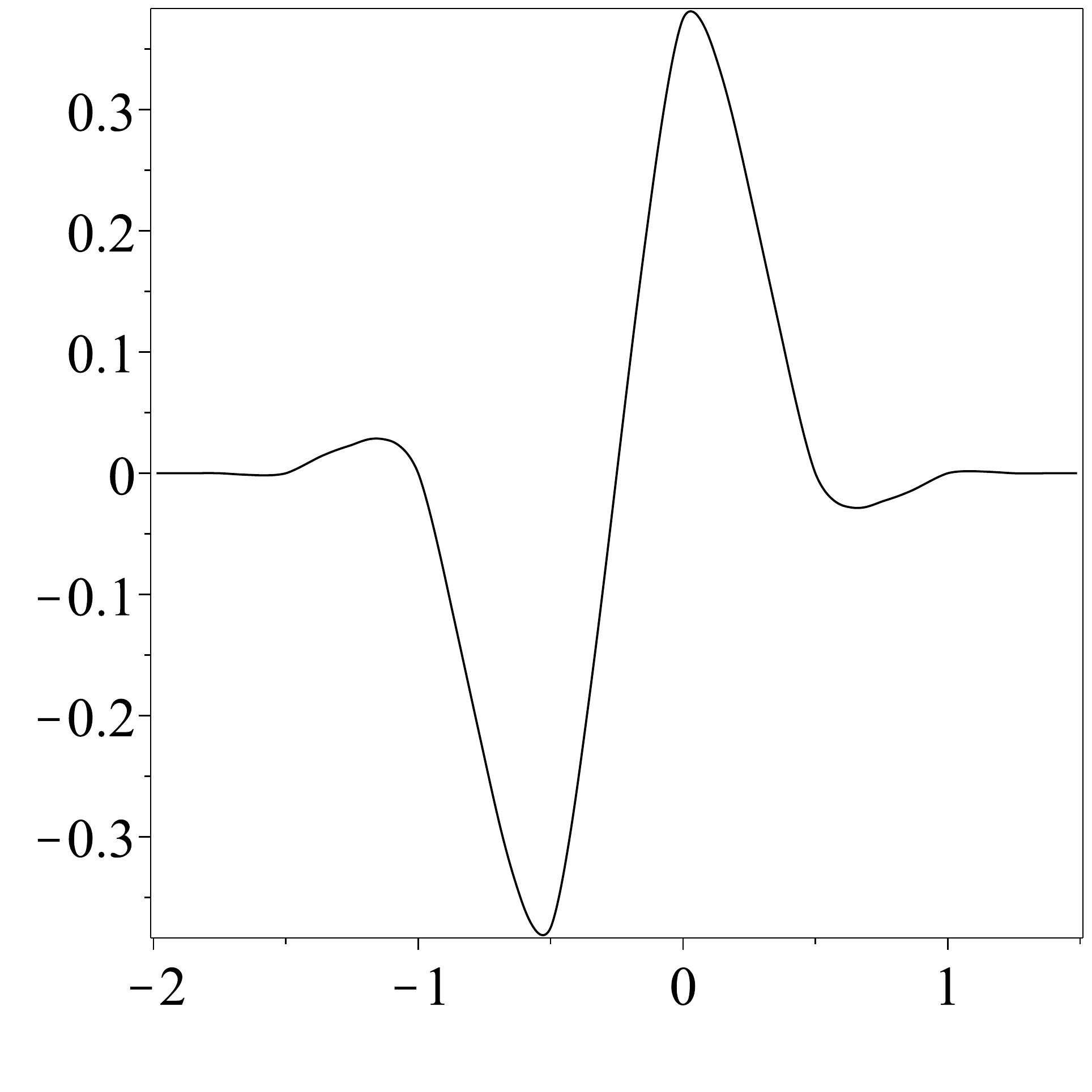}
		\caption{$\psi^2$}
	\end{subfigure}
	\begin{subfigure}[]{0.24\textwidth}
		\includegraphics[width=\textwidth, height=0.8\textwidth]{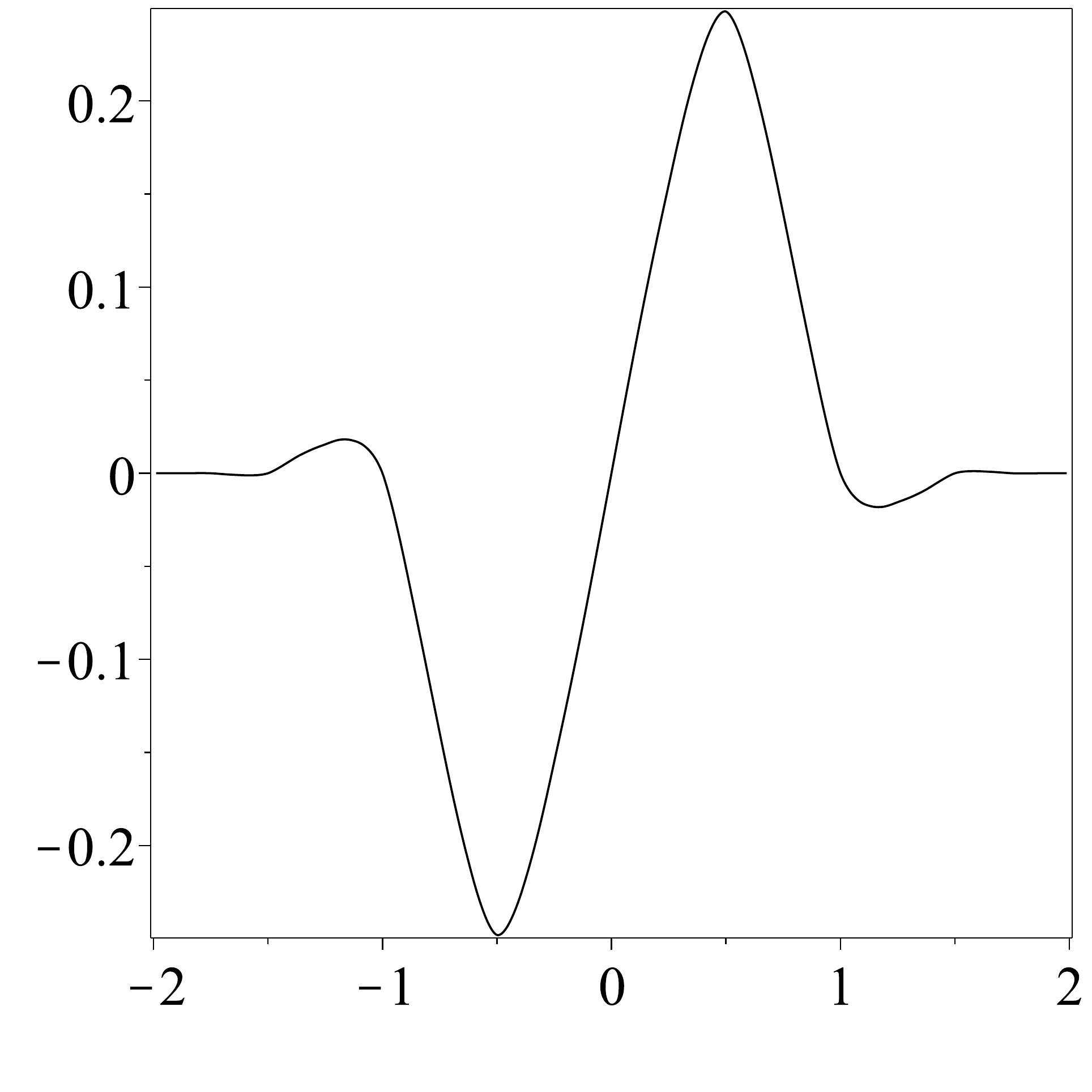}
		\caption{$\psi^3 $ }
	\end{subfigure}
	\\
	\begin{subfigure}[]{0.24\textwidth}
		\includegraphics[width=\textwidth, height=0.8\textwidth]{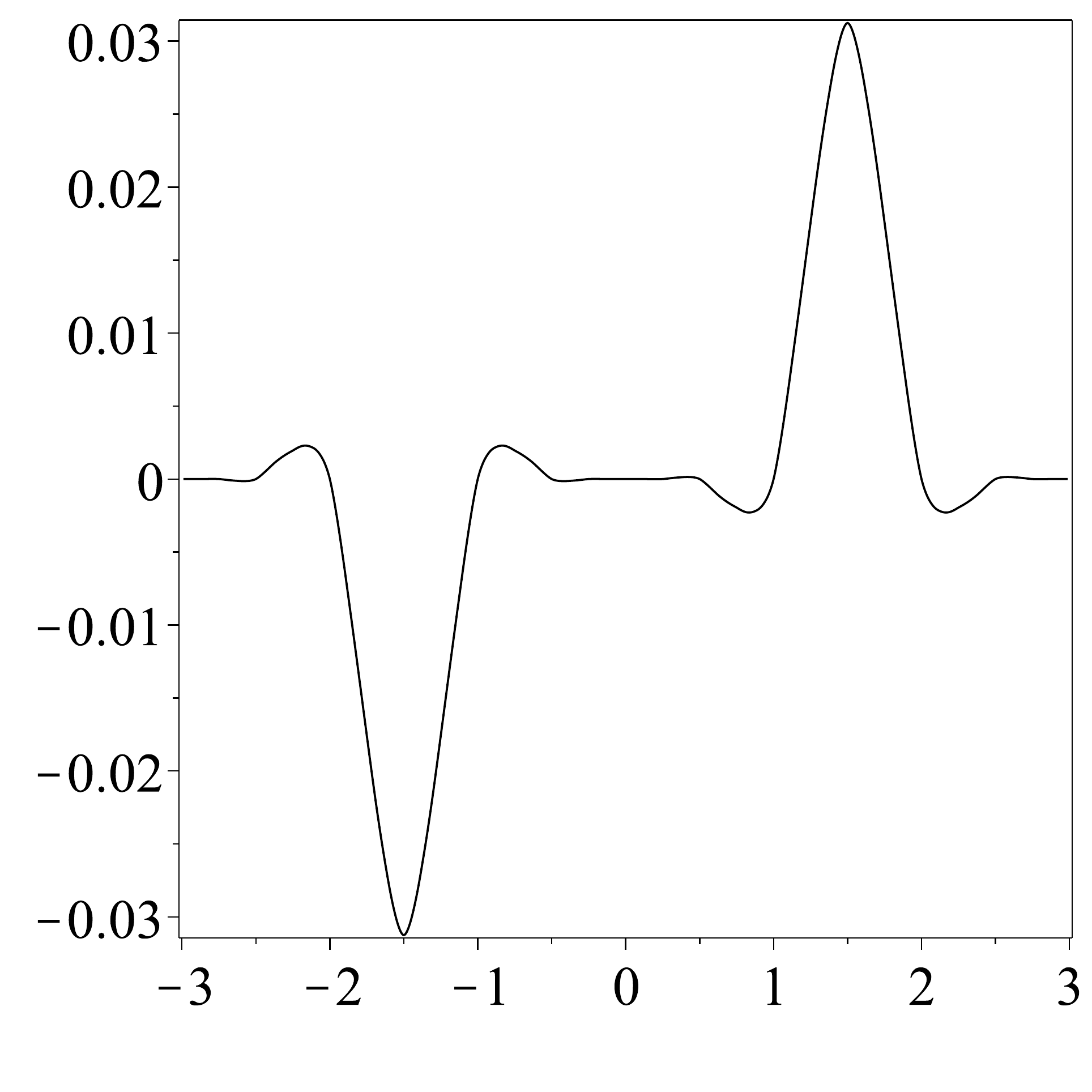}
		\caption{$\psi^4$}
	\end{subfigure}
	\begin{subfigure}[]{0.24\textwidth}
		\includegraphics[width=\textwidth, height=0.8\textwidth]{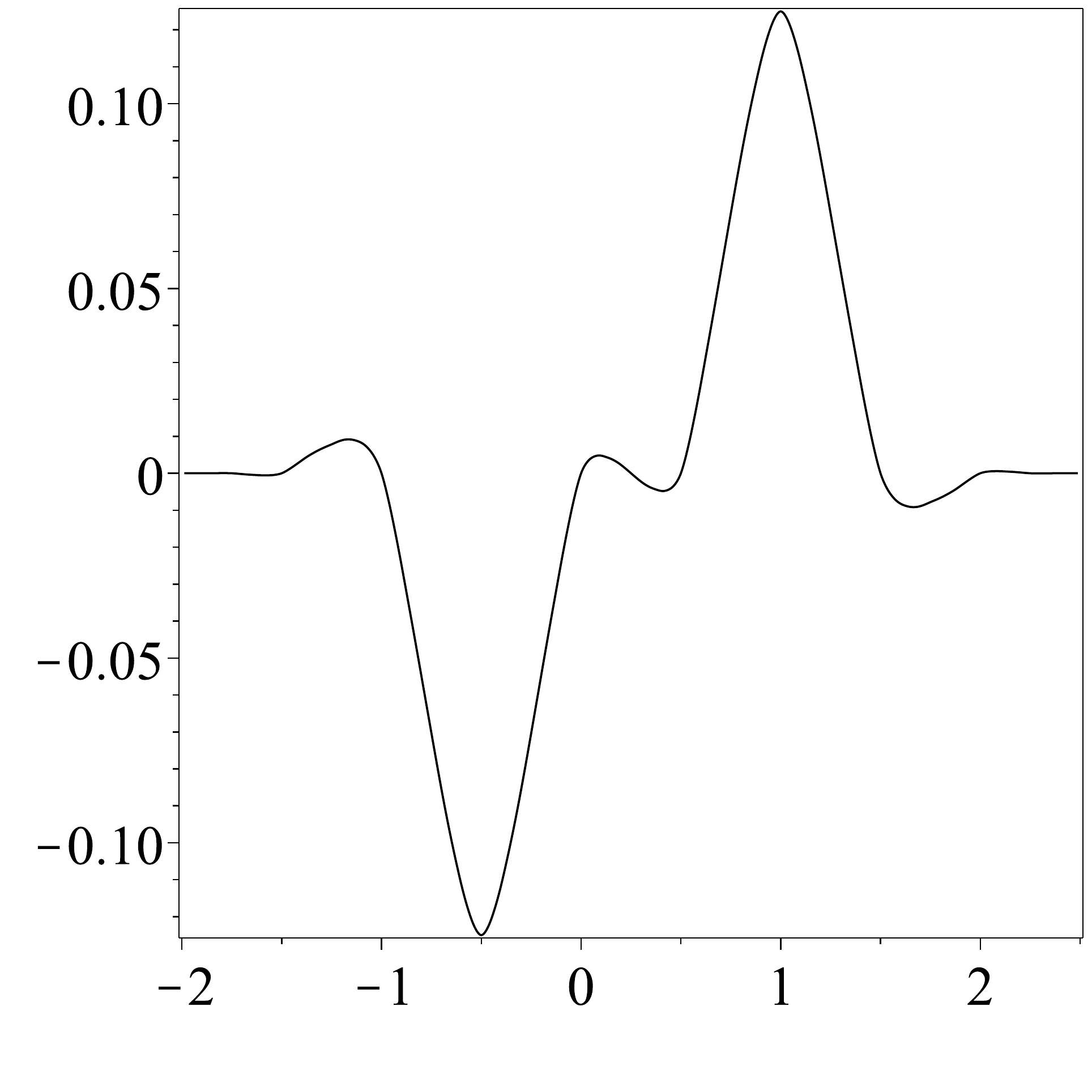}
		\caption{$\psi^5$}
	\end{subfigure}
	\begin{subfigure}[]{0.24\textwidth}
		\includegraphics[width=\textwidth, height=0.8\textwidth]{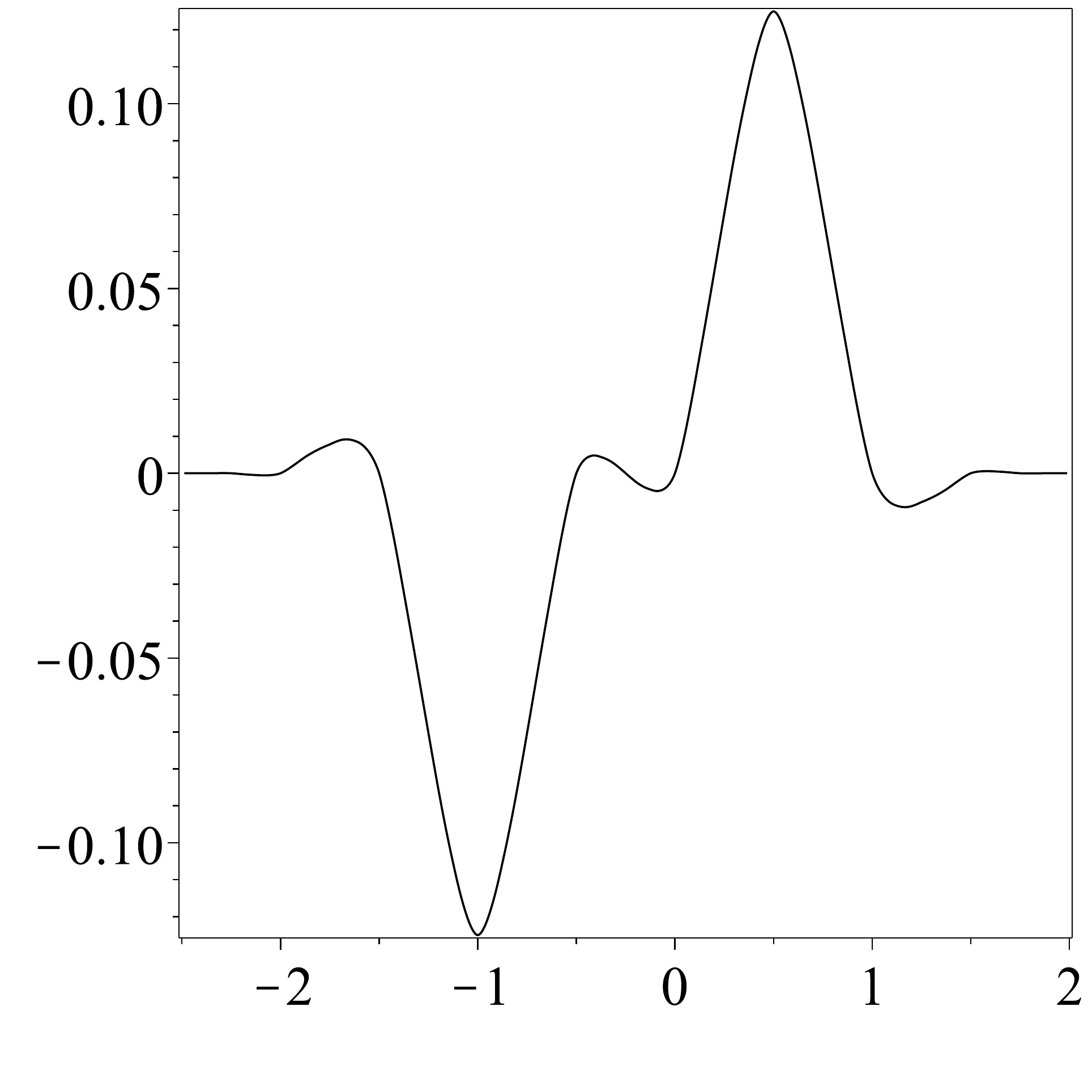}
		\caption{$\psi^6$}
	\end{subfigure}
	\begin{subfigure}[]{0.24\textwidth}
		\includegraphics[width=\textwidth, height=0.8\textwidth]{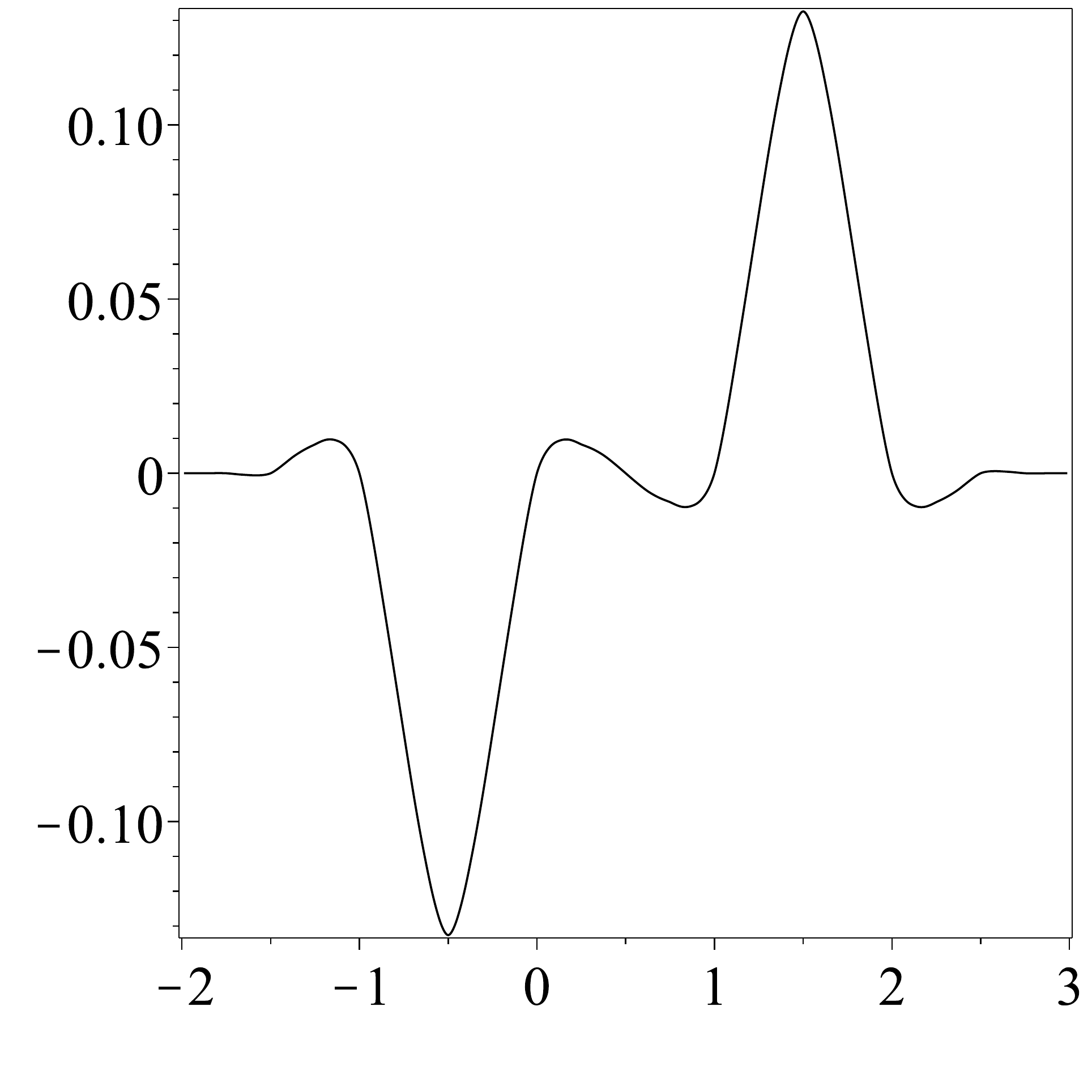}
		\caption{$\psi^7 $ }
	\end{subfigure}
\caption{Graphs of $\phi,\psi^1,\ldots,\psi^7$
in Example~\ref{ex:1dVM1}, where
$\{\phi;\psi^1,\ldots,\psi^7\}_{(1,1,1,1,-1,-1,-1)}$ is a compactly supported quasi-tight $2$-framelet in $\Lp{2}$.
}\label{fig:1dVM1}
\end{figure}

\begin{example}  \label{ex:qcxVM1} {\rm
For $d=2$, we consider the quincunx dilation matrix $\dm_{\sqrt{2}}$ and a low-pass filter $a$:
\begin{equation}\label{quincunx}
\dm_{\sqrt{2}} = \begin{bmatrix} 1 & 1 \\ 1 & -1
\end{bmatrix},\qquad\quad
\ta = \begin{bmatrix}
	0 & \tfrac{1}{8} & 0 \\[0.5em]
	\tfrac{1}{8} & {\bf\underline{\tfrac{1}{2}}} & \tfrac{1}{8} \\[0.5em]
	0 & \tfrac{1}{8} & 0
	\end{bmatrix}_{[-1, 1]\times [-1, 1]}.
\end{equation}
%
Using Theorem~\ref{thm:qtf}, we have a directional tight $\dm_{\sqrt{2}}$-framelet filter bank $\{\ta; \tb_1, \ldots, \tb_8\}$, where
{\small \begin{align*}
&\tb_1 = \tfrac{1}{4}\begin{bmatrix}
{\bf\underline{-1}} & 1
	\end{bmatrix}_{[0, 1]\times[0, 0]},
	\quad
	\tb_2 = \tfrac{1}{4}\begin{bmatrix}
	-1 & 1
	\end{bmatrix}_{[0, 1]\times[-1, -1]},
	\quad
	\tb_3 =  \tfrac{1}{4}\begin{bmatrix}
	-1 \\ 1
	\end{bmatrix}_{[1, 1]\times[-1, 0]},
	\quad
	\tb_4 = \tfrac{1}{4}\begin{bmatrix}
	-1 \\ 1
	\end{bmatrix}_{[1, 1]\times[0, 1]},\\
&\tb_5 = \tfrac{1}{8}\begin{bmatrix}
	-1 & {\bf\underline{0}} & 1
	\end{bmatrix}_{[-1, 1]\times[0, 0]},
	\quad
	\tb_6 =  \tfrac{1}{8} \begin{bmatrix}
	-1 \\ {\bf\underline{0}} \\ 1
	\end{bmatrix}_{[0, 0]\times[-1, 1]},
	\quad
	\tb_7 = \tfrac{\sqrt{2}}{8}\begin{bmatrix}
	{\bf\underline{0}} & -1 \\ 1 & 0
	\end{bmatrix}_{[0, 1]\times[-1, 0]},
	\quad
	\tb_8 = \tfrac{\sqrt{2}}{8}\begin{bmatrix}
	-1 & 0 \\
	{\bf\underline{0}} & 1
	\end{bmatrix}_{[0, 1]\times[0, 1]}.
\end{align*}
}
Note that $\sr(a,\dm_{\sqrt{2}})=2$.
Since $\sm(a,\dm_{\sqrt{2}})\approx 1.577645$, $\phi\in L_2(\R^2)$ and $\{\phi;\psi^1,\ldots,\psi^8\}$ is a (directional) tight $\dm_{\sqrt{2}}$-framelet in $L_2(\R^2)$, where $\phi,\psi^1,\ldots,\psi^8$ are defined in \eqref{phi:psi} with $\dm=\dm_{\sqrt{2}}$ and $s=8$. See Figure~\ref{fig:qcxVM1} for the graphs of $\phi,\psi^1,\ldots,\psi^8$.
} \end{example}

\begin{figure}[hbt]
	\centering
	\begin{subfigure}[]{0.31\textwidth}
		\includegraphics[width=\textwidth, height=0.7\textwidth]{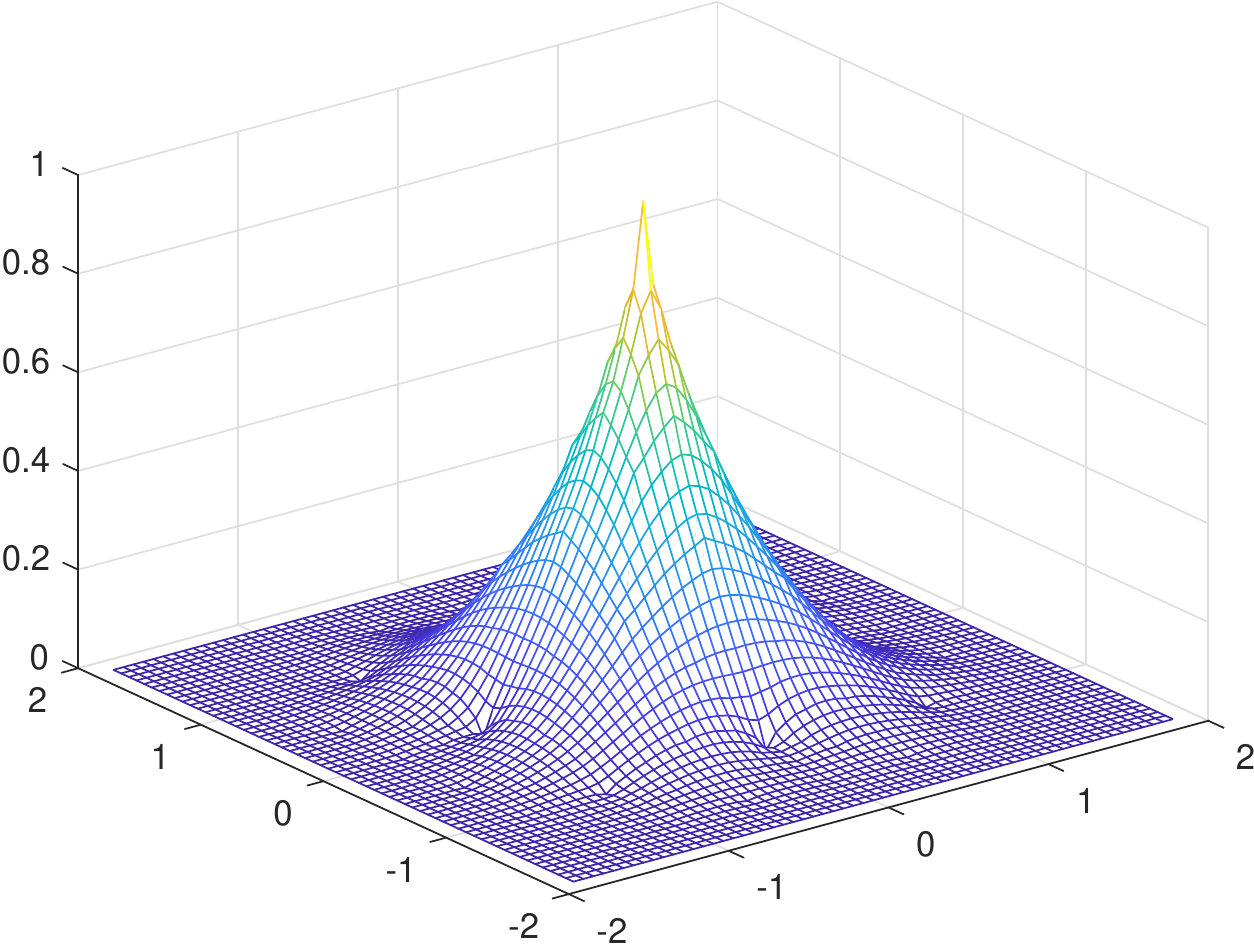}
		\caption{$\phi$}
	\end{subfigure}\vspace{0.05in}
	\begin{subfigure}[]{0.31\textwidth}
		\includegraphics[width=\textwidth, height=0.7\textwidth]{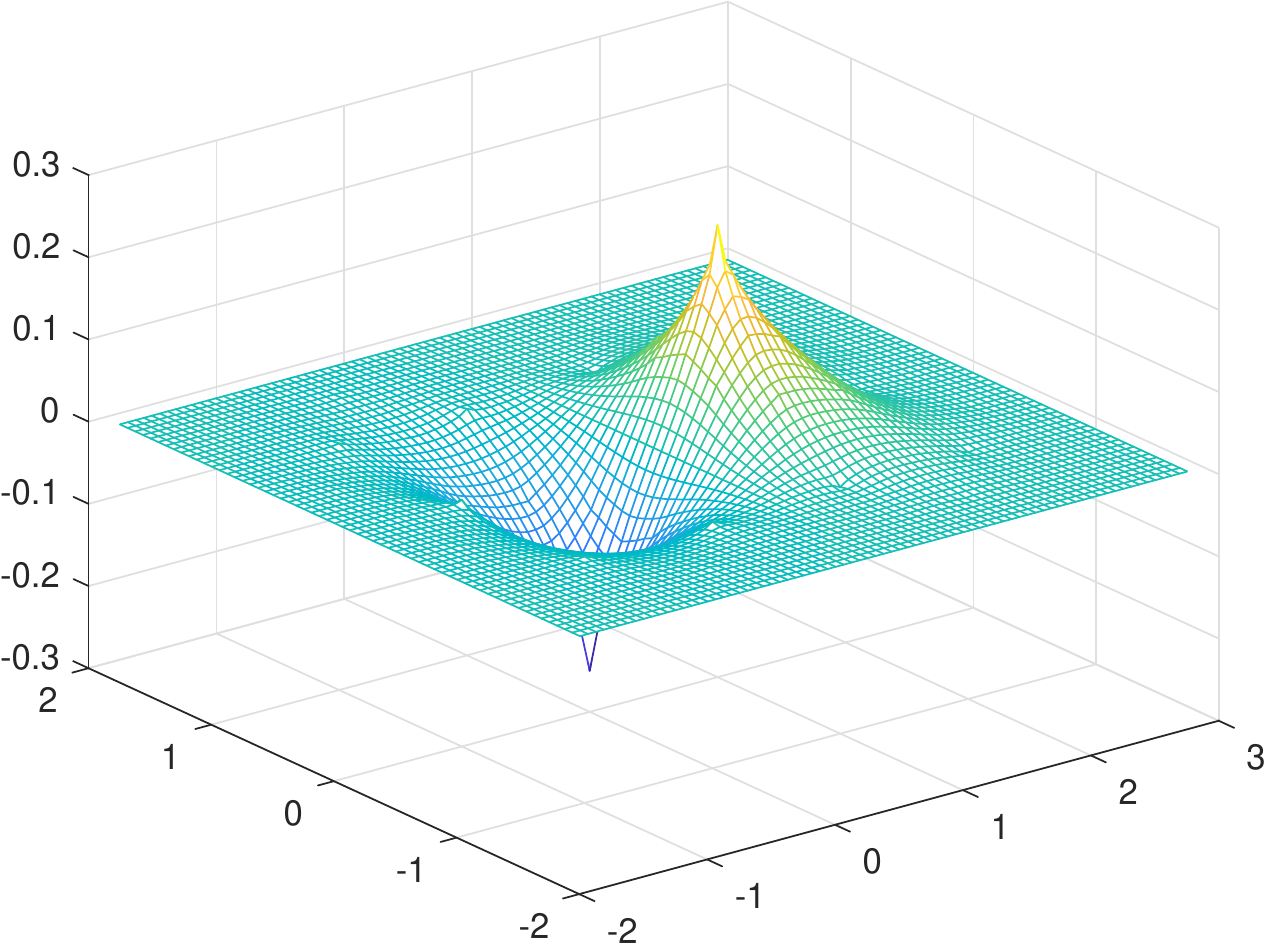}
		\caption{$\psi^1$}
	\end{subfigure}\vspace{0.05in}
	\begin{subfigure}[]{0.31\textwidth}
		\includegraphics[width=\textwidth, height=0.7\textwidth]{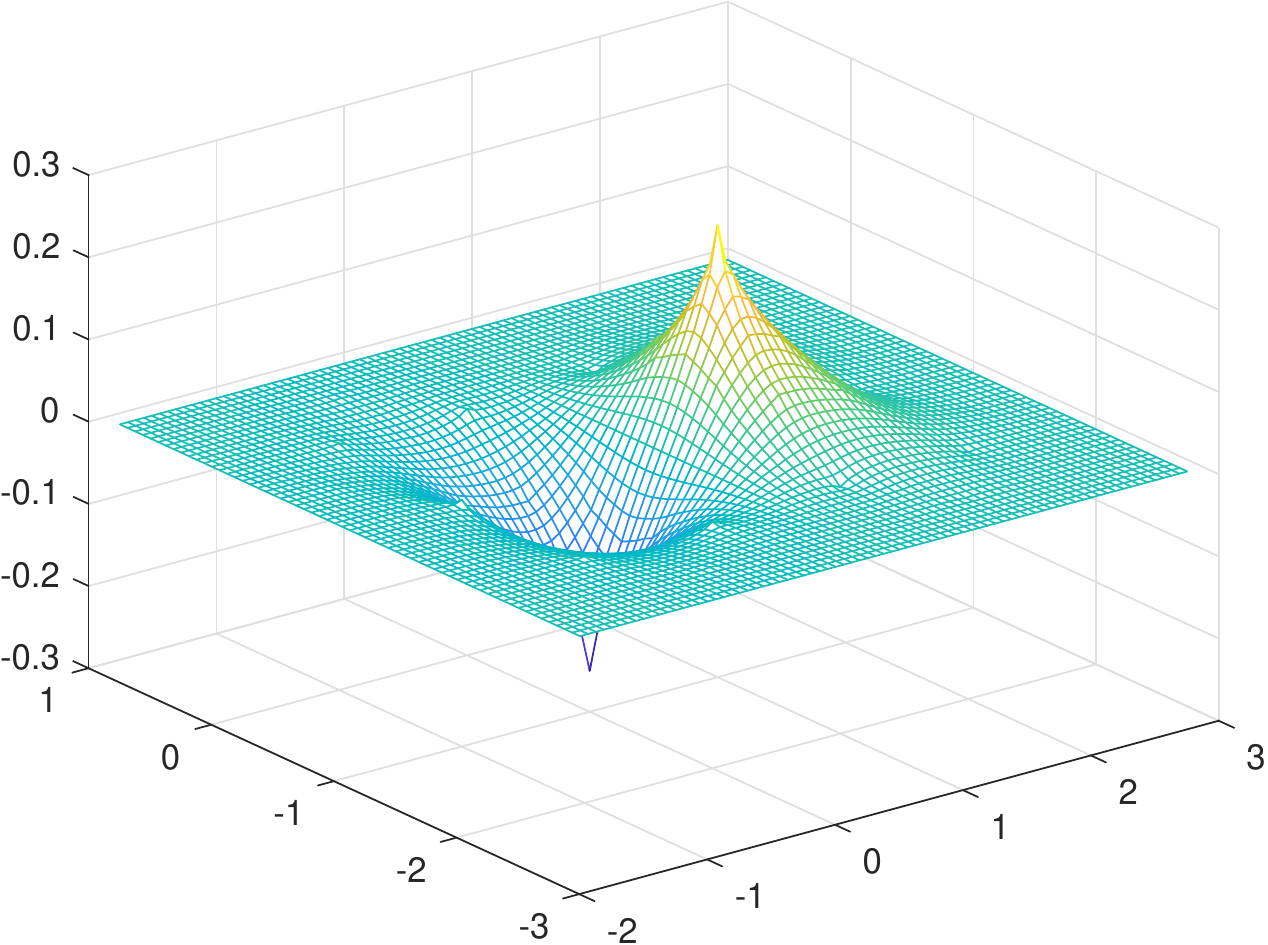}
		\caption{$\psi^2$}
	\end{subfigure}\vspace{0.05in}
	\\
	\begin{subfigure}[]{0.31\textwidth}
		\includegraphics[width=\textwidth, height=0.7\textwidth]{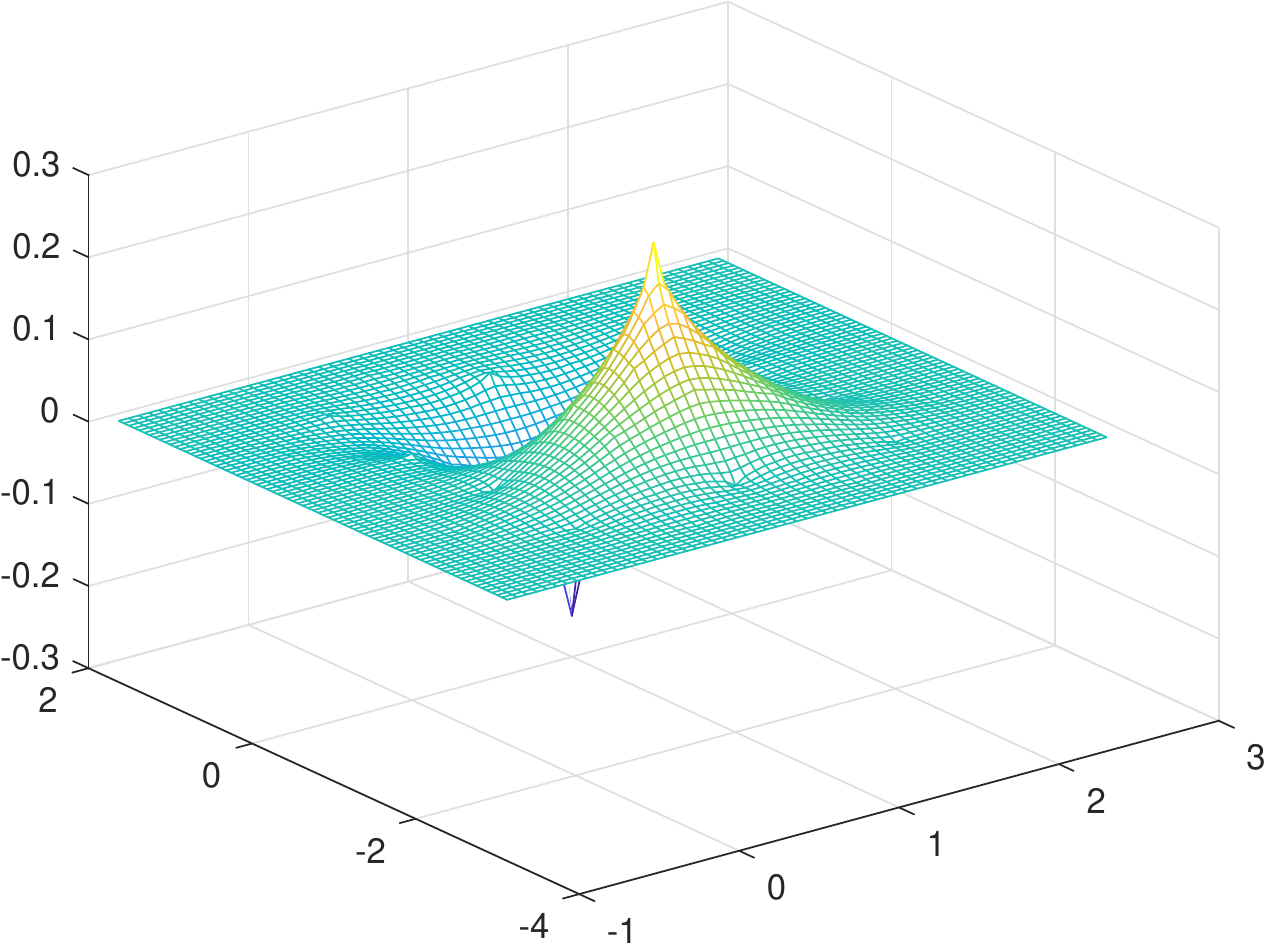}
		\caption{$\psi^3 $ }
	\end{subfigure}\vspace{0.05in}
	\begin{subfigure}[]{0.31\textwidth}
		\includegraphics[width=\textwidth, height=0.7\textwidth]{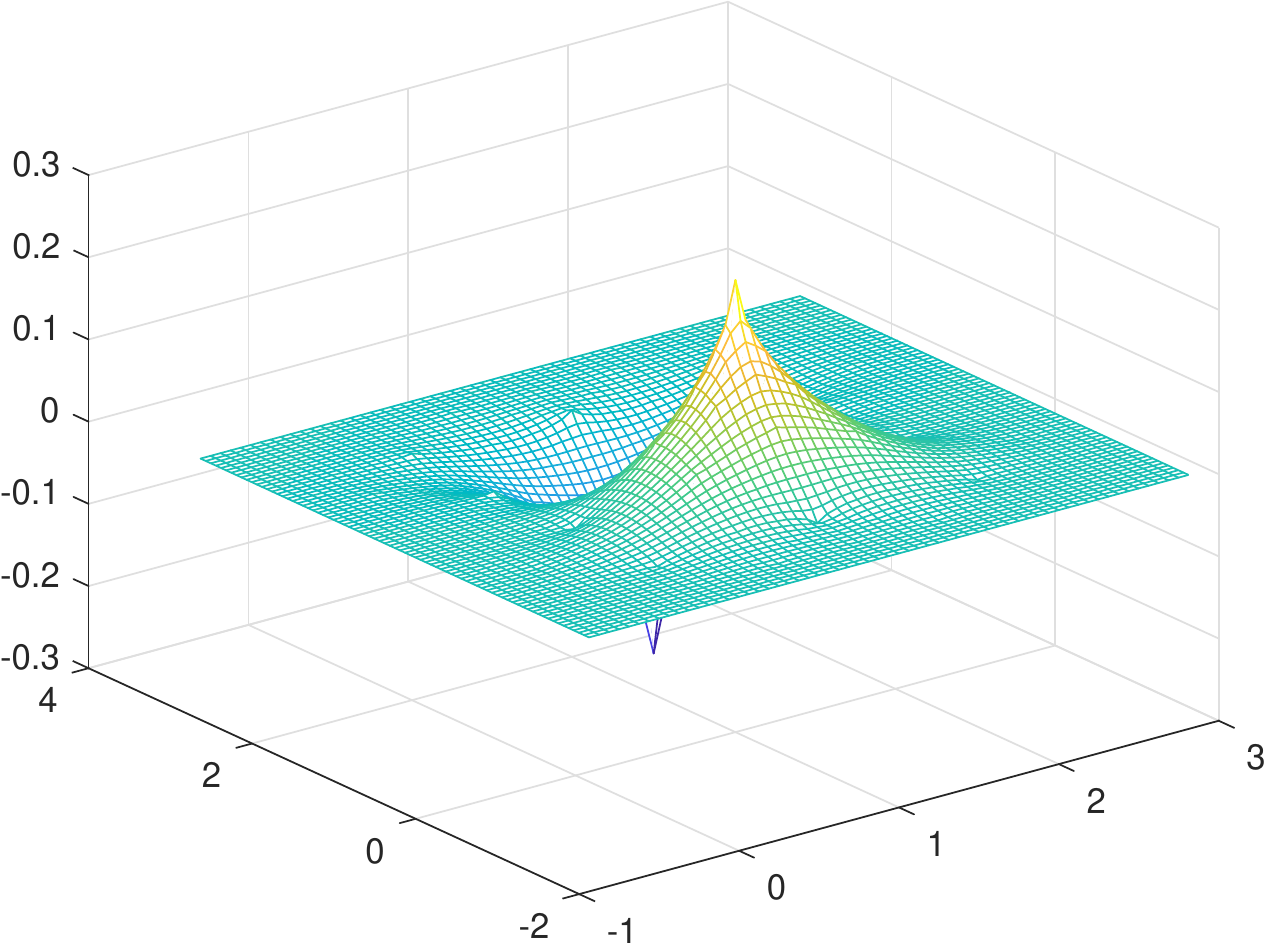}
		\caption{$\psi^4$}
	\end{subfigure}\vspace{0.05in}
	\begin{subfigure}[]{0.31\textwidth}
		\includegraphics[width=\textwidth, height=0.7\textwidth]{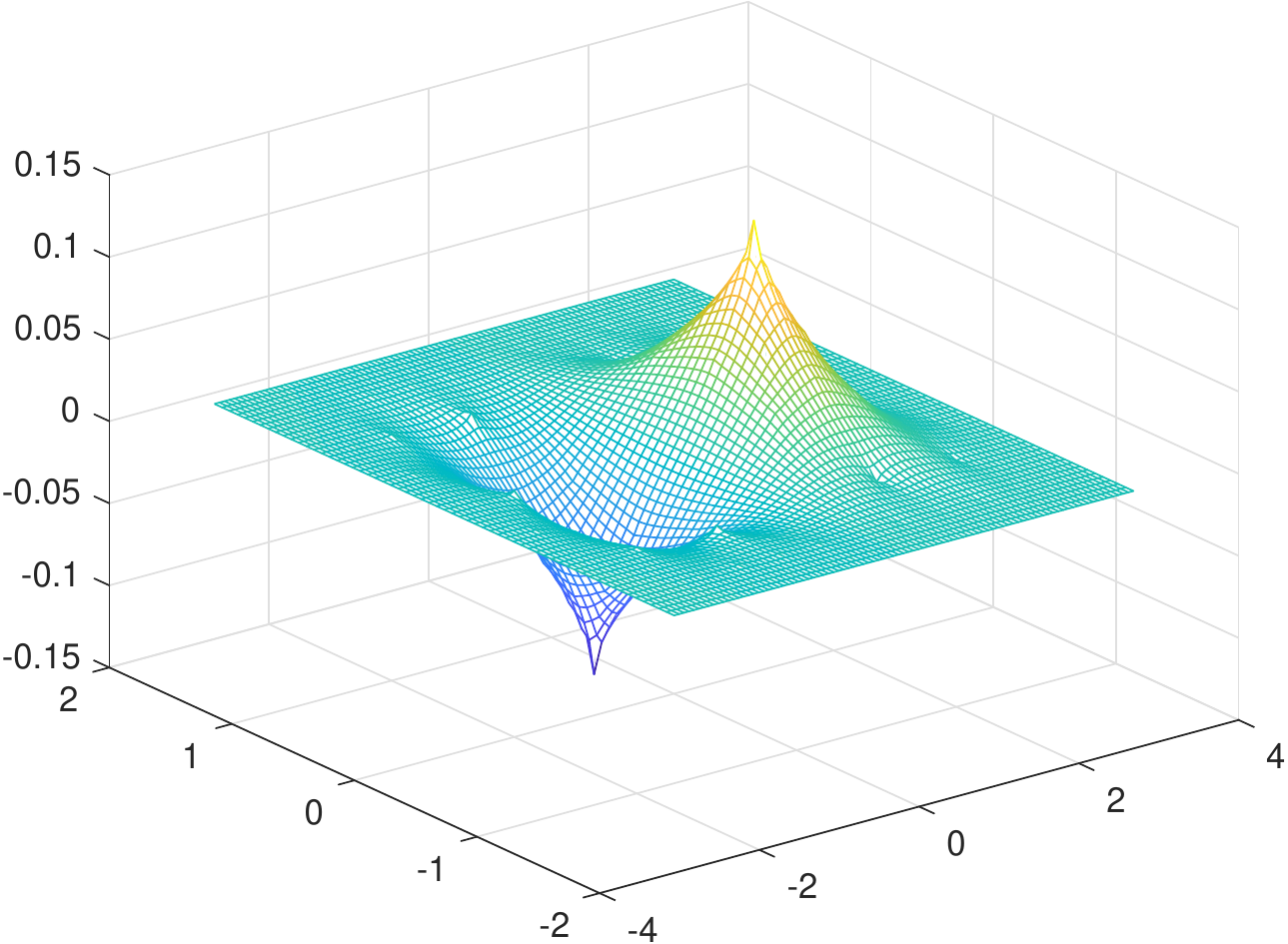}
		\caption{$\psi^5$}
	\end{subfigure}
	\\
	\begin{subfigure}[]{0.31\textwidth}
		\includegraphics[width=\textwidth, height=0.7\textwidth]{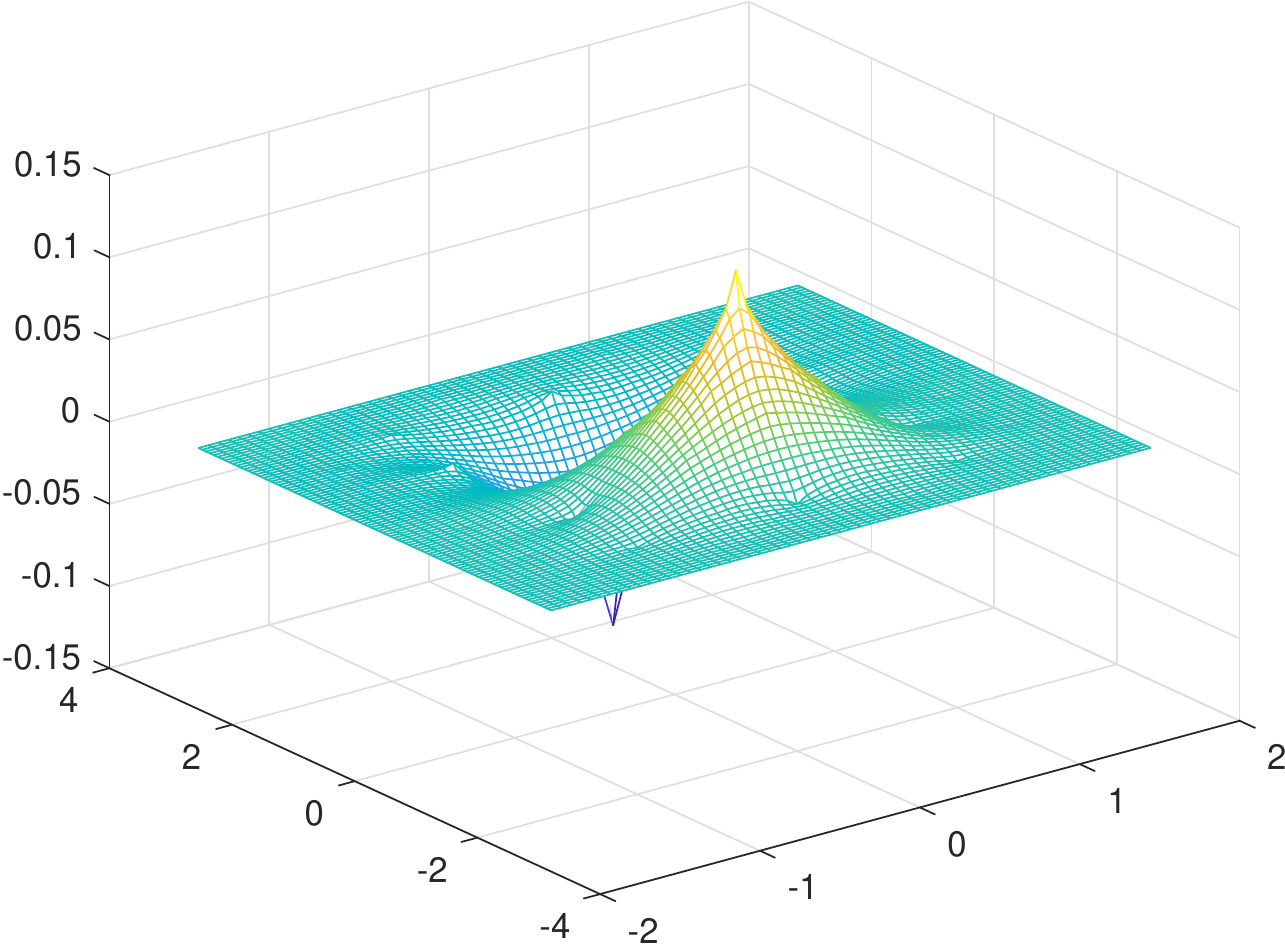}
		\caption{$\psi^6$}
	\end{subfigure}
	\begin{subfigure}[]{0.31\textwidth}
		\includegraphics[width=\textwidth, height=0.7\textwidth]{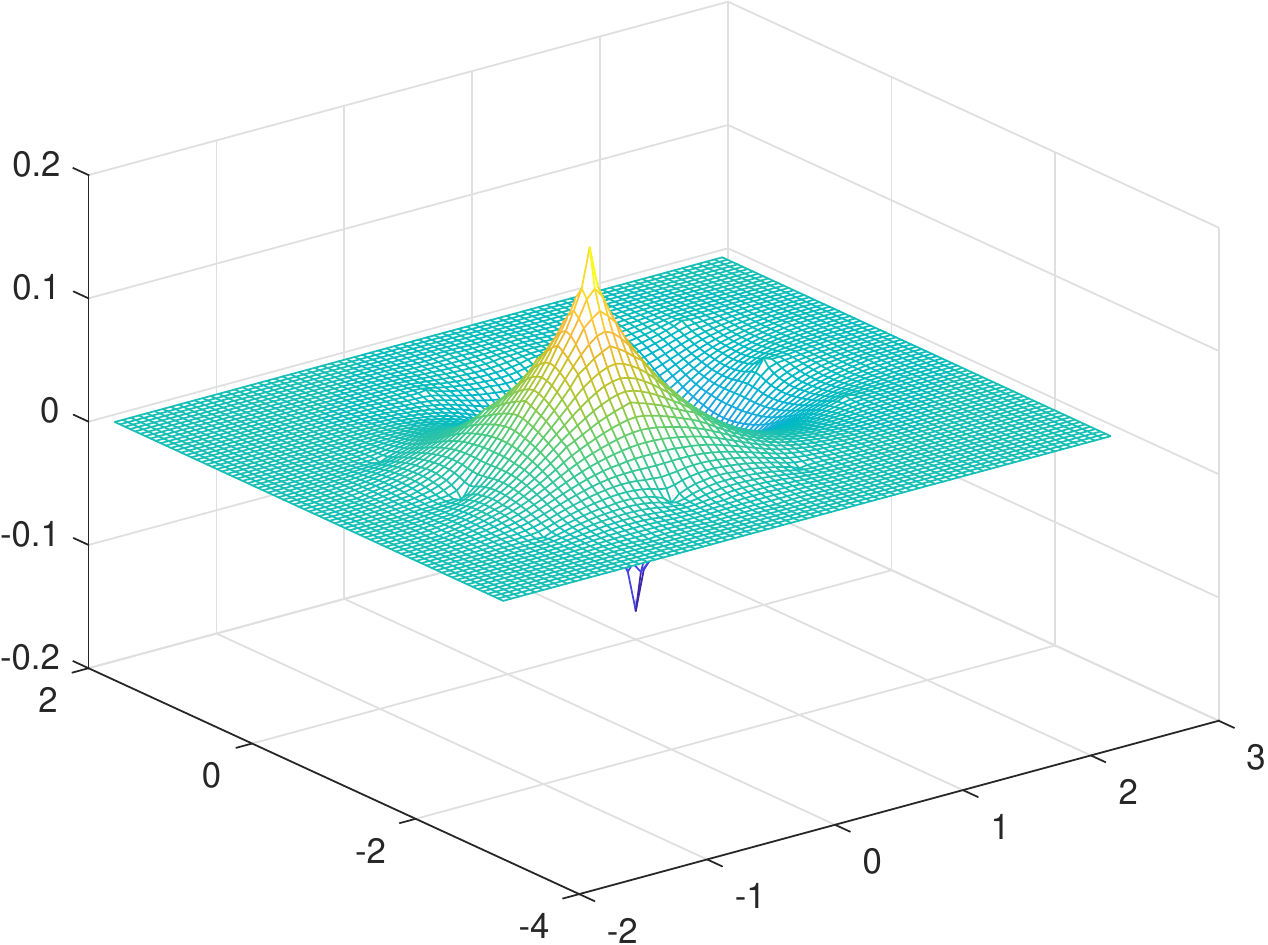}
		\caption{$\psi^7 $ }
	\end{subfigure}
	\begin{subfigure}[]{0.31\textwidth}
		\includegraphics[width=\textwidth, height=0.7\textwidth]{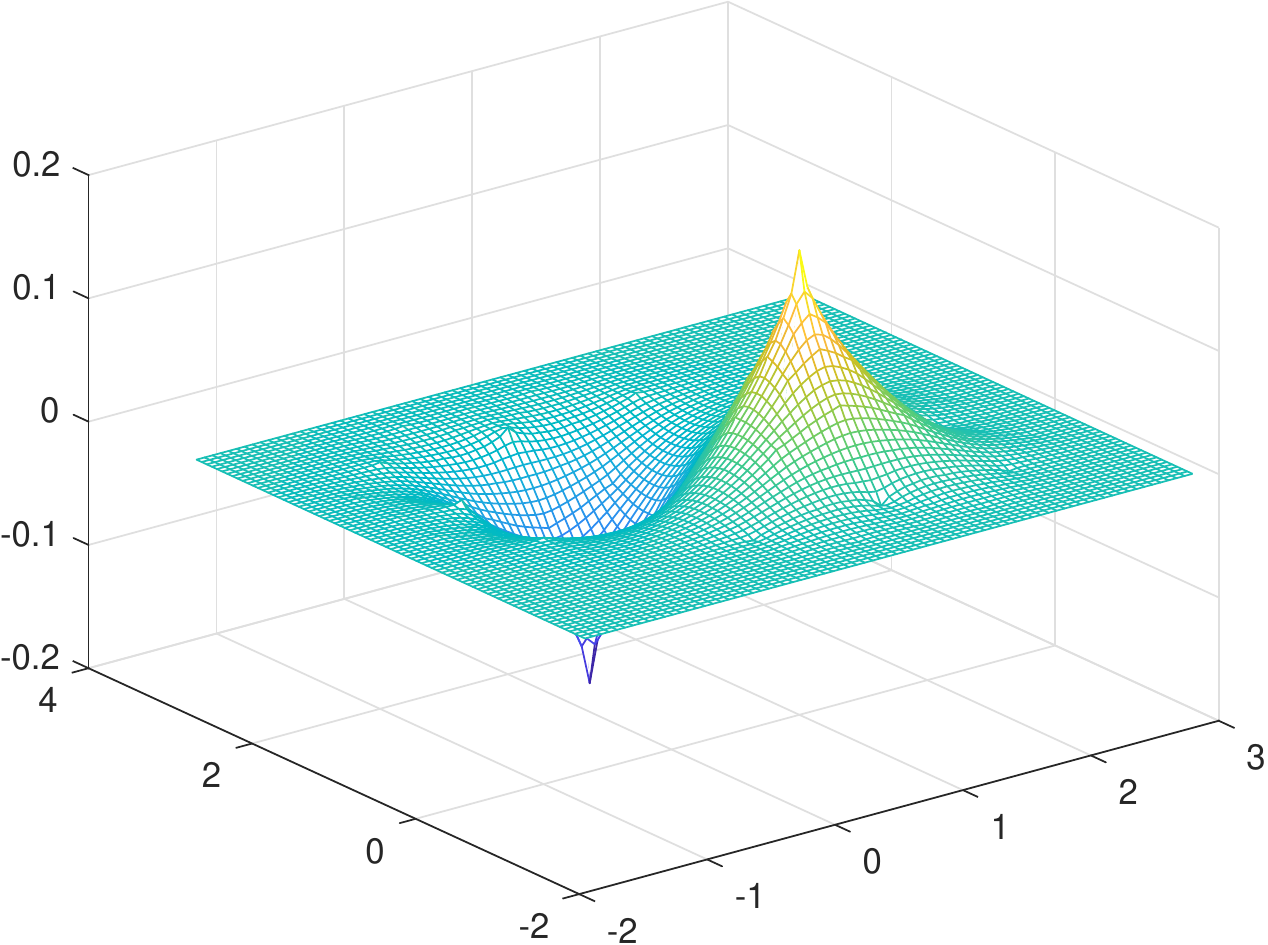}
		\caption{$\psi^8 $ }
	\end{subfigure}
\caption{Graphs of $\phi,\psi^1,\ldots,\psi^8$
in Example~\ref{ex:qcxVM1}, where
$\{\phi;\psi^1,\ldots,\psi^8\}$ is a compactly supported directional tight $\dm_{\sqrt{2}}$-framelet in $L_2(\R^2)$.			 }
\label{fig:qcxVM1}
\end{figure}

\begin{example} \label{ex:sqrt3}{\rm
For $d=2$, we consider the dilation matrix $\dm_{\sqrt{3}}$ and a low-pass filter $a$ as follows:
\[
\dm_{\sqrt{3}}=
\begin{bmatrix} 1 & -2 \\ 2 & -1
\end{bmatrix},\qquad\quad
\ta = \begin{bmatrix}
0 & \tfrac{1}{9} & \tfrac{1}{9} \\[0.5em]
\tfrac{1}{9} & {\bf\underline{\tfrac{1}{3}}} & \tfrac{1}{9} \\[0.5em]
\tfrac{1}{9} & \tfrac{1}{9} & 0
\end{bmatrix}_{[-1, 1]\times [-1, 1]}.
\]
Using Theorem~\ref{thm:qtf}, we have a directional tight $\dm_{\sqrt{3}}$-framelet filter bank
$ \{\ta; \tb_1, \ldots, \tb_{18}\}$, where
{\small
\begin{align*}
&b_1 = \tfrac{\sqrt{3}}{9}\begin{bmatrix}
{\bf\underline{-1}} & 1
\end{bmatrix}_{[0,1]\times[0,0]}, \quad
b_2 = \tfrac{\sqrt{3}}{9}\begin{bmatrix}
-1 & 1
\end{bmatrix}_{[1,2]\times[1,1]}, \quad
b_3 = \tfrac{\sqrt{3}}{9}\begin{bmatrix}
-1 \\ 1
\end{bmatrix}_{[1,1]\times[-1,0]}, \quad
b_4 = \tfrac{\sqrt{3}}{9}\begin{bmatrix}
-1 \\ 1
\end{bmatrix}_{[1,1]\times[1,2]},\\
&b_5 = \tfrac{\sqrt{2}}{9} \begin{bmatrix}
-1 \\ 1
\end{bmatrix}_{[1,1]\times[0,1]},\quad
b_6 = \tfrac{\sqrt{2}}{9} \begin{bmatrix}
-1 &  1
\end{bmatrix}_{[0,1]\times[1,1]}, \quad
b_7 = \tfrac{\sqrt{3}}{9} \begin{bmatrix}
0 & -1 \\ 1 & 0
\end{bmatrix}_{[1,2]\times[0,1]}, \quad
b_8 = \tfrac{\sqrt{3}}{9} \begin{bmatrix}
0 & -1 \\ {\bf\underline{1}} & 0
\end{bmatrix}_{[0,1]\times[0,1]},
\end{align*}
}
{\small
\begin{align*}
&b_9 = \tfrac{\sqrt{2}}{9}\begin{bmatrix}
0 & -1 \\ 1 & 0
\end{bmatrix}_{[1,2]\times[1,2]}, \quad
b_{10} = \tfrac{1}{9}\begin{bmatrix}
{{-1}} & 0 & 1
\end{bmatrix}_{[1,3]\times[1,1]}, \quad
b_{11} = \tfrac{1}{9}\begin{bmatrix}
{{-1}} \\ 0 \\ 1
\end{bmatrix}_{[1,1]\times[1,3]}, \quad
b_{12} = \tfrac{1}{9} \begin{bmatrix}
-1 & 0 \\ 0 & {{1}}
\end{bmatrix}_{[1,2]\times[-1, 0]},\\
&b_{13} = \tfrac{1}{9} \begin{bmatrix}
-1 & 0 \\ 0 & 1
\end{bmatrix}_{[1,2]\times[0, 1]}, \quad
b_{14} = \tfrac{1}{9}\begin{bmatrix}
{\bf\underline{0}} & -1 \\ 0 & 0 \\ 1 & 0
\end{bmatrix}_{[0,1]\times[-2, 0]}, \quad
b_{15} = \tfrac{1}{9} \begin{bmatrix}
0 & -1 \\ {\bf\underline{0}} & 0 \\ 1 & 0
\end{bmatrix}_{[0,1]\times[-1, 1]},\\
&b_{16} = \tfrac{1}{9} \begin{bmatrix}
0 & 0 & -1 \\  1 & 0 & 0
\end{bmatrix}_{[1,3]\times[0, 1]}, \quad
b_{17} = \tfrac{1}{9} \begin{bmatrix}
0 & 0 & -1 \\  1 & 0 & 0
\end{bmatrix}_{[1,3]\times[1, 2]}, \quad
b_{18} = \tfrac{1}{9} \begin{bmatrix}
0 & 0 & -1 \\
0 & {\bf\underline{0}} & 0 \\
1 & 0 & 0
\end{bmatrix}_{[-1,1]\times[-1, 1]}.
\end{align*}
}
Since $\sm(a,\dm_{\sqrt{3}})\approx 1.657138$, $\phi\in L_2(\R^2)$ and $\{\phi;\psi^1,\ldots,\psi^{18}\}$ is a (directional) tight $\dm_{\sqrt{3}}$-framelet in $L_2(\R^2)$, where $\phi,\psi^1,\ldots,\psi^{18}$ are defined in \eqref{phi:psi} with $\dm=\dm_{\sqrt{3}}$ and $s=18$.
}\end{example}

\section{Multivariate Quasi-tight Framelets with High Vanishing Moments}

Multivariate framelets having high vanishing moments are of interest and importance in both theory and applications.
In this section we shall prove Theorem~\ref{thm:qtf:vm} by showing that we can always construct a quasi-tight framelet with the highest possible order of vanishing moments from an arbitrarily given multivariate refinable function. In particular, we provide an algorithm for constructing quasi-tight framelets with high vanishing moments.

To prove Theorem~\ref{thm:qtf:vm}, we need a few auxiliary results and recall some necessary notations.
For $k\in \dZ$ and $u\in \dlp{0}$, the difference operator $\nabla_k u$ is defined to be $\nabla_k u:=u-u(\cdot-k)$. For $\nu=(\nu_1,\ldots,\nu_d)^\tp\in \dNN$, we define $\nabla^\nu:=\nabla_{e_1}^{\nu_1}\cdots\nabla_{e_d}^{\nu_d}$.
Recall that $\td$ is the Dirac sequence such that $\td(0)=1$ and $\td(k)=0$ for all $k\in \dZ\bs\{0\}$. Therefore, for $u\in \dlp{0}$,
\be \label{nabla}
\wh{\nabla^\nu u}(\xi)=
\wh{\nabla^\nu \td}(\xi)\wh{u}(\xi)=
(1-e^{-i\xi_1})^{\nu_1}\cdots(1-e^{-i\xi_d})^{\nu_d}\wh{u}(\xi),\qquad
\xi=(\xi_1,\ldots,\xi_d)^\tp\in \dR.
\ee

The following result is known in \cite[Theorem~3.6]{han03jat} and \cite[Lemma~7.2.2]{hanbook}. For the convenience of the reader, we provide a slightly modified proof from \cite[Lemma~7.2.2]{hanbook} here.

\begin{lemma}\label{lem:nabla}
Let $m\in \N$ and $u=\{u(k)\}_{k\in \dZ}\in \dlp{0}$. Then $u$ has order $m$ vanishing moments (i.e., $\wh{u}(\xi)=\bo(\|\xi\|^m)$ as $\xi \to 0$) if and only if there exist $u_\nu\in \dlp{0}$ for all $\nu\in \dNN$ with $|\nu|=m$ such that
$u=\sum_{\nu\in \dNN, |\nu|=m} \nabla^\nu u_\nu$, that is,
$\wh{u}(\xi)=\sum_{\nu\in \dNN, |\nu|=m} \wh{\nabla^\nu\td}(\xi) \wh{u_\nu}(\xi)$.
Moreover, if the filter $u$ has real coefficients, then all filters $u_\nu$ have real coefficients.
\end{lemma}

\begin{proof}
The sufficiency part is trivial, since $\wh{\nabla^\nu\td}(\xi)=\bo(\|\xi\|^{|\nu|})$ as $\xi \to 0$ (in fact $\vmo(\nabla^\nu \td)=|\nu|$) for all $\nu\in \dNN$. Consequently, it is trivial that $\wh{u}(\xi)=\sum_{\nu\in \dNN, |\nu|=m} \wh{\nabla^\nu\td}(\xi) \wh{u_\nu}(\xi)=\bo(\|\xi\|^m)$ as $\xi \to 0$.

For $r\in \NN$, we define $\Lambda_r:=\{ \mu \in \dNN \setsp |\mu|\le r\}$.
Let $v:=u$.
To prove the necessity part, without loss of generality, by shifting the filter $v$, we can assume that $\mbox{supp}(v)\subseteq \Lambda_r$ but $\mbox{supp}(v)$ is not contained inside $\Lambda_{r-1}$ for some $r\in \NN$. Suppose that $r\ge m$ and $v(k)\ne 0$ for some $k\in \dNN$ with $|k|=r$. Then we can easily write $k=\nu+j$ with $\nu,j\in \dNN$ and $|\nu|=m$. We replace/update $v$ by $v-v(k)(-1)^m [\nabla^\nu\td](\cdot-j)$.
Since $\mbox{supp}([\nabla^\nu \td](\cdot-j))\subseteq \Lambda_{r-1}\cup \{k\}$, we conclude that the updated filter $v$ is still supported inside $\Lambda_r$, $v(k)=0$, and the updated filter $v$ preserves the values as the previous filter at the set $\{ n\in \dNN \setsp |n|=r, n\ne k\}$.
Therefore, we can continue this procedure for other $n\in \dNN$ with $|n|=r$ so that finally the  updated
filter $v$ has support inside $\Lambda_{r-1}$. We can continue this procedure until $r<m$.
Note that the Fourier series of $v(k) [\nabla^\nu\td](\cdot-j)$ is simply $v(k) e^{-ij\cdot \xi}\wh{\nabla^\nu\td}(\xi)$.

Consequently, we can write $\wh{u}=\wh{v}+\sum_{\nu\in \dNN, |v|=m} \wh{\nabla^\nu \td}(\xi)\wh{u_\nu}(\xi)$ for some sequences $u_\nu \in \dlp{0}, \nu\in \dNN$ with $|\nu|=m$ and $\supp(v)\subseteq \Lambda_{m-1}$.
Since $\wh{u}(\xi)=\bo(\|\xi\|^m)$ as $\xi\to 0$, we trivially have $\wh{v}(\xi)=\bo(\|\xi\|^m)$ as $\xi \to 0$, that is, the filter $v$ has order $m$ vanishing moments. Hence,
\[
\sum_{k\in \Lambda_{m-1}} v(k) k^\mu=0, \qquad \forall\, \mu\in \Lambda_{m-1}.
\]
It is well known that the above system of linear equations can only have the trivial solution: $v(k)=0$ for all $k\in \Lambda_{m-1}$. Therefore, $v=0$ and we proved the claim.
\end{proof}

Using Lemma~\ref{lem:nabla}, we have the following result, which improves Corollary~\ref{cor:qsos}.

\begin{lemma}\label{lem:vm}
Let $m\in \NN$ and $u\in \dlp{0}$ with real coefficients. Then
\be \label{u:cond}
\ol{\wh{u}(\xi)}=\wh{u}(\xi) \quad \mbox{and}\quad \wh{u}(\xi)=\bo(\|\xi\|^{2m}),\quad \xi \to 0
\ee
if and only if there exist
$\eps_1,\ldots,\eps_s\in \{-1,1\}$ and $u_1,\ldots, u_s \in \dlp{0}$ with real coefficients satisfying
\be \label{u:qsos}
\wh{u}(\xi)=\sum_{\ell=1}^s \eps_\ell |\wh{u_\ell}(\xi)|^2 \quad \mbox{with}\quad
\wh{u_\ell}(\xi)=\bo(\|\xi\|^m),\quad \xi\to 0,\forall\, \ell=1,\ldots,s.
\ee
\end{lemma}

\begin{proof}
The sufficiency part is trivial and we only need to prove the necessity part. Suppose that \eqref{u:cond} holds. By Lemma~\ref{lem:nabla}, there exist $u_\nu\in \dlp{0}, |\nu|=2m$ with real coefficients such that $\wh{u}(\xi)=\sum_{\nu\in \dNN,|\nu|=2m} 2\wh{\nabla^\nu \td}(\xi) \wh{u_\nu}(\xi)$.
Since $\ol{\wh{u}(\xi)}=\wh{u}(\xi)$, we conclude that
\be \label{u:nabla:2}
\wh{u}(\xi)=\sum_{\nu\in \dNN, |\nu|=2m} \theta_\nu(\xi)
\quad \mbox{with}\quad \theta_\nu(\xi):= \wh{\nabla^\nu \td}(\xi) \wh{u_\nu}(\xi)+\ol{\wh{\nabla^\nu \td}(\xi)}\,\ol{ \wh{u_\nu}(\xi)}.
\ee
For $\nu\in \dNN$ with $|\nu|=2m$, we consider two cases.

Case 1. $\nu\in \dNN$ with $|\nu|=2m$ but $\nu\not \in 2\dNN$. Then there exist $\alpha,\beta\in \dNN$ such that $\alpha+\beta=\nu$ and $|\alpha|=|\beta|=m$. Define
\[
\wh{u_{\nu,\alpha,\beta}}(\xi):=
\ol{\wh{\nabla^\alpha\td}(\xi)}+\wh{\nabla^\beta\td}(\xi)\wh{u_\nu}(\xi).
\]
By $|\alpha|=|\beta|=m$, we see that
$\ol{\wh{\nabla^\alpha\td}(\xi)}=\bo(\|\xi\|^m)$ and $\wh{\nabla^\beta\td}(\xi)\wh{u_\nu}(\xi)=\bo(\|\xi\|^m)$ as $\xi \to 0$. Consequently, we have $\wh{u_{\nu,\alpha,\beta}}(\xi)=\bo(\|\xi\|^m)$ as $\xi \to 0$.
By calculation, we have
\[
|\wh{u_{\nu,\alpha,\beta}}(\xi)|^2
=|\wh{\nabla^\alpha\td}(\xi)|^2
+|\wh{\nabla^\beta\td}(\xi)\wh{u_\nu}(\xi)|^2
+\Big(
\wh{\nabla^\alpha\td}(\xi)\wh{\nabla^\beta\td}(\xi)
\wh{u_\nu}(\xi)+
\ol{\wh{\nabla^\alpha\td}(\xi)\wh{\nabla^\beta\td}(\xi)
\wh{u_\nu}(\xi)}\Big).
\]
Since $\alpha+\beta=\nu$, the last term in the above identity is simply $\theta_\nu(\xi)$. Consequently, we have
\[
\theta_\nu(\xi)=|\wh{u_{\nu,\alpha,\beta}}(\xi)|^2
-|\wh{\nabla^\alpha\td}(\xi)|^2
-|\wh{\nabla^\beta\td}(\xi)\wh{u_\nu}(\xi)|^2.
\]
To reduce the total number of filters, we can instead combine the last two terms in the above identity with the terms in Case~2 discussed below.

Case 2: $\nu\in 2\dNN$, that is, $\nu=2\mu$ for some $\mu\in \dNN$ with $|\mu|=m$. By $\wh{\nabla^{2\mu}\td}(\xi)=|\wh{\nabla^\mu\td}(\xi)|^2
(-1)^{|\mu|}e^{-i\mu\cdot \xi}$, we deduce that
\[
\theta_\nu(\xi)=\theta_{2\mu}(\xi)=
|\wh{\nabla^\mu\td}(\xi)|^2 \eta_\mu(\xi)\quad \mbox{with}\quad
\eta_\mu(\xi):=
(-1)^{|\mu|} \Big(e^{-i\mu\cdot \xi}\wh{u_{2\mu}}(\xi)+
e^{i\mu\cdot \xi}\ol{\wh{u_{2\mu}}(\xi)}\Big).
\]
Note that $\ol{\eta_{\mu}(\xi)}=\eta_{\mu}(\xi)$ and $\eta_\mu$ has real coefficients. By Corollary~\ref{cor:qsos}, there exist $\eps_0,\ldots,\eps_s\in \{-1,1\}$ and $2\pi\dZ$-periodic trigonometric polynomials $\fu_0,\ldots,\fu_s$ with real coefficients such that $\eta_\mu(\xi)=\eps_0|\fu_0(\xi)|^2+\cdots+\eps_s|\fu_s(\xi)|^2$.
That is, we proved
\[
\theta_{\nu}(\xi)=
\eps_0 |\wh{u_{\mu,0}}(\xi)|^2+\cdots
+\eps_s |\wh{u_{\mu,s}}(\xi)|^2\quad
\mbox{with}\quad
\wh{u_{\mu,\ell}}(\xi):=\wh{\nabla^\mu\td}(\xi)
\fu_\ell(\xi),\qquad \ell=0,\ldots,s.
\]
Note that $\wh{u_{\mu,\ell}}(\xi)=\bo(\|\xi\|^m)$ as $\xi \to 0$ for all $\ell=0,\ldots,s$ by $|\mu|=m$.
Now the conclusion follows trivially from \eqref{u:nabla:2}.
\end{proof}

We now prove Theorem~\ref{thm:qtf:vm}.
Define $\cA(\xi):=1-\ddm\sum_{j=1}^{\ddm}
|\wh{a^{[\gamma_j]}}(\xi)|^2$.
Obviously, $\ol{\cA(\xi)}=\cA(\xi)$.
By the definition of the coset sequences, we observe that $\cA(\dm^\tp\xi)=1-\sum_{\omega\in \Omega_{\dm}} |\wh{a}(\xi+2\pi \omega)|^2$.
By the definition $m=\min(\sr(a,\dm),\frac{1}{2}\vmo(u_a))$ with $\wh{u_a}(\xi):=1-|\wh{a}(\xi)|^2$, we have $\sr(a,\dm)\ge m$ and $\vmo(u_a)\ge 2m$.
Consequently, we have
$|\wh{a}(\xi+2\pi \omega)|^2=\bo(\|\xi\|^{2m})$ as $\xi \to 0$ for all $\omega\in \Omega_{\dm}\bs\{0\}$, and $1-|\wh{a}(\xi)|^2=\wh{u_a}(\xi)=\bo(\|\xi\|^{2m})$ as $\xi \to 0$. That is, we must have
\[
\wh{\cA}(\dm^\tp \xi)=
1-|\wh{a}(\xi)|^2-\sum_{\omega\in \Omega_{\dm}\bs\{0\}} |\wh{a}(\xi+2\pi \omega)|^2=
\bo(\|\xi\|^{2m}),\qquad \xi\to 0.
\]
Since $\dm$ is an invertible matrix, consequently we must have $\cA(\xi)=\bo(\|\xi\|^{2m})$ as $\xi \to 0$.
By Lemma~\ref{lem:vm}, since $\cA$ has real coefficients, there exist $\eps_1,\ldots,\eps_t\in \{-1,1\}$ and $u_1,\ldots,u_t\in \dlp{0}$ with real coefficients such that
\be \label{A:qsos}
\cA(\xi)=\sum_{\ell=1}^t \eps_\ell |\wh{u_\ell}(\xi)|^2 \quad \mbox{with}\quad
\wh{u_\ell}(\xi)=\bo(\|\xi\|^m),\quad \xi\to 0,\forall\, \ell=1,\ldots,t.
\ee
Employing a similar idea as in \cite{ls06}, we now define the high-pass filters $b_1,\ldots, b_s\in \dlp{0}$ with $s:=t+\ddm$ as follows:
\be \label{qtffb:sos:1}
\wh{b_{\ell}}(\xi):=\wh{a}(\xi) \wh{u_\ell}(\dm^\tp\xi),\qquad \ell=1,\ldots,t
\ee
and
\be \label{qtffb:sos:2}
\wh{b_{t+j}}(\xi):=\ddm^{-1/2} e^{-i\gamma_j\cdot\xi}-\ddm^{1/2} \wh{a}(\xi)\ol{\wh{a^{[\gamma_j]}}(\dm^\tp\xi)},\quad
j=1,\ldots,\ddm.
\ee
Define $\eps_{t+1}=\cdots=\eps_{t+\ddm}:=1$.
We show that
$\{a; b_1,\ldots,b_s\}_{(\eps_1,\ldots,\eps_s)}$ is a quasi-tight $\dm$-framelet filter bank and
$\vmo(b_\ell)\ge m$ for all $\ell=1,\ldots, s$.
Let $B_\ell$ be defined as in \eqref{B:ell}.
We can calculate $B_\ell$ for the high-pass filters $b_\ell$ defined in \eqref{qtffb:sos:1} and \eqref{qtffb:sos:2}.
Let $b_\ell$ be defined in \eqref{qtffb:sos:1}. Then
$\wh{b_\ell^{[\gamma]}}(\xi)=\wh{a^{[\gamma]}}(\xi)\wh{u_\ell}(\xi)$ for $\gamma\in \Gamma_{\dm}$.
Therefore, by \eqref{A:qsos}, we have
\be\label{Bell1}
\begin{split}
\sum_{\ell=1}^t B_\ell(\xi)
&=
\sum_{\ell=1}^t \eps_\ell  |\wh{u_\ell}(\xi)|^2\ddm \Big[\wh{a^{[\gamma_1]}}(\xi),\ldots, \wh{a^{[\gamma_{\ddm}]}}(\xi)\Big]^\star
\Big[\wh{a^{[\gamma_1]}}(\xi),\ldots, \wh{a^{[\gamma_{\ddm}]}}(\xi)\Big]\\
&=
\ddm \cA(\xi) \Big[\wh{a^{[\gamma_1]}}(\xi),\ldots, \wh{a^{[\gamma_{\ddm}]}}(\xi)\Big]^\star
\Big[\wh{a^{[\gamma_1]}}(\xi),\ldots, \wh{a^{[\gamma_{\ddm}]}}(\xi)\Big].
\end{split}
\ee
Let $b_\ell$ be defined in \eqref{qtffb:sos:2} with $\ell=t+j$. Then
$\wh{b_\ell^{[\gamma]}}(\xi)=
\ddm^{-1/2}\td(\gamma-\gamma_j)-
\ddm^{1/2}
\wh{a^{[\gamma]}}(\xi)
\ol{\wh{a^{[\gamma_j]}}(\xi)}$ for $\gamma\in \Gamma_{\dm}$.
Hence,
\[
\Big[\wh{b^{[\gamma_1]}_\ell}(\xi),\ldots, \wh{b^{[\gamma_{\ddm}]}_\ell}(\xi)\Big]
=\ddm^{-1/2}e_j^\tp-\ddm^{1/2} \ol{\wh{a^{[\gamma_j]}}(\xi)} \Big[\wh{a^{[\gamma_1]}}(\xi),\ldots, \wh{a^{[\gamma_{\ddm}]}}(\xi)\Big].
\]
Therefore, by $\eps_{t+1}=\cdots=\eps_{t+\ddm}=1$,
\begin{align*}
\sum_{\ell=t+1}^{t+\ddm}
B_\ell(\xi)
&=\sum_{j=1}^{\ddm}
\left(e_j e_j^\tp-\ddm e_j \ol{\wh{a^{[\gamma_j]}}(\xi)} \Big[\wh{a^{[\gamma_1]}}(\xi),\ldots, \wh{a^{[\gamma_{\ddm}]}}(\xi)\Big]
-\ddm \Big[\wh{a^{[\gamma_1]}}(\xi),\ldots, \wh{a^{[\gamma_{\ddm}]}}(\xi)\Big]^\star \wh{a^{[\gamma_j]}}(\xi)e_j^\tp \right.\\
&\qquad \qquad\qquad \left.
+\ddm^2 |\wh{a^{[\gamma_j]}}(\xi)|^2 \Big[\wh{a^{[\gamma_1]}}(\xi),\ldots, \wh{a^{[\gamma_{\ddm}]}}(\xi)\Big]^\star
\Big[\wh{a^{[\gamma_1]}}(\xi),\ldots, \wh{a^{[\gamma_{\ddm}]}}(\xi)\Big]\right)\\
&=I_{\ddm}+\left(-2\ddm+\ddm^2\sum_{j=1}^{\ddm}
|\wh{a^{[\gamma_j]}}(\xi)|^2\right)
\Big[\wh{a^{[\gamma_1]}}(\xi),\ldots, \wh{a^{[\gamma_{\ddm}]}}(\xi)\Big]^\star
\Big[\wh{a^{[\gamma_1]}}(\xi),\ldots, \wh{a^{[\gamma_{\ddm}]}}(\xi)\Big],
\end{align*}
where $I_{\ddm}$ stands for the $\ddm\times \ddm$ identity matrix.
Since $\cA(\xi)=1-\ddm\sum_{j=0}^{\ddm-1}
|\wh{a^{[\gamma_j]}}(\xi)|^2$, we have
\[
-2\ddm +\ddm^2\sum_{j=1}^{\ddm}
|\wh{a^{[\gamma_j]}}(\xi)|^2
=-2\ddm+\ddm(1-\cA(\xi))=
-\ddm(1+\cA(\xi)).
\]
In other words, we obtain
\[
\sum_{\ell=t+1}^{t+\ddm}
B_\ell(\xi)=
I_{\ddm}-\ddm (1+\cA(\xi))
\Big[\wh{a^{[\gamma_1]}}(\xi),\ldots, \wh{a^{[\gamma_{\ddm}]}}(\xi)\Big]^\star
\Big[\wh{a^{[\gamma_1]}}(\xi),\ldots, \wh{a^{[\gamma_{\ddm}]}}(\xi)\Big].
\]
Combining the above identity with \eqref{Bell1} and noting that $s=t+\ddm$, we have
\[
\sum_{\ell=1}^s B_\ell(\xi)=
I_{\ddm}-\ddm \Big [\wh{a^{[\gamma_1]}}(\xi),\ldots, \wh{a^{[\gamma_{\ddm}]}}(\xi)\Big]^\star
\Big[\wh{a^{[\gamma_1]}}(\xi),\ldots, \wh{a^{[\gamma_{\ddm}]}}(\xi)\Big]
=\cN_a(\xi).
\]
That is, we verified the identity in \eqref{qtffb:2}. This proves that
$\{a; b_1,\ldots,b_s\}_{(\eps_1,\ldots,\eps_s)}$ is a quasi-tight $\dm$-framelet filter bank.

We now prove that $\wh{b_\ell}(\xi)=\bo(\|\xi\|^m)$ as $\xi \to 0$ for all $\ell=1,\ldots,s$.
By the second identity in \eqref{A:qsos} and the definition of $b_\ell$ in \eqref{qtffb:sos:1}, we trivially have $\wh{b_\ell}(\xi)=\bo(\|\xi\|^m)$ as $\xi \to 0$ for all $\ell=1,\ldots,t$.

Let $\{\omega_1,\ldots,\omega_{\ddm}\}=\Omega_{\dm}$ with $\omega_1:=0$.
By the identity in \eqref{coset:U} with $E(\xi)$ and $U$ being defined in \eqref{EU}, we have
\be \label{a:acoset}
\Big[\wh{a}(\xi+2\pi \omega_1),\ldots,\wh{a}(\xi+2\pi \omega_{\ddm})\Big]=\Big[
e^{-i\gamma_1\cdot \xi} \wh{a^{[\gamma_1]}}(\dm^\tp\xi),
\ldots,e^{-i\gamma_{\ddm}\cdot \xi} \wh{a^{[\gamma_{\ddm}]}}(\dm^\tp \xi)\Big]  U
\ee
and $U U^\star=\ddm I_{\ddm}$.
Note that $\wh{a}(0)=1$ and $\omega_1=0$. Since $\sr(a,\dm)\ge m$, by the definition of sum rules in \eqref{sr}, we deduce from the above identity in \eqref{a:acoset} that
\[
\Big[
e^{-i\gamma_1\cdot \xi} \wh{a^{[\gamma_1]}}(\dm^\tp\xi),
\ldots,e^{-i\gamma_{\ddm}\cdot \xi} \wh{a^{[\gamma_{\ddm}]}}(\dm^\tp \xi)\Big]=
\Big[\wh{a}(\xi),\bo(\|\xi\|^m),\ldots,
\bo(\|\xi\|^m)\Big]U^{-1},\quad \xi \to 0.
\]
Since $U^{-1}=\ddm^{-1}U^\star=(\ddm^{-1} e^{-i2\pi \gamma_j\cdot \omega_k})_{1\le k,j\le \ddm}$ and $\omega_1=0$, all the entries in the first row of $U^{-1}$ are $\ddm^{-1}$. Therefore, we conclude from the above identity that
\[
\Big[
e^{-i\gamma_1\cdot \xi} \wh{a^{[\gamma_1]}}(\dm^\tp\xi),
\ldots,e^{-i\gamma_{\ddm}\cdot \xi} \wh{a^{[\gamma_{\ddm}]}}(\dm^\tp \xi)\Big]=
\ddm^{-1} \wh{a}(\xi)[1,\ldots,1]+\bo(\|\xi\|^m),\quad \xi \to 0.
\]
That is, we proved
\begin{equation} \label{eq:sr:coset}
e^{-i\gamma_j\cdot \xi} \wh{a^{[\gamma_j]}}(\dm^\tp \xi)=\ddm^{-1} \wh{a}(\xi)+\bo(\|\xi\|^m),\qquad \xi \to 0, j=1,\ldots,\ddm.
\end{equation}
For $b_\ell$ defined in \eqref{qtffb:sos:2} with $\ell=t+j$, we deduce from the above identity that
\begin{align*}
\wh{b_\ell}(\xi)
&=\ddm^{-1/2}e^{-i\gamma_j\cdot \xi}\ol{(1-\ddm \ol{\wh{a}(\xi)}
e^{-i\gamma_j\cdot \xi}\wh{a^{[\gamma_j]}}(\dm^\tp \xi))}\\
&=
\ddm^{-1/2}e^{-i\gamma_j\cdot \xi}(1-|\wh{a}(\xi)|^2)
+\bo(\|\xi\|^m)\\
&=
\ddm^{-1/2}e^{-i\gamma_j\cdot \xi}\bo(\|\xi\|^{2m})
+\bo(\|\xi\|^m)=\bo(\|\xi\|^m)
\end{align*}
as $\xi \to 0$, where we used our assumption $1-|\wh{a}(\xi)|^2=\wh{u_a}(\xi)=\bo(\|\xi\|^{2m})$ as $\xi \to 0$. This proves that all the high-pass filters $b_1,\ldots,b_s$ have at least order $m$ vanishing moments.

The above constructed quasi-tight $\dm$-framelet filter banks $\{a; b_1,\ldots,b_s\}_{(\eps_1,\ldots,\eps_s)}$ through \eqref{qtffb:sos:1} and \eqref{qtffb:sos:2} have several shortcomings.
First of all, the supports of all the high-pass filters $b_1,\ldots,b_s$ are much larger than that of the low-pass filter $a$.
Secondly, the number $t$ (and hence $s:=t+\ddm$) of generators could be very large even for low-pass filters with short support.
These considerations
motivate us to propose an algorithm solving only linear equations for constructing quasi-tight framelet filter banks with high vanishing moments. To do so, let us introduce some notations.
For $\mu=(\mu_1,\ldots,\mu_d)^\tp,\nu=(\nu_1,\ldots,\nu_d)^\tp\in \dNN$, we say that $\mu<\nu$ if either $|\mu|<|\nu|$ or $|\mu|=|\nu|$ and $\mu_j=\nu_j$ for $j=1,\ldots,\ell-1$ but $\mu_\ell<\nu_\ell$ for some $1\le \ell \le d$.
Note that $\Omega_\dm$ is a complete set of representatives of distinct cosets of the quotient group $[(\dm^\tp)^{-1}\dZ]/\dZ$. Therefore, $\Omega_\dm$ can be regarded as an additive group under modulo $\dZ$.

Before introducing our algorithm to prove Theorem~\ref{thm:qtf:vm}, we need the following auxiliary result.

\begin{lemma}\label{lem:bM}
Let $b\in \dlp{0}$ and $\beta\in \Omega_{\dm}$.
Define $F(\xi):=(e^{-i \gamma_j\cdot(\xi+2\pi \omega)})_{1\le j\le \ddm, \omega\in \Omega_{\dm}}$ and the $\ddm\times \ddm$ matrix $D_{b,\beta}$ by
\be \label{Dbeta}
[D_{b,\beta}(\xi)]_{\omega,\eta}:=
\begin{cases}
\wh{b}(\xi+2\pi \omega), &\text{if $\omega+\beta-\eta\in \dZ$,}\\
0, &\text{if $\omega+\beta-\eta\not \in \dZ$},\end{cases}
\qquad \omega,\eta\in \Omega_{\dm}.
\ee
Then
\be \label{bMbeta}
F(\xi)D_{b,\beta}(\xi) F^\star(\xi)=\ddm
E_{b,\beta}(\dm^\tp \xi)\quad \mbox{with}\quad
E_{b,\beta}(\xi):= \Big(\wh{b^{[\gamma_k-\gamma_j]}}(\xi) e^{i\gamma_k \cdot 2\pi \beta}\Big)_{1\le j,k\le \ddm}.
\ee
\end{lemma}

\begin{proof} Let $j,k=1,\ldots,\ddm$. We now compute the $(j,k)$-entry of the matrix on the left-hand side of \eqref{bMbeta}. Note that
$\wh{b}(\xi)=\sum_{p=1}^{\ddm} \wh{b^{[\gamma_p]}}(\dm^\tp\xi) e^{-i\gamma_p\cdot \xi}$. Then
\begin{align*}
[F(\xi)D_{b,\beta}(\xi) F^\star(\xi)]_{j,k}
&=\sum_{\omega\in \Omega_{\dm}} [F(\xi)]_{j,\omega}[D_{b,\beta}(\xi)]_{\omega,\omega+\beta}
[F^\star(\xi)]_{\omega+\beta,k}\\
&=\sum_{\omega\in \Omega_{\dm}} e^{-i\gamma_j\cdot (\xi+2\pi \omega)} \wh{b}(\xi+2\pi \omega)
e^{i\gamma_k\cdot(\xi+2\pi \omega+2\pi \beta)}\\
&=
\sum_{\omega\in \Omega_{\dm}} \sum_{p=1}^{\ddm} e^{-i\gamma_j\cdot (\xi+2\pi \omega)} \wh{b^{[\gamma_p]}}(\dm^\tp\xi) e^{-i\gamma_p\cdot(\xi+2\pi \omega)}
e^{i\gamma_k\cdot(\xi+2\pi \omega+2\pi \beta)}\\
&= \sum_{p=1}^{\ddm} \wh{b^{[\gamma_p]}}(\dm^\tp\xi)
e^{-i(\gamma_p+\gamma_j-\gamma_k)\cdot \xi}e^{i\gamma_k\cdot 2\pi \beta}
\sum_{\omega\in \Omega_{\dm}}
e^{-i(\gamma_p+\gamma_j-\gamma_k)\cdot 2\pi \omega}.
\end{align*}
Note that the last sum in the above identity is equal to $\ddm$ if $\gamma_p+\gamma_j-\gamma_k\in \dm \dZ$ and $0$ otherwise. Hence, we deduce from the above identity that
\[
[F(\xi)D_{b,\beta}(\xi) F^\star(\xi)]_{j,k}
=\ddm \wh{b^{[\gamma_p]}}(\dm^\tp\xi)
e^{-i\dm \alpha_{j,k}\cdot \xi}
e^{i\gamma_k\cdot 2\pi \beta}=
\ddm\wh{b^{[\gamma_k-\gamma_j]}}(\dm^\tp\xi)
e^{i\gamma_k\cdot 2\pi \beta},
\]
where we used the identity $\wh{u^{[\gamma+\dm \alpha]}}(\xi)=\wh{u^{[\gamma]}}(\xi)e^{i\alpha\cdot \xi}$ and
$\alpha_{j,k}\in \dZ$ is the unique integer such that $\gamma_p=\gamma_k-\gamma_j+\dm \alpha_{j,k}$ for the unique $\gamma_p \in \Gamma_{\dm}$ satisfying $\gamma_p+\gamma_j-\gamma_k\in \dm \dZ$.
This proves \eqref{bMbeta}.
\end{proof}

Now we are ready to state a constructive algorithm for constructing quasi-tight framelet filter banks with vanishing moments by solving only linear equations. For the convenience of the reader, we state the following result in an algorithmic way.
We now prove Theorem~\ref{thm:qtf:vm} as follows.

\begin{theorem}\label{thm:qtf:vm:linear}
Let $\dm$ be a $d\times d$ dilation matrix and let $a\in \dlp{0}$ be a finitely supported real-valued sequence on $\dZ$.
Let $m\in \NN$ such that $m\le \min(\sr(a,\dm),\frac{1}{2}\vmo(u_a))$, where $\wh{u_a}(\xi):=1-|\wh{a}(\xi)|^2$.
Define $F(\xi):=(e^{-i \gamma_j\cdot(\xi+2\pi \omega)})_{1\le j\le \ddm,\omega\in \Omega_{\dm}}$. Define $\cN_a$ as in \eqref{Na} and $E_{\mu}(\xi):=E_{\nabla^\mu\td,0}(\xi)$ as in \eqref{bMbeta}.

\begin{enumerate}
\item[(S1)] Solve the system $X$ of linear equations induced by
\be \label{Na:factorize}
\ddm^{-1} \cN_a(\xi)=
 \sum_{|\mu|=m} E^\star_\mu(\xi)A_{\mu,\mu}(\xi) E_\mu(\xi)+\sum_{\mu<\nu, |\mu|=|\nu|=m}
 \Big(E^\star_\mu(\xi)A_{\mu,\nu}(\xi) E_\nu(\xi)+E^\star_\nu(\xi)A^\star_{\mu,\nu}(\xi) E_\mu(\xi)\Big)
\ee
and
\be \label{Amu}
A^\star_{\mu,\mu}(\xi)=A_{\mu,\mu}(\xi),\qquad \forall\, \mu\in \dNN, |\mu|=m,
\ee
for the coefficients in all the entries of the matrices $A_{\mu,\nu}$ of $2\pi\dZ$-periodic trigonometric polynomials with $|\mu|=|\nu|=m$ and $\mu\le \nu$.
The linear system $X$ always has a solution of $A_{\mu,\nu}$ with real coefficients for $|\mu|=|\nu|=m$ and $\mu\le \nu$, as long as the supports of their coefficients are large enough.

\item[(S2)] For every $(\mu,\nu)$ with $|\mu|=|\nu|=m$ and $\mu<\nu$, factorize $A_{\mu,\nu}(\xi)=A^\star_{\mu,\nu,1}(\xi)A_{\mu,\nu,2}(\xi)$ (e.g.,  $A_{\mu,\nu,1}(\xi)=I_{\ddm}$ and $A_{\mu,\nu,2}=A_{\mu,\nu}(\xi)$) for some $\ddm\times \ddm$ matrices $A_{\mu,\nu,1}$ and $A_{\mu,\nu,2}$ of $2\pi\dZ$-periodic trigonometric polynomials with real coefficients.
    Define
\be \label{bmunu}
[\wh{b_{\mu,\nu,1}}(\xi),\ldots,\wh{b_{\mu,\nu,\ddm}}(\xi)]^\tp:=
A_{\mu,\nu,1}(\dm^\tp \xi) F(\xi) e_1 \wh{\nabla^\mu \td}(\xi)+
A_{\mu,\nu,2}(\dm^\tp \xi) F(\xi) e_1 \wh{\nabla^\nu \td}(\xi)
\ee
and $s_{\mu,\nu}:=\ddm$ and $\eps_{\mu,\nu,\ell}:=1$ for all $\ell=1,\ldots,\ddm$. Replace/update $A_{\mu,\mu}$ and $A_{\nu,\nu}$ by $A_{\mu,\mu}-A_{\mu,\nu,1}^\star A_{\mu,\nu,1}$ and $A_{\nu,\nu}-A_{\mu,\nu,2}^\star A_{\mu,\nu,2}$, respectively.

\item[(S3)] For every $\mu\in \dNN$ with $|\mu|=m$, apply Theorem~\ref{thm:qsos} to the updated $A_{\mu,\mu}$ so that $A_{\mu,\mu}(\xi)=\sum_{\ell=1}^{s_{\mu,\mu}}
    \eps_{\mu,\mu,\ell} \fu_\ell^\star(\xi) \fu_\ell(\xi)$, where $\eps_{\mu,\mu,\ell}\in \{-1,1\}$ and $\fu_\ell$ is a $1\times \ddm$ row vector of $2\pi\dZ$-periodic trigonometric polynomials with real coefficients for $\ell=1,\ldots,s_{\mu,\mu}$. Define
\be \label{bmu}
\wh{b_{\mu,\mu,\ell}}(\xi):=
\wh{\nabla^\mu\td}(\xi)\fu_\ell(\dm^\tp \xi)F(\xi) e_1,\qquad \ell=1,\ldots,s_{\mu,\mu}.
\ee
\end{enumerate}
Define
\[
\{(b_1,\eps_1),\ldots,(b_s,\eps_s)\}:=
\{(b_{\mu,\nu,\ell},\eps_{\mu,\nu,\ell}) \setsp \mu,\nu\in \dNN, |\mu|=|\nu|=m, \mu\le \nu, \ell=1,\ldots,s_{\mu,\nu}\}.
\]
Then $\{a;b_1,\ldots,b_s\}_{(\eps_1,\ldots,\eps_s)}$ is a quasi-tight $\dm$-framelet filter bank such that all the high-pass filters have at least order $m$ vanishing moments, i.e., $\vmo(b_\ell)\ge m$ for all $\ell=1,\ldots,s$.
\end{theorem}

\begin{proof}
Let $\{\omega_1,\ldots,\omega_{\ddm}\}:=\Omega_{\dm}$ with $\omega_1:=0$.
For $b\in (\dlp{0})^r$, for simplicity of presentation, we define $r\times \ddm$ matrices
\[
G_b(\xi):=[\wh{b^{[\gamma_1]}}(\xi),\ldots,
\wh{b^{[\gamma_{\ddm}]}}(\xi)]
\quad\mbox{and}\quad
H_b(\xi):=[\wh{b}(\xi+2\pi \omega_1),\ldots,\wh{b}(\xi+2\pi \omega_{\ddm})].
\]
By \eqref{coset:U}, we have $H_b(\xi)=G_b(\dm^\tp\xi) F(\xi)$.

Define $D_\mu(\xi):=D_{\nabla^\mu \td,0}(\xi)$ as in \eqref{Dbeta}.
For $\mu<\nu$ in (S2), by the definition of $b_{\mu,\nu,\ell}$ in \eqref{bmunu} and the identity in \eqref{coset:U}, we have
\[
H_{[b_{\mu,\nu,1},\ldots,b_{\mu,\nu,\ddm}]^\tp}(\xi)
=
A_{\mu,\nu,1}(\dm^\tp \xi)F(\xi) D_\mu(\xi)+
A_{\mu,\nu,2}(\dm^\tp \xi)F(\xi) D_\nu(\xi).
\]
Since $F(\xi)F^\star(\xi)=\ddm I_{\ddm}$, we deduce from the identity $H_b(\xi)=G_b(\dm^\tp\xi) F(\xi)$
and \eqref{bMbeta} that
\[
G_{[b_{\mu,\nu,1},\ldots,b_{\mu,\nu,\ddm}]^\tp}(\dm^\tp\xi)
=\ddm^{-1} H_{[b_{\mu,\nu,1},\ldots,b_{\mu,\nu,\ddm}]^\tp}(\xi)F^\star(\xi)
=A_{\mu,\nu,1}(\dm^\tp\xi) E_\mu(\dm^\tp \xi)
+A_{\mu,\nu,2}(\dm^\tp\xi) E_\nu(\dm^\tp \xi).
\]
That is, we proved
\[
G_{[b_{\mu,\nu,1},\ldots,b_{\mu,\nu,\ddm}]^\tp}(\xi)
=A_{\mu,\nu,1}(\xi) E_\mu(\xi)
+A_{\mu,\nu,2}(\xi) E_\nu(\xi).
\]
Therefore, by $\eps_{\mu,\nu,\ell}=1$ for all $\ell=1,\ldots,\ddm$, we deduce from the above identity that
\begin{align*}
\sum_{\ell=1}^{\ddm} &\eps_{\mu,\nu,\ell} G^\star_{b_{\mu,\nu,\ell}}(\xi)
G_{b_{\mu,\nu,\ell}}(\xi)
=G^\star_{[b_{\mu,\nu,1},\ldots,b_{\mu,\nu,\ddm}]^\tp}(\xi)
G_{[b_{\mu,\nu,1},\ldots,b_{\mu,\nu,\ddm}]^\tp}(\xi)\\
&=\Big(E^\star_\mu A^\star_{\mu,\nu,1}(\xi)+
E^\star_{\nu}(\xi) A^\star_{\mu,\nu,2}(\xi)\Big)
\Big(A_{\mu,\nu,1}(\xi)E_\mu(\xi)+
A_{\mu,\nu,2}(\xi)E_\nu(\xi)\Big)\\
&=E^\star_\mu(\xi)A^\star_{\mu,\nu,1}(\xi)
A_{\mu,\nu,1}(\xi) E_\mu(\xi)+
E^\star_\nu(\xi)A^\star_{\mu,\nu,2}(\xi)
A_{\mu,\nu,2}(\xi) E_\nu(\xi)\\
&\qquad+
\Big( E^\star_\mu(\xi)A^\star_{\mu,\nu,1}(\xi)
A_{\mu,\nu,2}(\xi) E_\nu(\xi)
+
E^\star_\nu(\xi)A^\star_{\mu,\nu,2}(\xi)
A_{\mu,\nu,1}(\xi) E_\mu(\xi)\Big).
\end{align*}
As we shall see below,
the first two terms in the last expression of the last identity have been handled by the updated $A_{\mu,\mu}$ and $A_{\nu,\nu}$ in (S2) (see proof below).

For $\mu\in \dNN$ with $|\mu|=m$, we have $H_{b_{\mu,\mu,\ell}}(\xi)=\fu_\ell(\dm^\tp\xi)F(\xi) D_\mu(\xi)$. Therefore, $G_{b_{\mu,\mu,\ell}}(\xi)=
\fu_\ell(\xi)E_\mu(\xi)$ for all $\ell=1,\ldots,s_{\mu,\mu}$. Hence,
\[
\sum_{\ell=1}^{s_{\mu,\mu}}
\eps_{\mu,\mu,\ell} G^\star_{b_{\mu,\mu,\ell}}(\xi) G_{b_{\mu,\mu,\ell}}(\xi)
=E_\mu^\star(\xi)
\sum_{\ell=1}^{s_{\mu,\mu}}
\eps_{\mu,\mu,\ell} \fu_\ell^\star(\xi) \fu_\ell(\xi) E_\mu(\xi)
=E_\mu^\star(\xi)A_{\mu,\mu}(\xi) E_\mu(\xi),
\]
where $A_{\mu,\mu}$ is the updated version in (S2). Therefore, we proved
\begin{align*}
\sum_{\ell=1}^s \eps_\ell G_{b_\ell}^\star(\xi) G_{b_\ell}(\xi)
&=
\sum_{|\mu|=m} E^\star_\mu(\xi)A_{\mu,\mu}(\xi) E_\mu(\xi)+\sum_{\mu<\nu, |\mu|=|\nu|=m}
 (E^\star_\mu(\xi)A_{\mu,\nu}(\xi) E_\nu(\xi)+E^\star_\nu(\xi)A^\star_{\mu,\nu}(\xi) E_\mu(\xi))\\
&=\ddm^{-1} \cN_a(\xi).
\end{align*}
Hence, we verified the condition in \eqref{qtffb:2} and consequently, $\{a;b_1,\ldots,b_s\}_{(\eps_1,\ldots,\eps_s)}$ is a quasi-tight $\dm$-framelet filter bank.

Since $\wh{\nabla^\mu\td}(\xi)=\bo(\|\xi\|^{|\mu|})$ as $\xi \to 0$ for all $\mu\in \dNN$, it follows directly from \eqref{bmunu} and \eqref{bmu} that $\wh{b_\ell}(\xi)=\bo(\|\xi\|^m)$ as $\xi \to 0$ for all $\ell=1,\ldots,s$. Hence, all the high-pass filters have at least order $m$ vanishing moments.

To complete the proof, we now prove the existence of a desired solution to the linear system $X$ induced by \eqref{Na:factorize} and \eqref{Amu}. We first prove that $X$ must have a solution (probably with complex coefficients) and then we prove that $X$ must have a solution with real coefficients.
Define
\begin{equation} \label{aj:def}
\wh{a_1}(\xi):=1-|\wh{a}(\xi)|^2 \quad \mbox{and}\quad
\wh{a_j}(\xi):=-\ol{\wh{a}(\xi)}\wh{a}(\xi+2\pi \omega_j),\qquad j=2,\ldots,\ddm.
\end{equation}
By $\omega_1=0$ and the definition of the matrices $D_{b,\beta}$ in
\eqref{Dbeta}, it is straightforward to observe that
\[
\cN(\xi):=I_{\ddm}-\Big[\wh{a}(\xi+2\pi \omega_1),\ldots,\wh{a}(\xi+2\pi \omega_{\ddm})\Big]^\star
\Big[\wh{a}(\xi+2\pi \omega_1),\ldots,\wh{a}(\xi+2\pi \omega_{\ddm})\Big]
=\sum_{j=1}^{\ddm} D_{a_j,\omega_j}(\xi).
\]
Since $F(\xi)F^\star(\xi)=\ddm I_{\ddm}$, we deduce from \eqref{coset:U}, \eqref{bMbeta} and the above identity that
\be \label{N:Da}
\cN_a(\dm^\tp \xi)=\ddm^{-1}F(\xi) \cN(\xi) F^\star(\xi)
=\ddm^{-1}\sum_{j=1}^{\ddm} F(\xi) D_{a_j,\omega_j}(\xi) F^\star(\xi).
\ee
Suppose that we can prove
\be \label{nabla:tobeproved}
\wh{a_j}(\xi)=\sum_{\mu,\nu\in \dNN, |\mu|=|\nu|=m}
\ol{\wh{\nabla^\mu \td}(\xi)}  \wh{\nabla^\nu\td}(\xi+2\pi \omega_j)\wh{u_{j,\mu,\nu}}(\xi),\qquad j=1,\ldots,\ddm,
\ee
for some $u_{j,\mu,\nu}\in \dlp{0}$. Then by the definition in \eqref{Dbeta}, we must have
\[
D_{a_j,\omega_j}(\xi)=\sum_{\mu,\nu\in \dNN,|\mu|=|\nu|=m}
D_{\nabla^\mu \td,0}^\star(\xi) D_{u_{j,\mu,\nu},\omega_j}(\xi)
D_{\nabla^\nu\td,0}(\xi).
\]
Consequently, by \eqref{bMbeta} and $E_\mu=E_{\nabla^\mu\td,0}$, we deduce that
\begin{align*}
F(\xi)D_{a_j,\omega_j}(\xi) F^\star(\xi)
&=\ddm^{-2}\sum_{\mu,\nu\in \dNN,|\mu|=|\nu|=m}
F(\xi) D_{\nabla^\mu\td,0}^\star(\xi) F^\star(\xi)
F(\xi) D_{u_{j,\mu,\nu},\omega_j}(\xi) F^\star(\xi) F(\xi) D_{\nabla^\nu\td,0}(\xi) F^\star(\xi)\\
&=\ddm \sum_{\mu,\nu\in \dNN, |\mu|=|\nu|=m}
E_\mu^\star(\dm^\tp \xi) E_{u_{j,\mu,\nu},\omega_j}(\dm^\tp \xi)
E_\nu(\dm^\tp \xi).
\end{align*}
Now we deduce from \eqref{N:Da} that
\[
\ddm^{-1} \cN_a(\dm^\tp\xi)=
\ddm^{-1} \sum_{j=1}^{\ddm} \sum_{\mu,\nu\in \dNN,|\mu|=|\nu|=m}
E_\mu^\star(\dm^\tp \xi) E_{u_{j,\mu,\nu},\omega_j}(\dm^\tp \xi)
E_\nu(\dm^\tp \xi).
\]
Therefore, we proved
\be \label{Na:eqn}
\ddm^{-1} \cN_a(\xi)=
\ddm^{-1} \sum_{j=1}^{\ddm} \sum_{\mu,\nu\in \dNN,|\mu|=|\nu|=m}
E_\mu^\star(\xi) E_{u_{j,\mu,\nu},\omega_j}(\xi)
E_\nu(\xi).
\ee
Note that $\cN_a^\star(\xi)=\cN_a(\xi)$.
Define
\[
A_{\mu,\nu}(\xi):=\frac{1}{2\ddm}\sum_{j=1}^{\ddm} \Big(E_{u_{j,\mu,\nu}}(\xi)+E^\star_{u_{j,\nu,\mu}}(\xi)\Big),\qquad \mu,\nu\in \dNN, |\mu|=|\nu|=m, \mu \le \nu.
\]
From \eqref{Na:eqn}, it is trivial to verify that these $A_{\mu,\nu}$ satisfy both \eqref{Na:factorize} and \eqref{Amu}. That is, we proved that the linear system $X$ induced by \eqref{Na:factorize} and \eqref{Amu} must have a solution (but probably with complex coefficients).

For a $2\pi\dZ$-periodic trigonometric polynomial $\wh{u}$, it is straightforward to see that $\wh{u}$ has real coefficients if and only if $\ol{\wh{u}(-\xi)}=\wh{u}(\xi)$. Since the low-pass filter $a$ and all the filters $\nabla^\mu \td$ have real coefficients, we observe that $\cN_a$ and $E_\mu$ have real coefficients.
Changing $\xi$ into $-\xi$ and applying complex conjugate to \eqref{Na:factorize} , it is trivial to see that \eqref{Na:factorize} still holds if we replace all $A_{\mu,\nu}(\xi)$ by $\ol{A_{\mu,\nu}(-\xi)}$, respectively. Consequently, if we replace $A_{\mu,\nu}(\xi)$ by $\frac{1}{2}(A_{\mu,\nu}(\xi)+\ol{A_{\mu,\nu}(-\xi)})$, then
\eqref{Na:factorize} and \eqref{Amu} still hold. Since all $\frac{1}{2}(A_{\mu,\nu}(\xi)+\ol{A_{\mu,\nu}(-\xi)})$ have real coefficients, we proved that the linear system $X$ induced by \eqref{Na:factorize} and \eqref{Amu} must have a solution with real coefficients.

To complete the proof, we now prove \eqref{nabla:tobeproved}. From the definition in \eqref{aj:def}, we have $ \wh{a_1}(\xi) = \bo(\|\xi\|^{2m}) $ as $ \xi \to 0 $. Note that $ \ol{\wh{\nabla^\mu \td}(\xi)} \wh{\nabla^{\nu}\td}(\xi)
= \wh{\nabla^{\mu + \nu} \td}(\xi)(-1)^{|\mu|}e^{i\mu \cdot \xi} $. Hence, we conclude from Lemma~\ref{lem:nabla} that \eqref{nabla:tobeproved} holds for $ j = 1 $. For $ j = 2, \ldots, \ddm $, the sum rule condition of $ a $ implies that $ \wh{a}(\xi \pm 2\pi \omega_j) = \bo(\|\xi\|^{m}) $ as $ \xi \to 0 $.
According to Lemma~\ref{lem:nabla}, there exist some $ u_{j,\mu}, v_{j, \mu}\in \dlp{0} $ for each $\mu\in \dNN$, $|\mu|=m $, such that
\[
\wh{a}(\xi + 2\pi \omega_j) = \sum_{\mu\in \dNN, |\mu|=m}\wh{\nabla^{\mu}\td}(\xi) \wh{u_{j,\mu}}(\xi), \qquad
\wh{a}(\xi - 2\pi \omega_j) = \sum_{\nu\in \dNN, |\nu|=m}\wh{\nabla^{\nu}\td}(\xi) \wh{v_{j,\nu}}(\xi)
\]
hold. The above second identity implies that
$\wh{a}(\xi)= \sum_{\nu\in \dNN, |\nu|=m}\wh{\nabla^{\nu}\td}(\xi + 2\pi\omega_j) \wh{v_{j,\nu}}(\xi + 2\pi\omega_j) $.
Therefore,
\begin{align*}
\ol{\wh{a}(\xi)}\wh{a}(\xi+2\pi \omega_j)
=&  \sum_{\mu, \nu\in \dNN, \, |\mu|=|\nu|=m}
\ol{\wh{\nabla^{\nu}\td}(\xi + 2\pi\omega_j)} \ol{\wh{v_{j,\nu}}(\xi + 2\pi\omega_j)}
\wh{\nabla^{\mu}\td}(\xi) \wh{u_{j,\mu}}(\xi) \\
=& \sum_{\mu, \nu\in \dNN, \, |\mu|=|\nu|=m}
\ol{\wh{\nabla^{\mu}\td}(\xi)}
\wh{\nabla^{\nu}\td}(\xi + 2\pi\omega_j)
e^{-i\mu \cdot \xi}
e^{i\nu \cdot (\xi+2\pi\omega_j)}
\wh{u_{j,\mu}}(\xi)
\ol{\wh{v_{j,\nu}}(\xi + 2\pi\omega_j)}.
\end{align*}
Define
$ \wh{u_{j,\mu,\nu}}(\xi):=
- e^{-i\mu \cdot \xi}
e^{i\nu \cdot (\xi+2\pi\omega_j)}
\wh{u_{j,\mu}}(\xi)
\ol{\wh{v_{j,\nu}}(\xi + 2\pi\omega_j)}$, we proved
\eqref{nabla:tobeproved} for $ j = 2,\ldots,\ddm $.
\end{proof}

As a special case of Theorem~\ref{thm:qtf:vm:linear}, we have the following result.

\begin{cor}\label{cor:qtf:vm}
Let $\dm$ be a $d\times d$ dilation matrix and let $a\in \dlp{0}$ be a finitely supported real-valued sequence on $\dZ$.
Let $m\in \NN$ such that $m\le \min(\sr(a,\dm),\frac{1}{2}\vmo(u_a))$, where $\wh{u_a}(\xi):=1-|\wh{a}(\xi)|^2$.
Define $F(\xi):=(e^{-i \gamma_j\cdot(\xi+2\pi \omega)})_{1\le j\le \ddm,\omega\in \Omega_{\dm}}$. Define $\cN_a$ as in \eqref{Na} and $E_{\mu}(\xi):=E_{\nabla^\mu\td,0}(\xi)$ as in \eqref{bMbeta}.
If there exist $A_{\mu,\mu}$ with real coefficients for $|\mu|=m$ satisfying
\be \label{Na:factorize:Amu}
\ddm^{-1} \cN_a(\xi)=
 \sum_{|\mu|=m} E^\star_\mu(\xi)A_{\mu,\mu}(\xi) E_\mu(\xi)
 \quad \mbox{and}\quad
 A_{\mu,\mu}^\star(\xi)=A_{\mu,\mu}(\xi),\qquad
|\mu|=m,\mu\in \dNN,
\ee
(Such a solution to \eqref{Na:factorize:Amu} always exists in dimension one, i.e., $d=1$),
then there exist $b_1,\ldots,b_s\in \dlp{0}$ with real coefficients and $\eps_1,\ldots,\eps_s\in \{-1,1\}$ such that $\{a;b_1,\ldots,b_s\}_{(\eps_1,\ldots,\eps_s)}$ is a quasi-tight $\dm$-framelet filter bank such that all high-pass filters $b_\ell, \ell=1,\ldots,s$ have at least order $m$ vanishing moments and all the high-pass filters take the form either $\wh{b_\ell}(\xi)=c_\ell e^{-i\alpha_\ell\cdot \xi} \wh{\nabla^{\mu_\ell} \td}(\xi)$ (i.e., $b_\ell=c_\ell (\nabla^{\mu_\ell}\td)(\cdot-\alpha_\ell)$) or
$\wh{b_\ell}(\xi)=c_\ell e^{-i\alpha_\ell\cdot \xi} \wh{\nabla^{\mu_\ell}\td}(\xi)\wh{\nabla_{\beta_\ell}\td}(\xi)$
(i.e., $b_\ell=c_\ell [(\nabla^{\mu_\ell} \td)(\cdot-\alpha_\ell)-
(\nabla^{\mu_\ell} \td)(\cdot-\alpha_\ell-\beta_\ell)]$)
for some $c_\ell\in \R$, $\alpha_\ell,\beta_\ell \in \dZ$ and some $\mu_\ell\in \dNN$ with $|\mu_\ell|=m$ for all $\ell=1,\ldots,s$.
\end{cor}

\begin{proof} For the one dimensional case $d=1$, there are no terms satisfying $\mu<\nu$ and $|\mu|=|\nu|$. Consequently, \eqref{Na:factorize} becomes \eqref{Na:factorize:Amu}. Therefore, the existence of a solution to \eqref{Na:factorize:Amu} with $d=1$ is guaranteed by Theorem~\ref{thm:qtf:vm:linear}.
The claim follows directly by
applying Theorem~\ref{thm:qsos} to each $A_{\mu,\mu}$.
\end{proof}

We call the high-pass filters constructed in Corollary~\ref{cor:qtf:vm} as differencing filters since
all of them take the form $\nabla^\mu \td$ or their differences.

\section{Examples of Quasi-tight Framelets with Vanishing Moments}

In this section we shall illustrate Theorem~\ref{thm:qtf:vm:linear} and Corollary~\ref{cor:qtf:vm}
for constructing quasi-tight framelets with high vanishing moments from arbitrary refinable functions.
Let us first present a one-dimensional example to illustrate Corollary~\ref{cor:qtf:vm} for constructing quasi-tight framelets with all high-pass filters being special differencing filters. Recall that $\wh{u_a}(\xi):=1-|\wh{a}(\xi)|^2$.

\begin{example}  \label{ex:1dVM2} {\rm
Consider the interpolatory low-pass filter $\ta$ in \eqref{aI4} of Example~\ref{ex:1dVM1}.
Since $\sr(a,2)=4$ and $\vmo(u_a)=4$, according to the inequality in \eqref{qtffb:vm}, the highest order of vanishing moments that we can achieve is $2$.
Using Corollary~\ref{cor:qtf:vm} with $m=2$, we have a quasi-tight $2$-framelet filter bank
$\{\ta; \tb_1, \ldots, \tb_9\}_{(\eps_1, \ldots,\eps_9)} $, where all the high-pass filters are differencing filters given by
\begin{align*}
&\wh{b_1}(\xi)= \tfrac{\sqrt{2}}{32}(1-e^{-i\xi})^2(1-e^{3i\xi}),
	\quad
\wh{b_2}(\xi) =	 \tfrac{\sqrt{2}}{32}(1-e^{-i\xi})^2(e^{-i\xi}-e^{2i\xi}),
	\quad
\wh{b_3}(\xi) = \tfrac{1}{16}(1-e^{-i\xi})^2
(1-e^{2i\xi}),\\
&\wh{b_4}(\xi) = \tfrac{1}{32}(1-e^{-i\xi})^2
	(e^{3i\xi}-e^{-i\xi}),
	\quad
	\wh{b_5}(\xi) = \tfrac{\sqrt{3}}{4}(1-e^{-i\xi})^2,
	\quad
	\wh{b_6}(\xi) = \tfrac{\sqrt{3}}{4}(1-e^{-i\xi})^2 e^{-i\xi},\\
&\wh{b_7}(\xi) = \tfrac{\sqrt{42}}{32}(1-e^{-i\xi})^3,
	\quad
	\wh{b_8}(\xi) = \tfrac{\sqrt{42}}{32}(1-e^{-i\xi})^3 e^{i\xi},
	\quad
	\wh{b_9}(\xi) = \tfrac{\sqrt{3}}{16}(1-e^{-i\xi})^2 (e^{-i\xi}-e^{i\xi}),
\end{align*}
and $\eps_1=\cdots=\eps_6=1$ and $ \eps_7 =\eps_8=\eps_9=-1$. Note that the high-pass filters $b_5$ and $b_6$ have $2$ vanishing moments, while all other high-pass filters have $3$ vanishing moments.
Since $\sm(a,2)\approx 2.440765$, $\phi\in \Lp{2}$ and $\{\phi;\psi^1,\ldots,\psi^{9}\}_{(\eps_1,\ldots,\eps_9)}$ is a quasi-tight $2$-framelet in $L_2(\R)$, where $\phi,\psi^1,\ldots,\psi^{9}$ are defined in \eqref{phi:psi} with $\dm=2$ and $s=9$. Note that all the functions $\psi^1,\ldots,\psi^{9}$ have at least $2$ vanishing moments.
See Figure~\ref{fig:1dVM2} for the graphs of $\psi^1,\ldots,\psi^9$.
}\end{example}

\begin{figure}[ht!]
	\centering
	\begin{subfigure}[]{0.25\textwidth} 
		\includegraphics[width=\textwidth, height=0.8\textwidth]{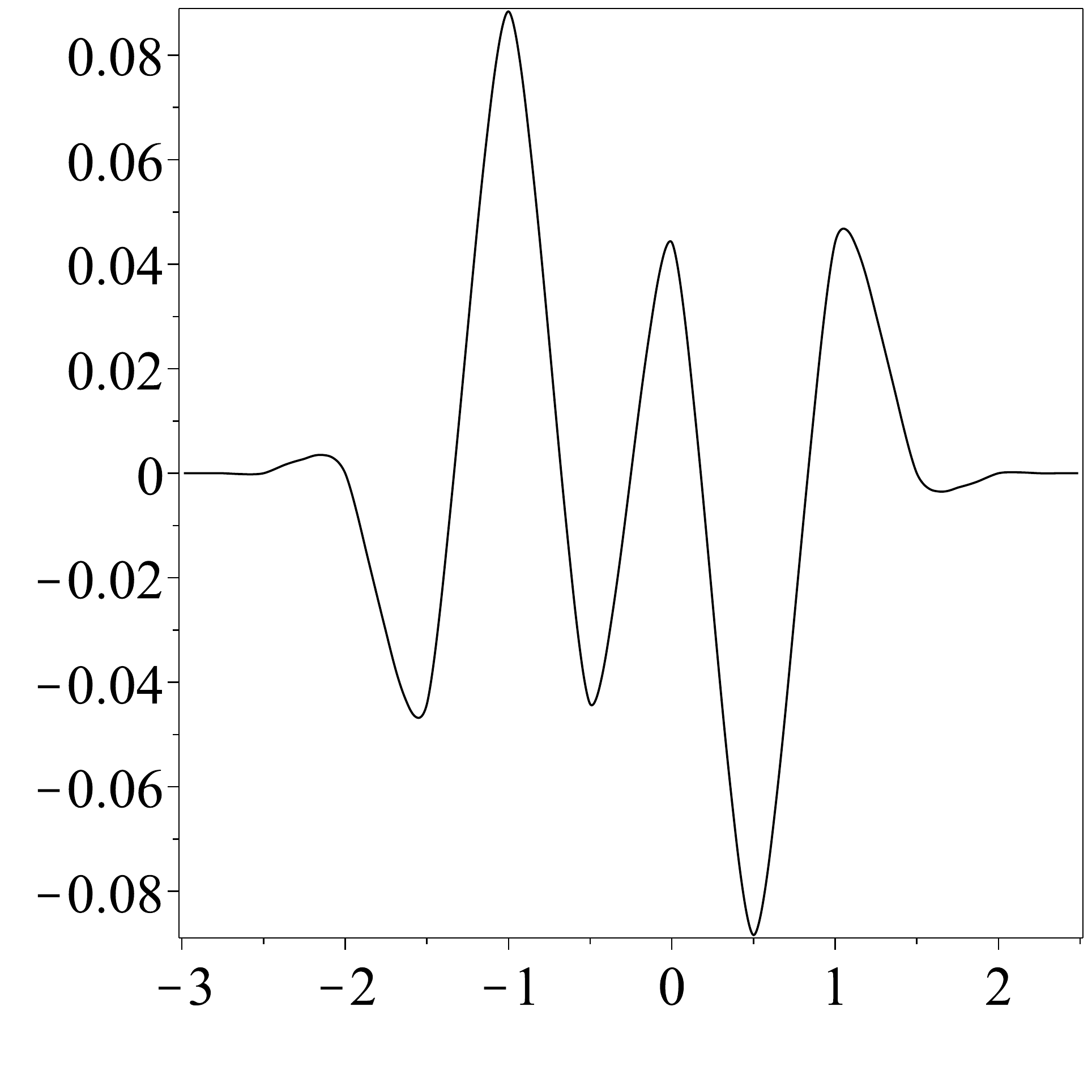}
		\caption{$\psi^1$}
	\end{subfigure}
	\begin{subfigure}[]{0.25\textwidth}
		\includegraphics[width=\textwidth, height=0.8\textwidth]{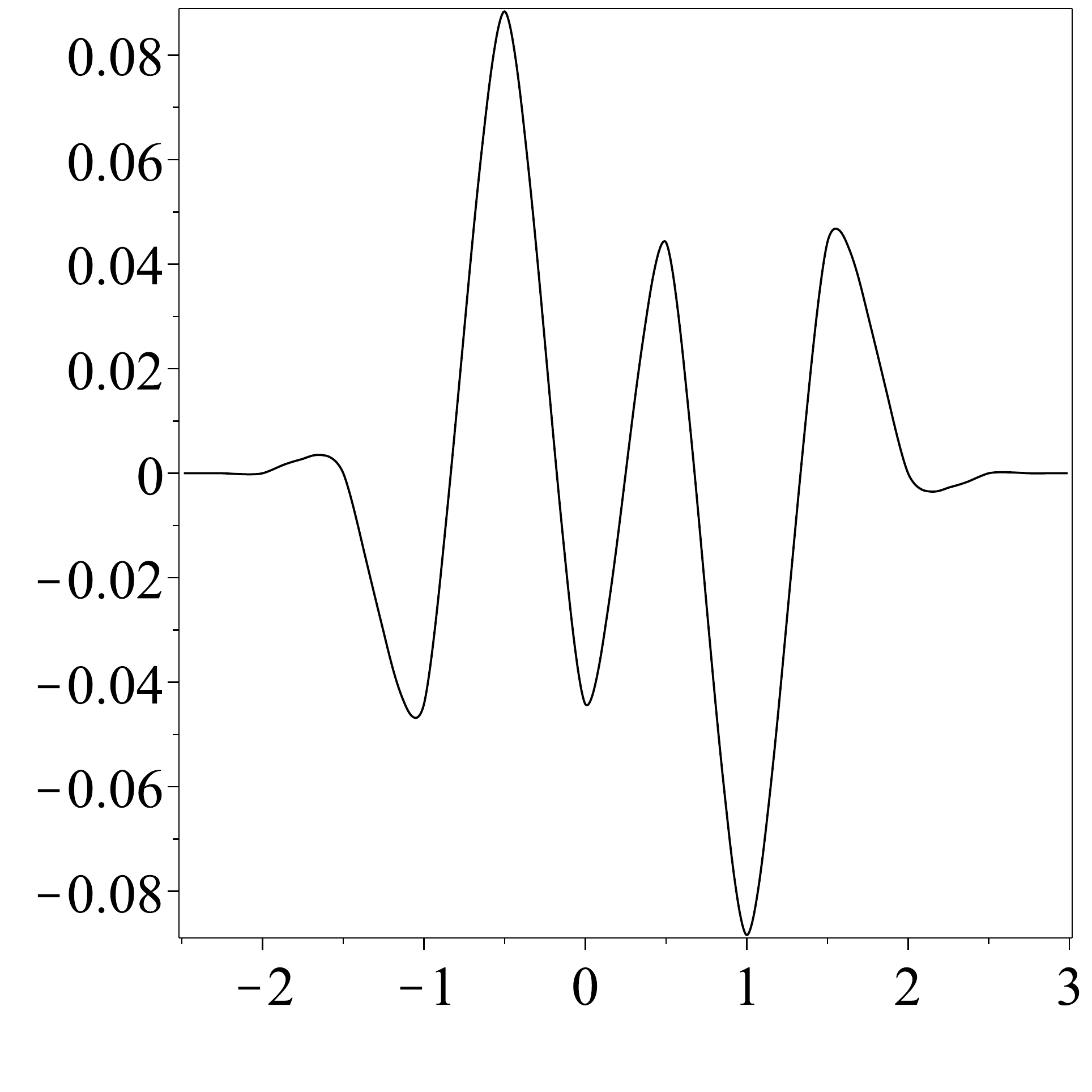}
		\caption{$\psi^2$}
	\end{subfigure}
	\begin{subfigure}[]{0.25\textwidth}
		\includegraphics[width=\textwidth, height=0.8\textwidth]{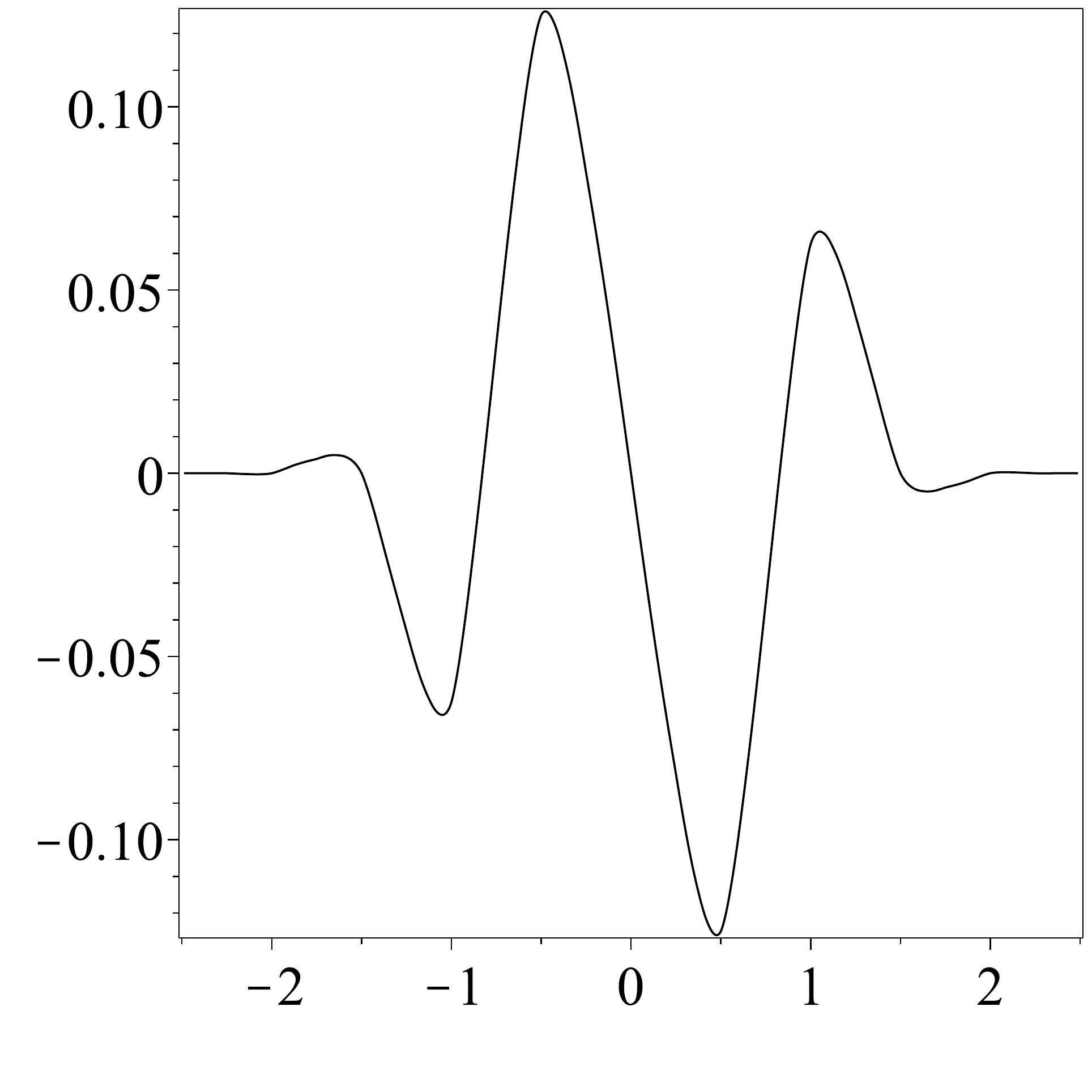}
		\caption{$\psi^3$}
	\end{subfigure}
	\\
	\begin{subfigure}[]{0.25\textwidth}
		\includegraphics[width=\textwidth, height=0.8\textwidth]{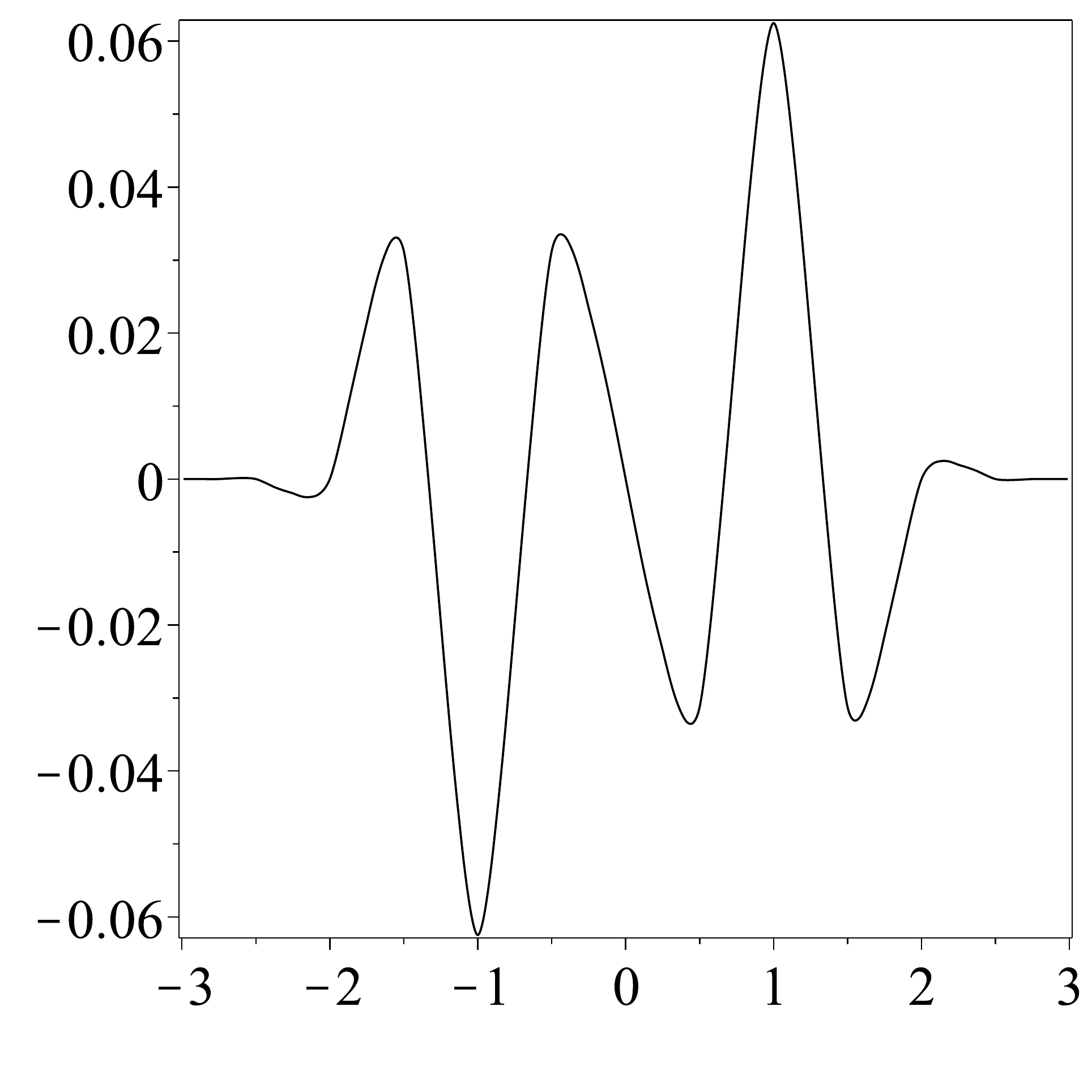}
		\caption{$\psi^4 $ }
	\end{subfigure}
	\begin{subfigure}[]{0.25\textwidth}
		\includegraphics[width=\textwidth, height=0.8\textwidth]{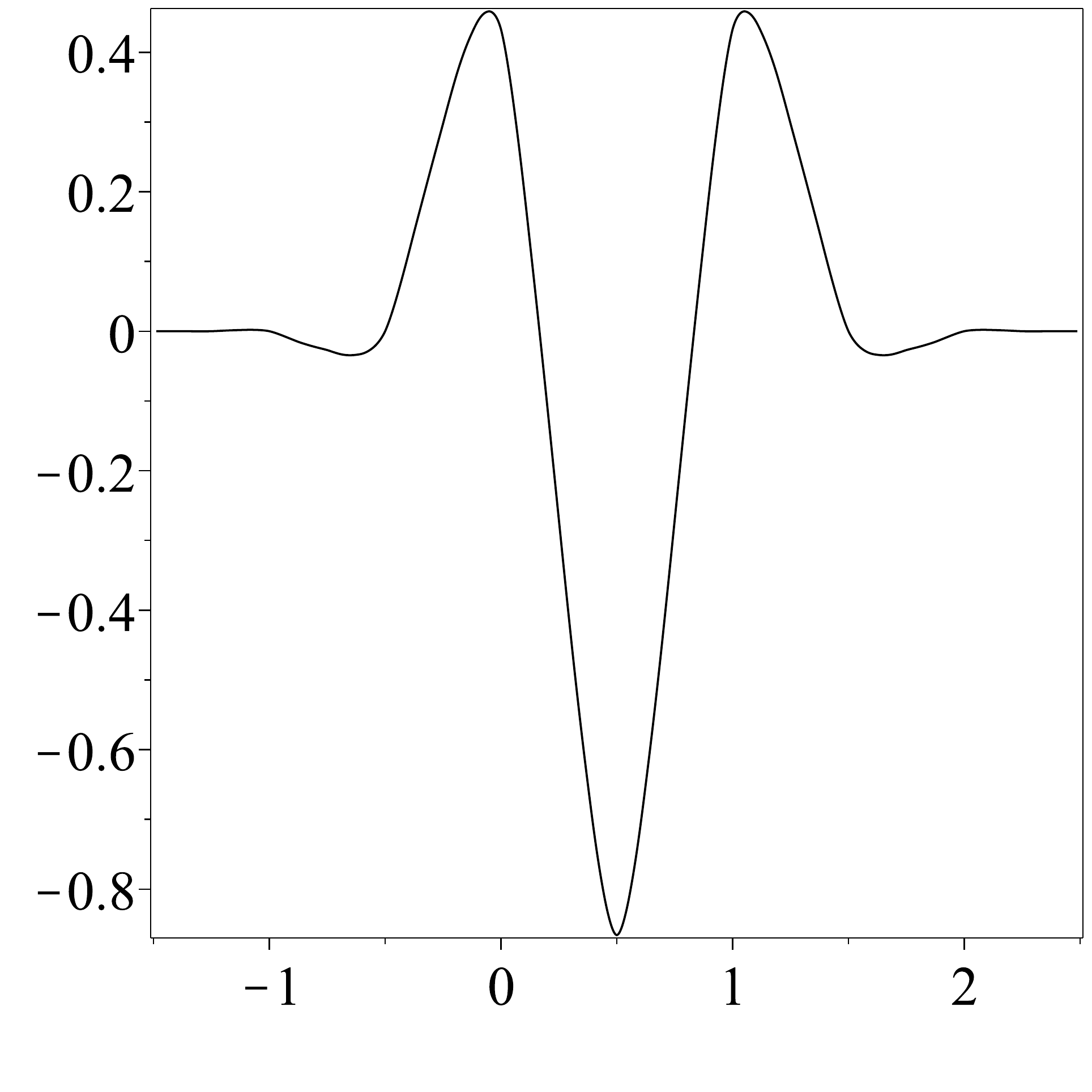}
		\caption{$\psi^5$}
	\end{subfigure}
	\begin{subfigure}[]{0.25\textwidth}
		\includegraphics[width=\textwidth, height=0.8\textwidth]{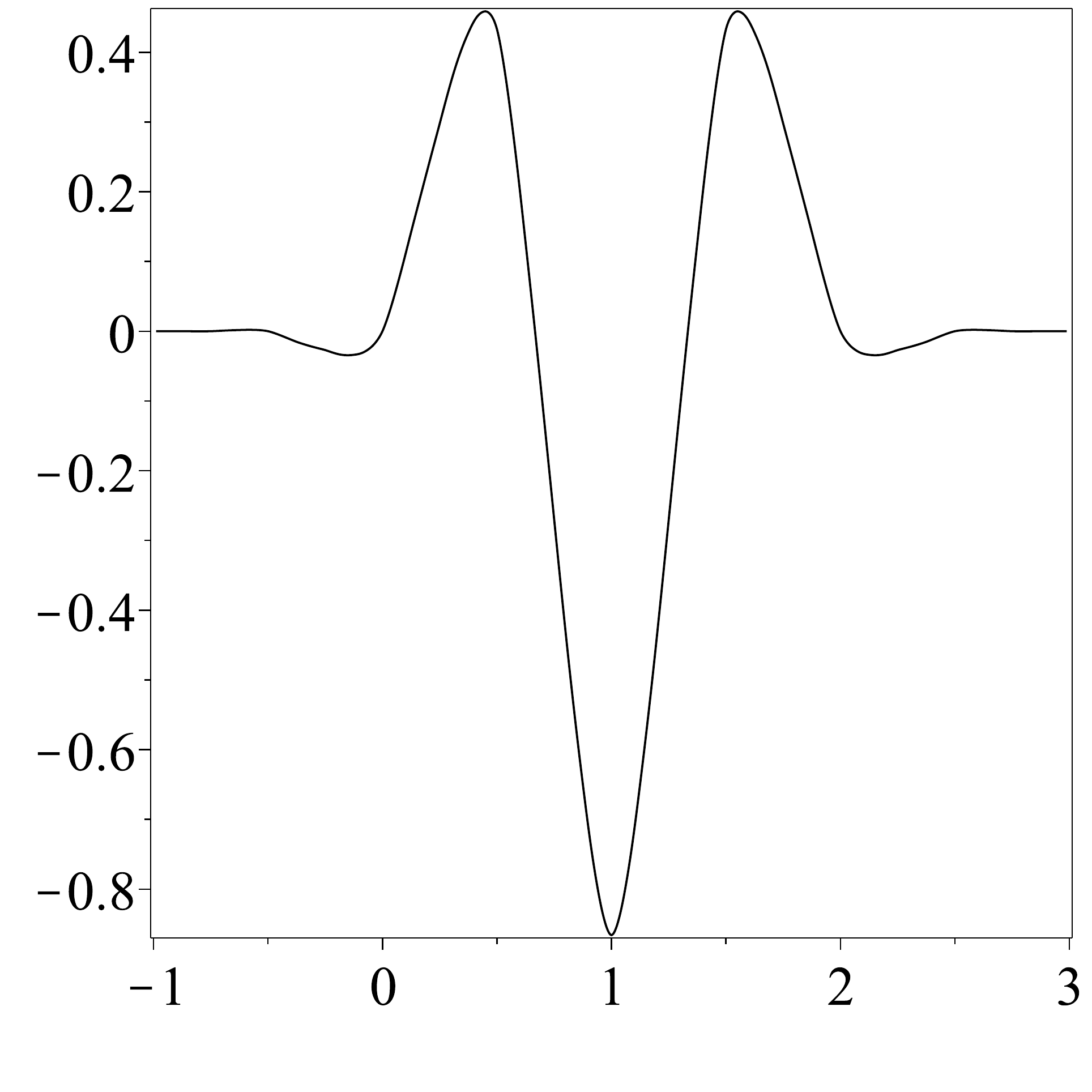}
		\caption{$\psi^6$}
	\end{subfigure}
	\\
	\begin{subfigure}[]{0.25\textwidth}
		\includegraphics[width=\textwidth, height=0.8\textwidth]{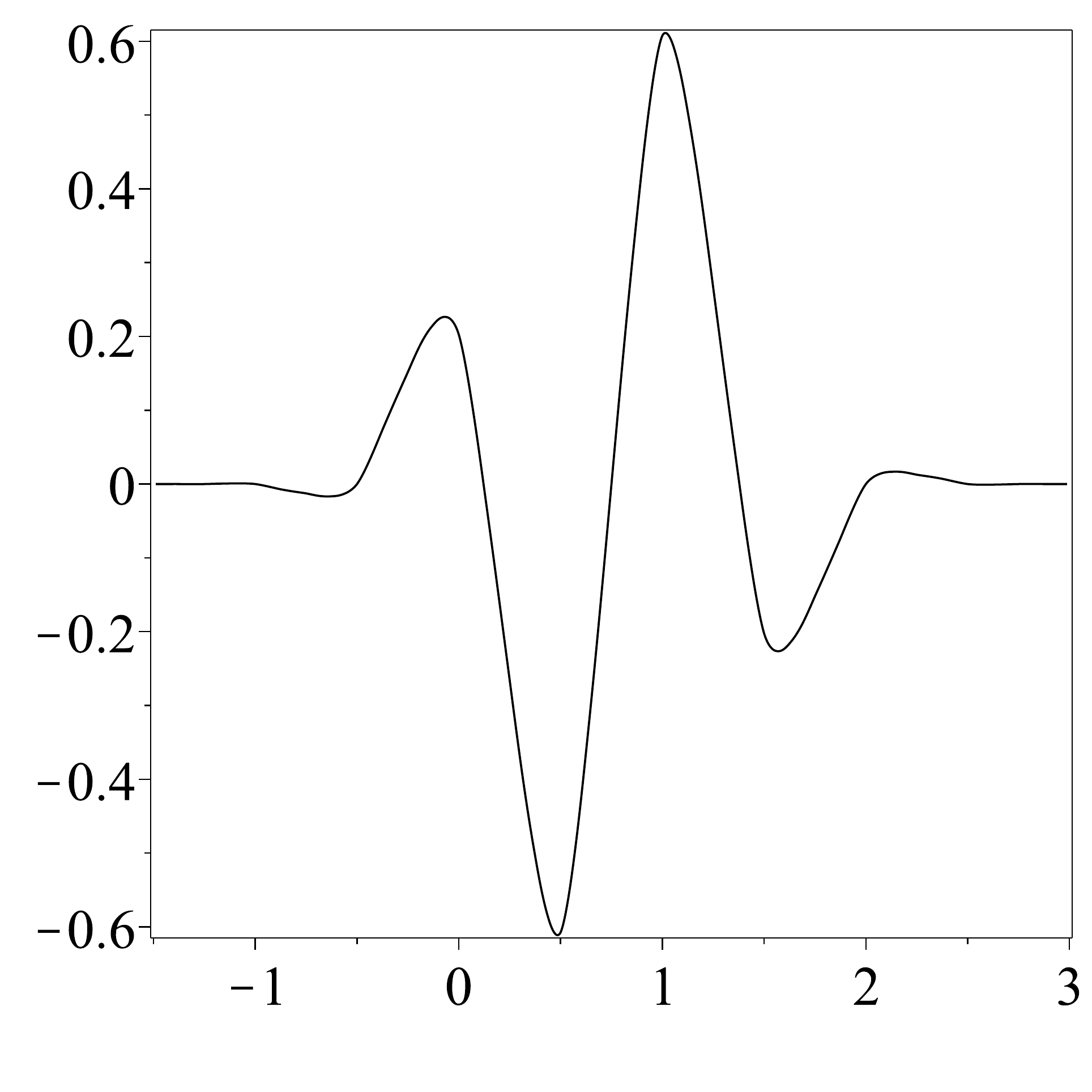}
		\caption{$\psi^7$}
	\end{subfigure}
	\begin{subfigure}[]{0.25\textwidth}
		\includegraphics[width=\textwidth, height=0.8\textwidth]{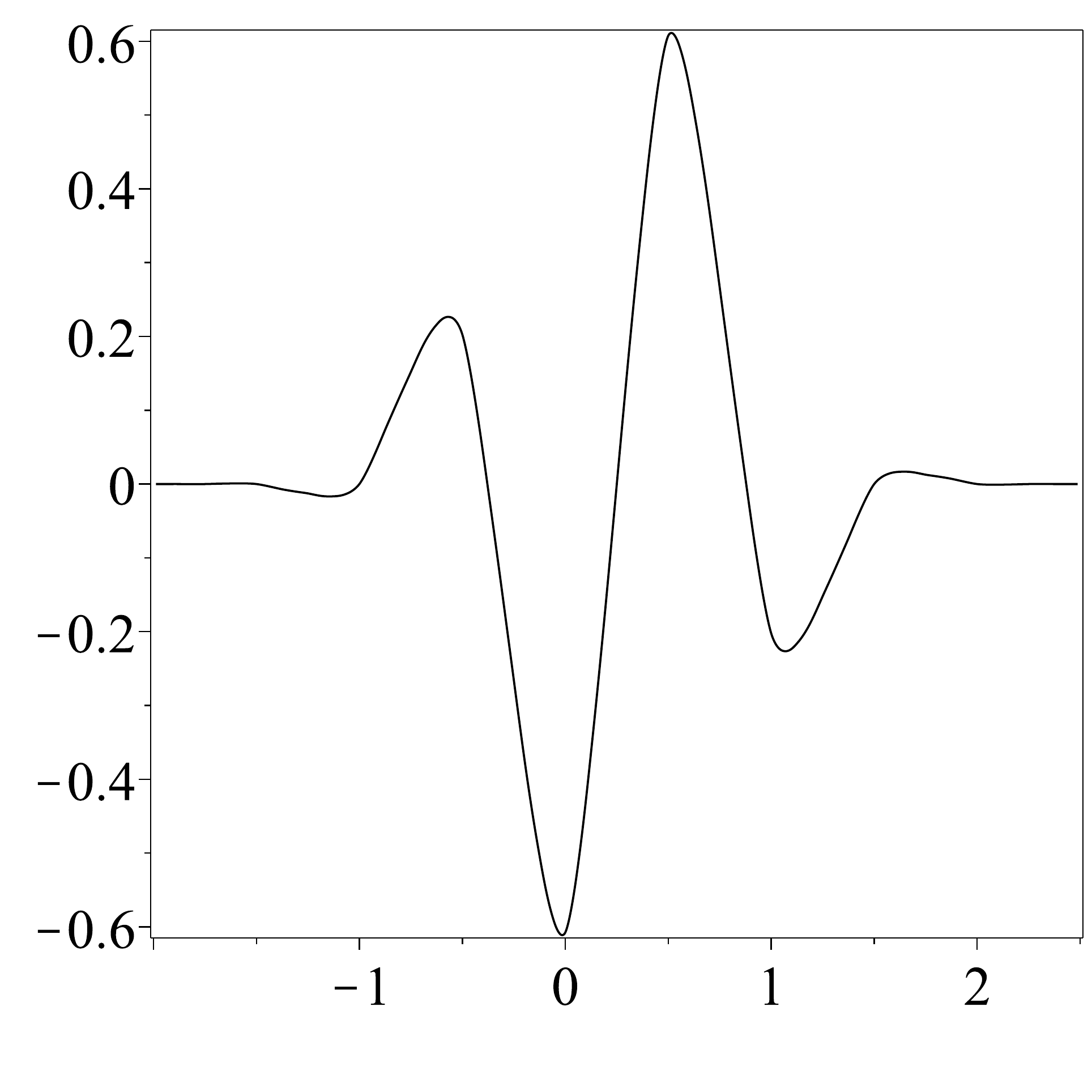}
		\caption{$\psi^8 $ }
	\end{subfigure}
	\begin{subfigure}[]{0.25\textwidth}
		\includegraphics[width=\textwidth, height=0.8\textwidth]{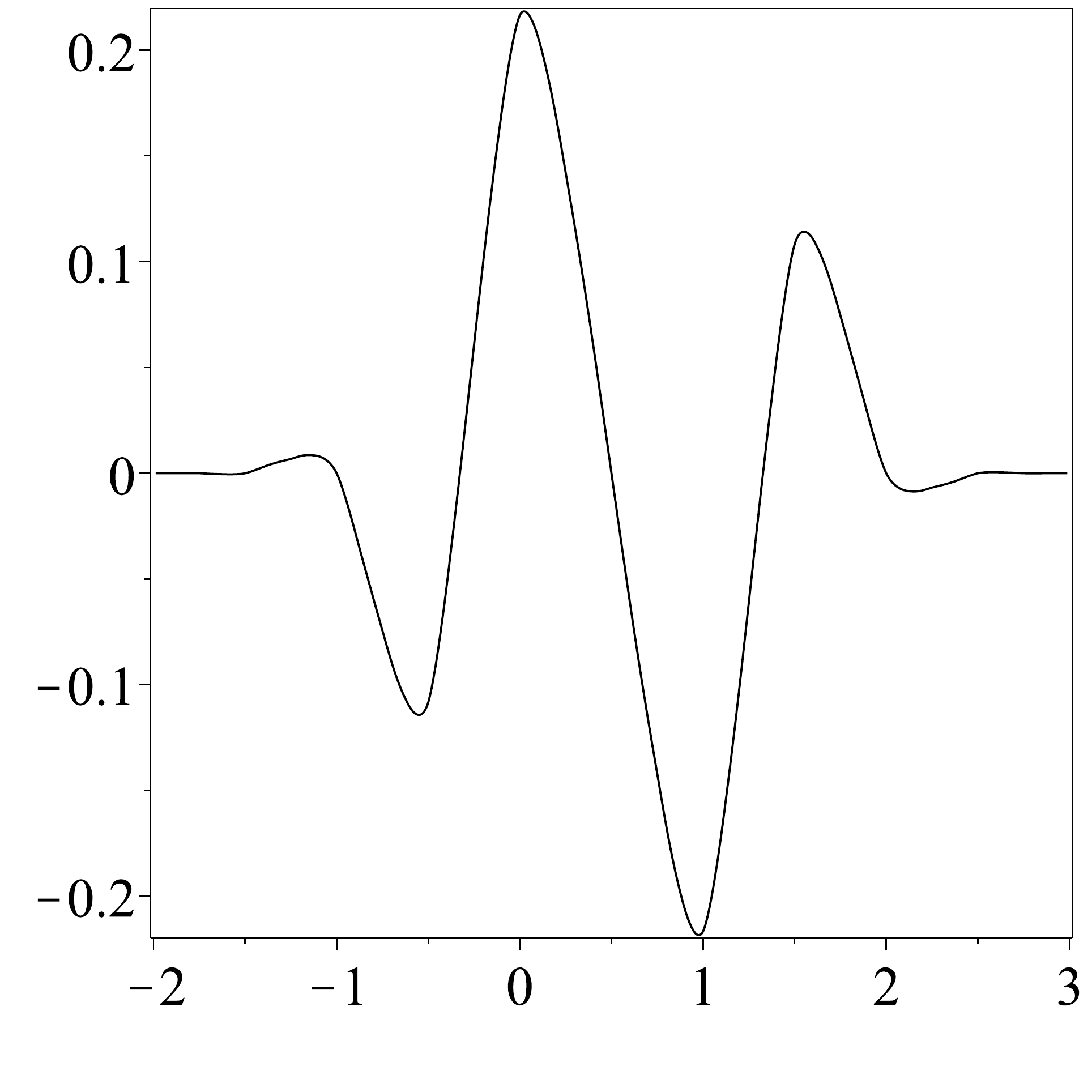}
		\caption{$\psi^9 $ }
	\end{subfigure}
\caption{
Graphs of $\psi^1,\ldots,\psi^9$
in Example~\ref{ex:1dVM2}, where
$\{\phi;\psi^1,\ldots,\psi^9\}_{(\eps_1,\ldots,\eps_9)}$ with $\eps_1=\cdots=\eps_6=1$ and $\eps_7=\eps_8=\eps_9=-1$ is a compactly supported quasi-tight $2$-framelet in $\Lp{2}$ with all their associated high-pass filters being differencing filters.	}\label{fig:1dVM2}
\end{figure}

\begin{example} \label{ex:qcx1VM2} {\rm
Consider the two-dimensional low-pass filter
\[
\ta = \begin{bmatrix}
	-\tfrac{1}{16} & \tfrac{1}{8} & -\tfrac{1}{16} \\[0.5em]
	\tfrac{1}{8} & {\bf\underline{\tfrac{3}{4}}} & \tfrac{1}{8} \\[0.5em]
	-\tfrac{1}{16}  & \tfrac{1}{8} & -\tfrac{1}{16}
	\end{bmatrix}_{[-1, 1]\times [-1, 1]}.
\]
Since $ \sr(a,\dm_{\sqrt{2}})=2$ and $\vmo(u_a)=4$, according to the inequality in \eqref{qtffb:vm}, the highest order of vanishing moments that we can achieve is $2$.
Using Theorem~\ref{thm:qtf:vm:linear} with $m=2$, we obtain a quasi-tight $\dm_{\sqrt{2}}$-framelet filter bank
$\{\ta; \tb_1, \ldots, \tb_7\}_{(\eps_1, \ldots,\eps_7)} $, where
\begin{align*}
&\wh{b_1}(\xi_1,\xi_2) = \tfrac{\sqrt{2}}{2}(1-e^{-i\xi_2})^2
	- \tfrac{\sqrt{2}}{256}(1-e^{-i\xi_1})(1-e^{-i\xi_2})
	(1-e^{-i(\xi_2-\xi_1)}),\\
&\wh{b_2}(\xi_1,\xi_2) = \tfrac{\sqrt{2}}{2}(1-e^{-i\xi_2})^2
	- \tfrac{\sqrt{2}}{512}(1-e^{-i\xi_1})^2
	 (e^{-2i\xi_2}-1-16e^{-i(\xi_2-\xi_1)}),\\
&\wh{b_3}(\xi_1,\xi_2) = \tfrac{129\sqrt{2}}{256}(1-e^{-i\xi_1})^2
	+ \tfrac{\sqrt{2}}{256}(1-e^{-i\xi_1})(1-e^{-i\xi_2})
	(e^{-i\xi_1} - 2 + e^{-i(\xi_1-2\xi_2)}),\\
&\wh{b_4}(\xi_1,\xi_2) = \tfrac{\sqrt{2}}{2}(1-e^{-i\xi_2})^2
	+ \tfrac{\sqrt{2}}{256}(1-e^{-i\xi_1})(1-e^{-i\xi_2})
	(1-e^{-i(\xi_2-\xi_1)}),\\
&\wh{b_5}(\xi_1,\xi_2) = \tfrac{\sqrt{2}}{2}(1-e^{-i\xi_2})^2
	+ \tfrac{\sqrt{2}}{512}(1-e^{-i\xi_1})^2
	 (e^{-2i\xi_2}-1-16e^{-i(\xi_2-\xi_1)}),\\
&\wh{b_6}(\xi_1,\xi_2) = \tfrac{127\sqrt{2}}{256}(1-e^{-i\xi_1})^2
	- \tfrac{\sqrt{2}}{256}(1-e^{-i\xi_1})(1-e^{-i\xi_2})
	(e^{-i\xi_1} - 2 + e^{-i(\xi_1-2\xi_2)}),\\
&\wh{b_7}(\xi_1,\xi_2) = \tfrac{1}{8}(1-e^{-i\xi_1})^2 e^{-i\xi_1},
\end{align*}
with $\eps_1=\eps_2=\eps_3=1 $, and $ \eps_4=\eps_5=\eps_6=\eps_7=-1 $. All the high-pass filters have at least $2$ vanishing moments.
Since $\sm(a,\dm_{\sqrt{2}})\approx 0.235724$, $\phi\in L_2(\R^2)$ and $\{\phi;\psi^1,\ldots,\psi^{7}\}_{(\eps_1,\ldots,\eps_7)}$ is a quasi-tight $\dm_{\sqrt{2}}$-framelet in $L_2(\R^2)$, where $\phi,\psi^1,\ldots,\psi^{7}$ are defined in \eqref{phi:psi} with $\dm=\dm_{\sqrt{2}}$ and $s=7$. Note that all the functions $\psi^1,\ldots,\psi^{7}$ have at least $2$ vanishing moments.
}\end{example}

\begin{example} \label{ex:qcx2VM2} {\rm
For $d = 2$, we consider the following two-dimensional low-pass filter
$$ \ta = \left[
	\begin {array}{ccccc} 0&0&-\frac{1}{16}&0&0\\
	\noalign{\medskip}0&\frac{1}{16}
	&\frac{1}{8}&\frac{1}{16}&0\\
	 \noalign{\medskip}-\frac{1}{16}&\frac{1}{8}&
	 {\bf\underline{\tfrac{1}{2}}}&\frac{1}{8}&-\frac{1}{16} \\
	 \noalign{\medskip}0&\frac{1}{16}&\frac{1}{8}&\frac{1}{16}&0\\
	\noalign{\medskip}0&0&-\frac{1}{16}&0
	&0\end {array} \right]_{[-2, 2]\times [-2, 2]}.
$$
Since $ \sr(a,\dm_{\sqrt{2}}) = 2 $ and $\vmo(u_a)=4$,
according to the inequality in \eqref{qtffb:vm}, the highest order of vanishing moments that we can achieve is $2$.
Using Corollary~\ref{cor:qtf:vm} with $m=2$, we obtain a quasi-tight $\dm_{\sqrt{2}}$-framelet filter bank $\{\ta; \tb_1, \ldots, \tb_{19}\}_{(\eps_1, \ldots,\eps_{19})} $, where all the high-pass filters are differencing filters given by
\begin{align*}
&\wh{b_1}(\xi_1,\xi_2) = \tfrac{1}{16}(1-e^{2i\xi_2})(1-e^{-i\xi_2})^2, \quad
	\wh{b_2}(\xi_1,\xi_2) = \tfrac{\sqrt{6}}{8} (1-e^{-i\xi_2})^2, \quad
	\wh{b_3}(\xi_1,\xi_2) = \tfrac{\sqrt{2}}{8} e^{-i\xi_1} (1-e^{-i\xi_2})^2,\\
&\wh{b_4}(\xi_1,\xi_2) = \tfrac{3}{16}
	 (1-e^{-i(\xi_2-\xi_1)})(1-e^{-i\xi_1})(1-e^{-i\xi_2}), \quad
	\wh{b_{5}}(\xi_1,\xi_2) = \tfrac{3}{16}
	 (e^{-i\xi_1}-e^{i\xi_2})(1-e^{-i\xi_1})(1-e^{-i\xi_2}),\\
&\wh{b_{6}}(\xi_1,\xi_2) = \tfrac{\sqrt{3}}{16}
	(1-e^{-2i\xi_2})(1-e^{-i\xi_1})^2, \quad
\wh{b_{7}}(\xi_1,\xi_2) = \tfrac{1}{16}
	(1-e^{2i\xi_1})(1-e^{-i\xi_1})^2,\\
&\wh{b_{8}}(\xi_1,\xi_2) = \tfrac{\sqrt{6}}{8} (1-e^{-i\xi_1})^2, \quad
	\wh{b_{9}}(\xi_1,\xi_2) = \tfrac{\sqrt{2}}{8}e^{-i\xi_1} (1-e^{-i\xi_1})^2,\\
&\wh{b_{10}}(\xi_1,\xi_2) = \tfrac{\sqrt{2}}{16}(1-e^{-i\xi_1})(1-e^{-i\xi_2})^2, \quad
	\wh{b_{11}}(\xi_1,\xi_2) = \tfrac{\sqrt{2}}{16} (1-e^{i(\xi_1+\xi_2)})(1-e^{-i\xi_2})^2,\\
&\wh{b_{12}}(\xi_1,\xi_2) = \tfrac{\sqrt{2}}{16} (1-e^{-i(\xi_2-\xi_1)})(1-e^{-i\xi_2})^2, \quad
	\wh{b_{13}}(\xi_1,\xi_2) = \tfrac{\sqrt{10}}{32}
	(e^{-i\xi_1} - e^{i\xi_1})(1-e^{-i\xi_2})^2,\\
&\wh{b_{14}}(\xi_1,\xi_2) = \tfrac{\sqrt{2}}{16}(e^{-2i\xi_1}-e^{-i\xi_1})
	(1-e^{-i\xi_2})^2, \quad
\wh{b_{15}}(\xi_1,\xi_2) = \tfrac{\sqrt{2}}{16}
(1-e^{i(\xi_1+\xi_2)})(1-e^{-i\xi_1})^2,\\
&\wh{b_{16}}(\xi_1,\xi_2) = \tfrac{1}{4}(1-e^{-i\xi_1})(1-e^{-i\xi_2}), \quad
\wh{b_{17}}(\xi_1,\xi_2) = \tfrac{1}{4}e^{-i\xi_1}(1-e^{-i\xi_1})(1-e^{-i\xi_2}),\\
&\wh{b_{18}}(\xi_1,\xi_2) = \tfrac{\sqrt{2}}{16}	 (1-e^{-i(\xi_2-\xi_1)})(1-e^{-i\xi_1})^2, \quad
\wh{b_{19}}(\xi_1,\xi_2) = \tfrac{\sqrt{26}}{32}
(e^{-i(\xi_1+2\xi_2)}-e^{-i\xi_1}) (1-e^{-i\xi_1})^2,
\end{align*}
with $\eps_1=\cdots=\eps_9= 1$ and $\eps_{10}=\cdots=\eps_{19}=-1$.
Since $\sm(a,\dm_{\sqrt{2}})\approx 1.801593$, $\phi\in L_2(\R^2)$ and $\{\phi;\psi^1,\ldots,\psi^{19}\}_{(\eps_1,\ldots,\eps_{19})}$ is a quasi-tight $\dm_{\sqrt{2}}$-framelet in $L_2(\R^2)$, where $\phi,\psi^1,\ldots,\psi^{19}$ are defined in \eqref{phi:psi} with $\dm=\dm_{\sqrt{2}}$ and $s=19$. Note that all the functions $\psi^1,\ldots,\psi^{19}$ have at least $2$ vanishing moments.

Without requiring differencing high-pass filters, using Theorem~\ref{thm:qtf:vm:linear}, we can obtain a quasi-tight $\dm_{\sqrt{2}}$-framelet filter bank $\{\ta; \tb_1, \tb_2, \tb_{3}\}_{(1, 1, -1)}$, where
\begin{align*}
&b_1 = \tfrac{\sqrt{2}}{8}\begin{bmatrix}
	0 & -1 & 0 \\
	{\bf\underline{-1}} & 4  & -1 \\
	0 & -1 & 0
	\end{bmatrix}_{[0, 2]\times [-1, 1]}, \qquad
	b_2 = \tfrac{\sqrt{6}}{8} \begin{bmatrix}
	0 & -1 & 0 \\
	{\bf\underline{1}} & 0  & 1 \\
	0 & -1 & 0
	\end{bmatrix}_{[0, 2]\times [-1, 1]},\\
&b_3 = \left[
	\begin {array}{ccccc} 0&0&-\frac{1}{16}&0&0\\
	\noalign{\medskip}0&\frac{1}{16}
	&\frac{1}{8}&\frac{1}{16}&0\\
	 \noalign{\medskip}-\frac{1}{16}&\frac{1}{8}&
	 {\bf\underline{-\tfrac{1}{2}}}&\frac{1}{8}&-\frac{1}{16} \\
	 \noalign{\medskip}0&\frac{1}{16}&\frac{1}{8}&\frac{1}{16}&0\\
	\noalign{\medskip}0&0&-\frac{1}{16}&0
	&0\end {array} \right]
	_{[-2, 2]\times [-2, 2]}
\end{align*}
with $\vmo(b_1) =\vmo(b_2) = 2$, and $\vmo(b_3) = 4$.
Since $\sm(a,\dm_{\sqrt{2}})\approx 1.801593$, $\phi\in L_2(\R^2)$ and $\{\phi;\psi^1,\psi^2,\psi^{3}\}_{(1,1,-1)}$ is a quasi-tight $\dm_{\sqrt{2}}$-framelet in $L_2(\R^2)$, where $\phi,\psi^1,\psi^2,\psi^{3}$ are defined in \eqref{phi:psi} with $\dm=\dm_{\sqrt{2}}$ and $s=3$. Note that all the functions $\psi^1,\psi^2,\psi^{3}$ have at least $2$ vanishing moments. See Figure~\ref{fig:qcx2VM2} for the graphs of $\phi,\psi^1,\psi^2,\psi^3$.
}\end{example}

\begin{figure}[ht!]
	\centering
	\begin{subfigure}[]{0.24\textwidth} 
		 \includegraphics[width=1.1\textwidth, height=0.8\textwidth]{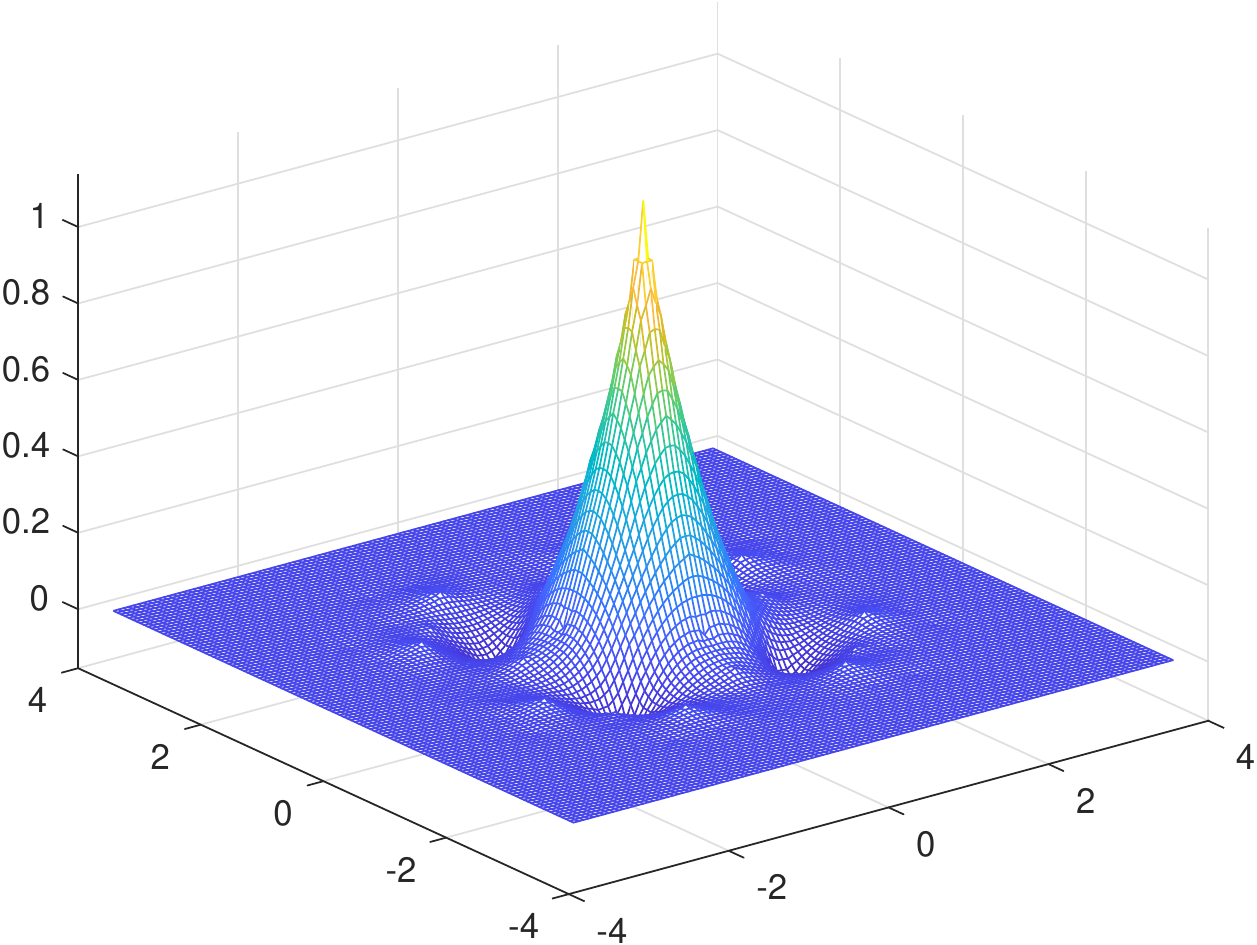}
		\caption{$\phi$}
	\end{subfigure}
	\begin{subfigure}[]{0.24\textwidth}
		 \includegraphics[width=1.1\textwidth, height=0.8\textwidth]{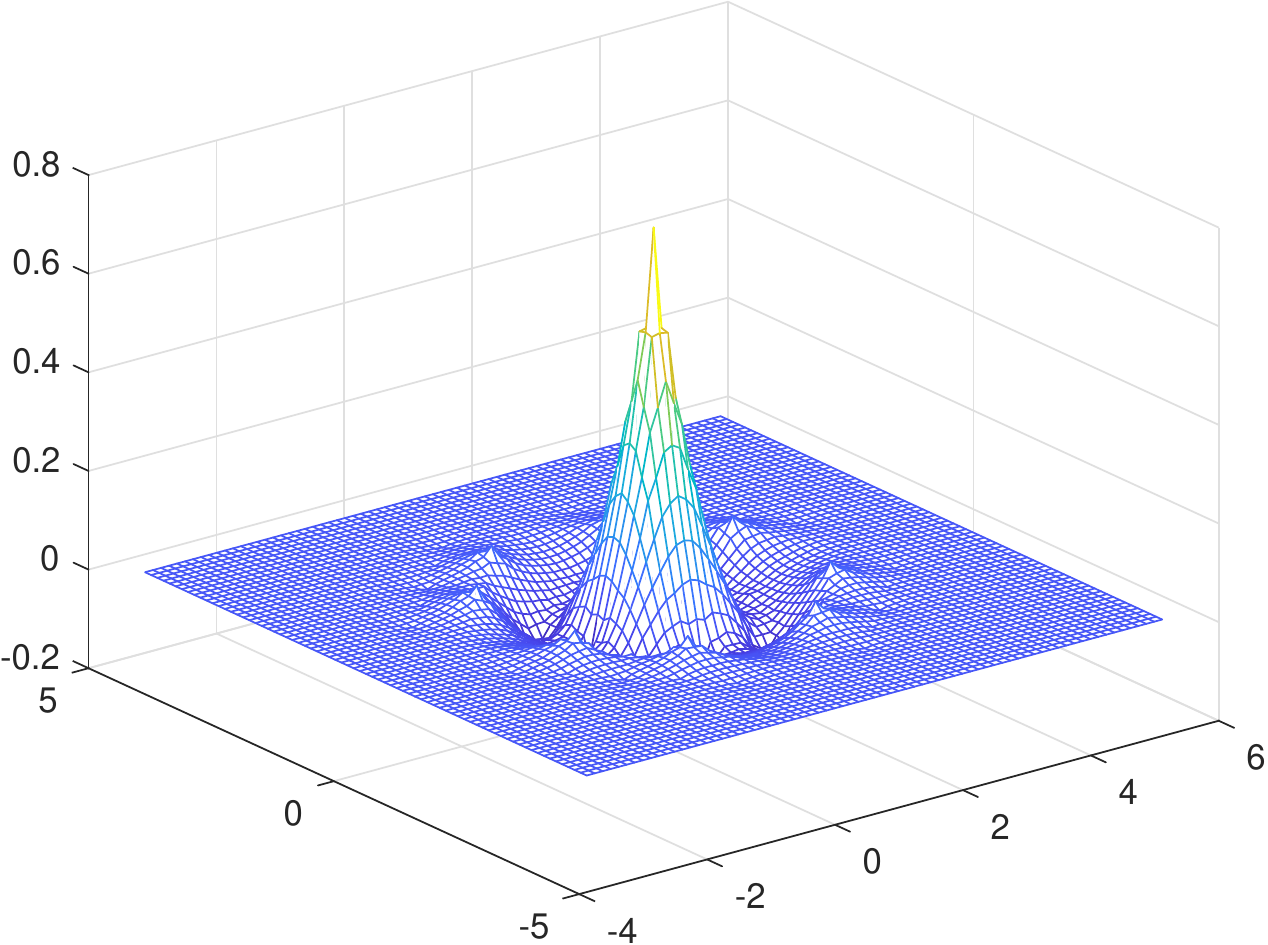}
		\caption{$\psi^1$}
	\end{subfigure}
	\begin{subfigure}[]{0.24\textwidth}
		 \includegraphics[width=1.1\textwidth, height=0.8\textwidth]{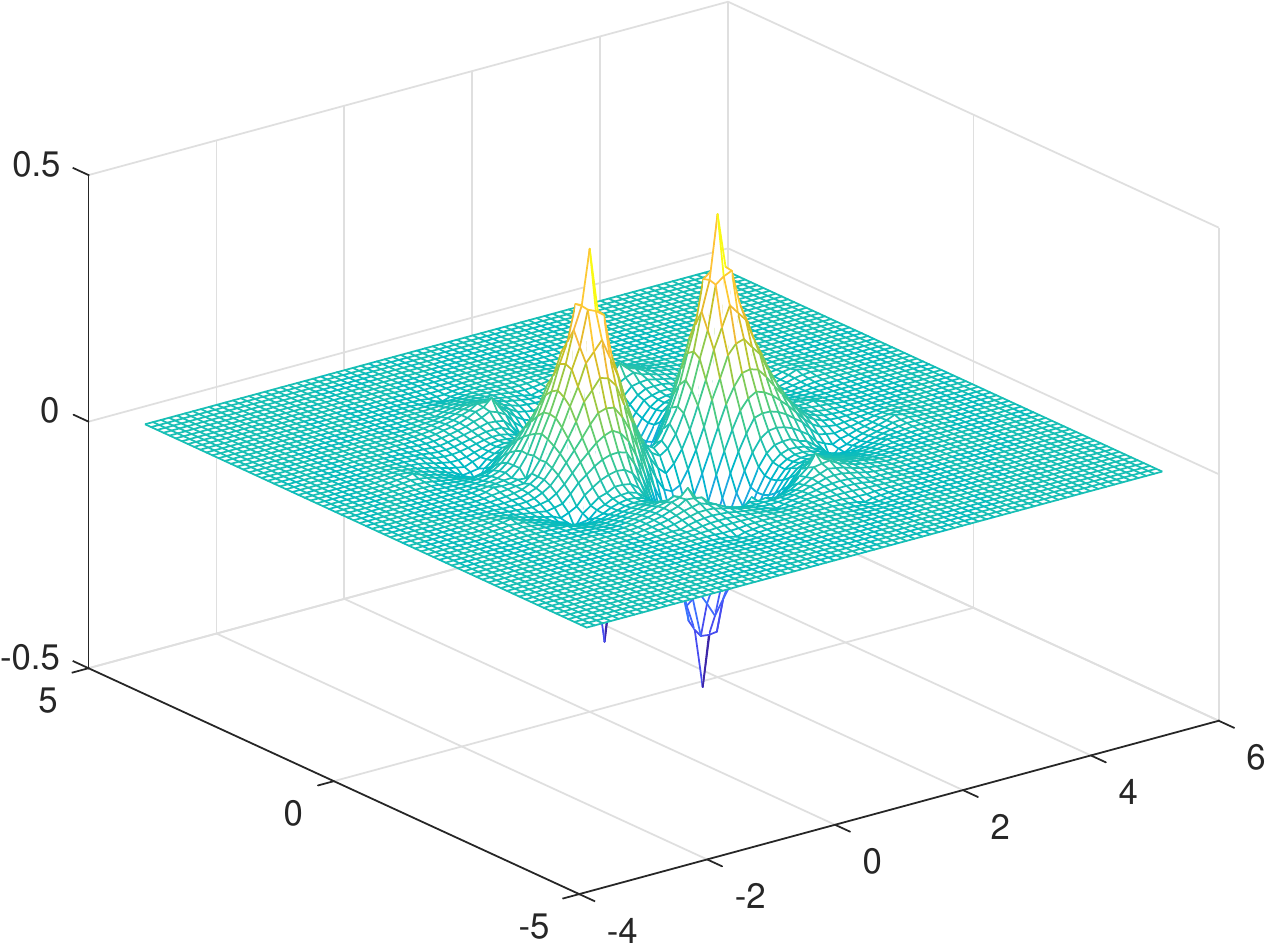}
		\caption{$\psi^2 $ }
	\end{subfigure}
	\begin{subfigure}[]{0.24\textwidth}
		 \includegraphics[width=1.1\textwidth, height=0.8\textwidth]{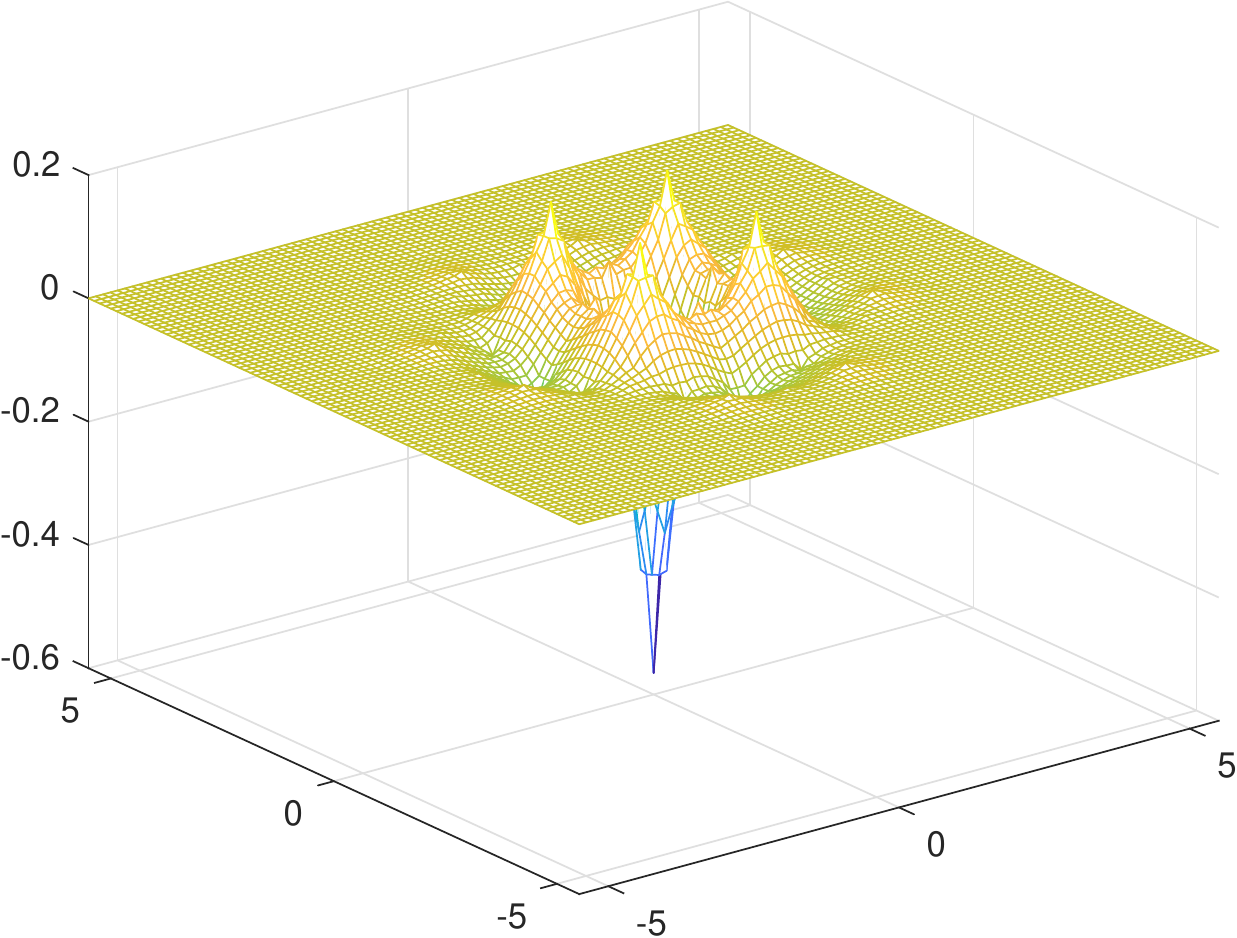}
		\caption{$\psi^3$}
	\end{subfigure}
\caption{
Graphs of $\phi,\psi^1,\psi^2,\psi^3$
in Example~\ref{ex:qcx2VM2}, where
$\{\phi;\psi^1,\psi^2,\psi^3\}_{(1,1,-1)}$ is a compactly supported quasi-tight $\dm_{\sqrt{2}}$-framelet in $L_2(\R^2)$ with at least $2$ vanishing moments.}
\label{fig:qcx2VM2}
\end{figure}

\begin{example}  \label{ex:I2VM2} {\rm
For $d=2$ and $\dm=2I_2$, we consider the following two-dimensional low-pass filter
\[
\ta =  \left[ \begin {array}{ccccc}
	 0&0&-{\tfrac{1}{64}}&-\tfrac{1}{32}&-{\tfrac{1}{64}}\\ \noalign{\medskip}0&-\tfrac{1}{32}&{\frac{5}{32}}&
	{\frac {5}{32}}&-\tfrac{1}{32}
	\\ \noalign{\medskip}-{\frac {1}{64}}&{\frac {5}{32}}& {\bf\underline{\tfrac{11}{32}}}
	&{\frac {5}{32}}&-{\frac {1}{64}}\\
	 \noalign{\medskip}-\tfrac{1}{32}&{\frac {5}
		{32}}&{\frac {5}{32}}&-\tfrac{1}{32}&0\\
	\noalign{\medskip}-{\frac {1}{64}}&
	-\tfrac{1}{32}&-{\frac {1}{64}}&0&0\end {array}
	\right]_{[-2,2]\times [-2, 2]} .
\]
Since $ \sr(a,2I_2) = 2 $ and $\vmo(u_a)=4$,
according to the inequality in \eqref{qtffb:vm}, the highest order of vanishing moments that we can achieve is $2$.
Using Theorem~\ref{thm:qtf:vm:linear} with $m=2$, we obtain a quasi-tight $2I_2$-framelet filter bank $\{\ta; \tb_1, \ldots, \tb_{5}\}_{(\eps_1, \ldots,\eps_{5})} $, where
{\small \begin{align*}
&\tb_1 =
\left[ \begin {array}{ccc} \frac{1}{16}&0&-{\frac{5}{16}}\\
	\noalign{\medskip}0&\frac{1}{2}&0\\
	\noalign{\medskip}
	 {\bf\underline{-\tfrac{5}{16}}}&0&\frac{1}{16}\end {array}
	\right]_{[0, 2]\times [0, 2]}, \quad
	\tb_2 =
	\left[ \begin {array}{ccccc} 0&0&-{\frac {5}{16}}&0&\frac{1}{16}\\
	 \noalign{\medskip}0&0&\frac{1}{2}&0&0\\
	\noalign{\medskip}\frac{1}{16}&0&
	 {\bf\underline{-\tfrac{5}{16}}}&0&0\end{array}
	\right]_{[-2, 2]\times [0, 2]}, \quad
	\tb_3 =
	\left[ \begin {array}{ccc} 0&0&\frac{1}{16}\\
	\noalign{\medskip}0&0&0\\
	\noalign{\medskip}
	 {\bf\underline{-\tfrac{5}{16}}}&\frac{1}{2}&-{\frac {5}{16}}\\
	\noalign{\medskip}0&0&0\\
	\noalign{\medskip}\frac{1}{16}&0&0\end {array}
	\right]_{[0, 2]\times [-2, 2]},\\
&\tb_4 =
	\left[ \begin {array}{ccccc} 0&0&\frac{1}{16}&0&0\\
	\noalign{\medskip}0&0&0&0&0\\
	\noalign{\medskip}0&0&
	 {\bf\underline{-\tfrac{3}{16}}}&0&\frac{1}{16}\\
	\noalign{\medskip}0&0&0&0&0\\
	 \noalign{\medskip}\frac{1}{16}&0&0&0&0\end {array}
	\right]_{[-2, 2]\times [-2, 2]}, \qquad
	\tb_5 = \left[ \begin {array}{ccccc}
	 0&0&-{\tfrac{1}{64}}&-\tfrac{1}{32}&-{\tfrac{1}{64}}\\ \noalign{\medskip}0&-\tfrac{1}{32}&{\frac{5}{32}}&
	{\frac {5}{32}}&-\tfrac{1}{32}\\
	\noalign{\medskip}-{\frac {1}{64}}&{\frac {5}{32}}& {\bf\underline{-\tfrac{21}{32}}}
	&{\frac {5}{32}}&-{\frac {1}{64}}\\
	 \noalign{\medskip}-\tfrac{1}{32}&{\frac {5}
		{32}}&{\frac {5}{32}}&-\tfrac{1}{32}&0\\
	\noalign{\medskip}-{\frac {1}{64}}&
	-\tfrac{1}{32}&-{\frac {1}{64}}&0&0\end {array}
	\right]_{[-2, 2]\times [-2, 2]}
\end{align*}
}
with $\eps_1=\eps_2 = \eps_3= 1$ and $\eps_4 = \eps_5 = -1$. Note that $\vmo(\tb_1)=\cdots=\vmo(\tb_4)=2$ and
$\vmo(\tb_5)=4$.
Since $\sm(a,2I_2)\approx 0.885296$, $\phi\in L_2(\R^2)$ and $\{\phi;\psi^1,\ldots,\psi^{5}\}_{(\eps_1,\ldots,\eps_{5})}$ is a quasi-tight $2I_2$-framelet in $L_2(\R^2)$, where $\phi,\psi^1,\ldots,\psi^{5}$ are defined in \eqref{phi:psi} with $\dm=2I_2$ and $s=5$. Note that all the functions $\psi^1,\ldots,\psi^{5}$ have at least $2$ vanishing moments.
}\end{example}

Many more examples of quasi-tight framelets can be obtained by Theorems~\ref{thm:qtf} and~\ref{thm:qtf:vm:linear}. We finish this paper by discussing quasi-tight framelets for future research.
First of all, it is important to
explore applications of the constructed quasi-tight framelets with directionality or high vanishing moments in this paper.
This in turn will provide us further insights and motivations for constructing quasi-tight framelets with properties other than directionality or vanishing moments.
Secondly, the number of high-pass filters in Theorem~\ref{thm:qtf} is often unavoidably large, due to the simple structure of all high-pass filters. It is interesting to ask how to construct suitable low-pass filters so that the constructed quasi-tight or tight framelets in Theorem~\ref{thm:qtf} have good directionality while the number of high-pass filters is reasonably small for applications.
Thirdly, the directional quasi-tight framelets constructed in Theorem~\ref{thm:qtf} have only one vanishing moment while those obtained in Theorem~\ref{thm:qtf:vm} have high vanishing moments but lack directionality. It is very natural to consider how to combine Theorems~\ref{thm:qtf} and~\ref{thm:qtf:vm} with both directionality and vanishing moments.
Lastly, due to the great flexibility in constructing quasi-tight framelets as demonstrated in this paper, it is highly possible that one may be able to combine known/famous filters (such as filters for edge detection and image processing) into a quasi-tight framelet. Such custom-designed quasi-tight framelets may lead to interesting applications in some areas.
We shall leave these tasks as future research problems while only concentrate in this paper on the theoretical investigation and construction of quasi-tight framelets.

\end{document}